\documentclass[11pt]{article}
\usepackage[margin=1in]{geometry}

\usepackage{amsmath,amsthm,amssymb,color,verbatim,algorithm2e}
\usepackage{algpseudocode}

    \usepackage[bookmarks=false]{hyperref}

\newcommand{\remove}[1]{}
\makeatletter
\newtheorem*{rep@theorem}{\rep@title} \newcommand{\newreptheorem}[2]{%
\newenvironment{rep#1}[1]{%
\def\rep@title{\bf #2 \ref{##1} }%
\begin{rep@theorem} }%
{\end{rep@theorem} } }
\makeatother
\newreptheorem{theorem}{Theorem}
\newreptheorem{lemma}{Lemma}

\newtheorem{thm}{Theorem}[section]
\newtheorem{claim}[thm]{Claim}
\newtheorem{lemma}[thm]{Lemma}
\newtheorem{define}[thm]{Definition}
\newtheorem{cor}[thm]{Corollary}

\newtheorem{remark}[thm]{Remark}

\newtheorem*{thm1}{Theorem 1}
\newtheorem*{thm2}{Theorem 2}
\newtheorem*{thm3}{Theorem 3}
\newtheorem*{thm4}{Theorem 4}

\renewcommand{\remove}[1]{}
\newcommand{\polylog}{{\rm polylog}}

\newcommand{\eps}{{\varepsilon}}
\renewcommand{\l}{\left}
\renewcommand{\r}{\right}
\newcommand{\ep}{{\epsilon}}

\newcommand{\la}{{\lambda}}
\newcommand{\cp}{\textnormal{cp}}
\newcommand{\comments}[1]{}
\newcommand{\same}{\textnormal{$same^{\star}$}}
\newcommand{\cpy}{\textnormal{copy}}
\newcommand{\nmExt}{\textnormal{nmExt}}

\renewcommand{\S}{\mathcal{S}}

\newcommand{\D}{\mathcal{D}}

\newcommand{\scirc}{\hspace{0.1cm}\circ \hspace{0.1cm}}

\newcommand{\IExt}{\textnormal{IExt}}
\newcommand{\slice}{\textnormal{Slice}}
\newcommand{\ind}{\textnormal{Ind}}
\newcommand{\samp}{\textnormal{Samp}}
\newcommand{\IP}{\textnormal{IP}}
\newcommand{\Enc}{\textnormal{Enc}}
\newcommand{\Dec}{\textnormal{Dec}}

\newcommand{\Bou}{\textnormal{Bou}}
\newcommand{\RS}{\textnormal{RS}}

\newcommand{\snmExt}{\textnormal{snmExt}}

\newcommand{\LSExt}{\textnormal{LSExt}}
\newcommand{\rank}{\textnormal{rank}}

\newcommand{\Ext}{\textnormal{Ext}}
\newcommand{\laExt}{\textnormal{laExt}}
\newcommand{\ilaExt}{\textnormal{ilaExt}}
\newcommand{\inmExt}{\textnormal{inmExt}}
\newcommand{\iExt}{\textnormal{iExt}}

\newcommand{\SRExt}{\textnormal{SRExt}}

\newcommand{\support}{\textnormal{support}}

\newcommand{\E} {\mathbb{E}}

\newcommand{\A}{\mathcal{A}}

\newcommand{\C}{\mathcal{C}}

\newcommand{\zo}{\{0, 1\}}

\newcommand{\F}{\mathbb{F}}

\makeatletter
\def\old@comma{,}
\catcode`\,=13
\def,{%
  \ifmmode%
    \old@comma\discretionary{}{}{}%
  \else%
    \old@comma%
  \fi%
}
\newcommand{\x}[1]{{}$\kern-2\mathsurround${}
  \binoppenalty10000 \relpenalty10000 #1{}$\kern-2\mathsurround${}}
  
\makeatother
\def\draft{1}   

\ifnum\draft=1 
    \def\ShowAuthNotes{1}
\else
    \def\ShowAuthNotes{0}
\fi

\ifnum\ShowAuthNotes=1
\newcommand{\authnote}[2]{{ \footnotesize \bf{\color{red}[#1's Note: {\color{blue}#2}]}}}
\else
\newcommand{\authnote}[2]{}
\fi

\makeatletter
\def\sand{%
  \end{tabular}%
  \hskip 0.5em \@plus.17fil\relax
  \begin{tabular}[t]{c}}
\makeatother

\begin{document} 
\title{\textbf{Non-Malleable Extractors and Codes, with their Many Tampered Extensions}}
\author{  Eshan Chattopadhyay\thanks{Department of Computer Science, University of Texas at Austin. Research supported in part by NSF Grant CCF-1218723. Part of this research was done while visiting Microsoft Research, India and John Hopkins University. eshanc@cs.utexas.edu.} 
 \and Vipul Goyal   \thanks{Microsoft Research, India. vipul@microsoft.com} \and Xin Li\thanks{Department of Computer Science, John Hopkins University. lixints@cs.jhu.edu}}
 \maketitle
\thispagestyle{empty}

\begin{abstract}
Randomness extractors and error correcting codes are fundamental objects in computer science. Recently, there have been several natural generalizations of these objects, in the context and study of tamper resilient cryptography. These are \emph{seeded non-malleable extractors}, introduced by Dodis and Wichs \cite{DW09}; \emph{seedless non-malleable extractors}, introduced by Cheraghchi and Guruswami \cite{CG14b}; and \emph{non-malleable codes}, introduced by Dziembowski, Pietrzak and  Wichs  \cite{DPW10}. Besides being interesting on their own, they also have important applications in cryptography. For example, seeded non-malleable extractors are closely related to privacy amplification with an active adversary, non-malleable codes are related to non-malleable secret sharing, and seedless non-malleable extractors provide a universal way to construct explicit non-malleable codes.

However, explicit constructions of non-malleable extractors appear to be hard, and the known constructions are far behind their non-tampered counterparts. Indeed, the best known seeded non-malleable extractor requires min-entropy rate at least $0.49$ \cite{Li12b}; while explicit constructions of non-malleable two-source extractors were not known even if both sources have full min-entropy, and was left as an open problem in \cite{CG14b}. In addition, current constructions of non-malleable codes in the information theoretic setting only deal with the situation where the codeword is tampered once, and may not be enough for certain applications.

In this paper we make progress towards solving the above problems. Our contributions are as follows.

\begin{itemize}
\item We construct an explicit seeded non-malleable extractor for min-entropy $k \geq \log^2 n$. This dramatically improves all previous results and gives a simpler 2-round privacy amplification protocol with optimal entropy loss, matching the best known result in \cite{Li15b}.

\item We construct the first explicit non-malleable two-source extractor for min-entropy $k \geq n-n^{\Omega(1)}$, with output size $n^{\Omega(1)}$ and error $2^{-n^{\Omega(1)}}$.

\item We motivate and initiate the study of two natural generalizations of seedless non-malleable extractors and non-malleable codes, where the sources or the codeword may be tampered many times. For this, we construct the first explicit non-malleable two-source extractor with tampering degree $t$ up to $n^{\Omega(1)}$, which works for min-entropy $k \geq n-n^{\Omega(1)}$, with output size $n^{\Omega(1)}$ and error $2^{-n^{\Omega(1)}}$. We further show that we can efficiently sample uniformly from any pre-image. By the connection in \cite{CG14b}, we also obtain the first explicit non-malleable codes with tampering degree $t$ up to $n^{\Omega(1)}$, relative rate $n^{\Omega(1)}/n$, and error $2^{-n^{\Omega(1)}}$.
\end{itemize}
\end{abstract}
\clearpage 
\setcounter{page}{1}

\section{Introduction}
Randomness extractors are fundamental objects in the study of randomness in computation. They are efficient algorithms that transform imperfect randomness into almost uniform random bits. Here we use the standard model of \emph{weak random source} to model imperfect randomness. The \emph{min-entropy} of a random variable~$X$ is defined as $ H_\infty(X)=\min_{x \in \mathsf{Supp}(X)}\log_2(1/\Pr[X=x])$. For a source $X$ supported on  $\zo^n$, we call $X$ an $(n,H_\infty(X))$-source, and we say $X$ has \emph{entropy rate} $H_\infty(X)/n$.

As one can show that no deterministic extractor works for all weak random sources even with min-entropy $k =n-1$, randomness extractors are studied in two different settings. In one setting the extractor is given a short independent uniform random seed, and these extractors are called \emph{seeded extractors}. Informally, a seeded extractor $\Ext: \{0,1 \}^n \times \{ 0,1\}^d \rightarrow \{ 0,1\}^m$ for min-entropy $k$ and error $\epsilon$ takes as input  any $(n,k)$ source $X$ and a uniform seed $S$, and has the property that $|\Ext(X,S) - U_m|<\epsilon$, where the distance used is the standard statistical distance. If the output of the extractor is guaranteed to be close  to uniform even after seeing the value of the seed $S$,  then it is called a strong seeded extractor. In the other setting there is no such random seed, but the source is assumed to have some special structure. These extractors are called \emph{seedless extractors} (see Section \ref{section:prelims} for formal definitions). A special kind of seedless extractors that received a lot of attention is extractors for independent weak random sources. Here one can use the probabilistic method to show that such extractors exist for only two independent sources (such extractors are called \emph{two-source extractors}), but the known constructions are still not optimal.

Both kinds of extractors have been studied extensively, and shown to have many connections and applications in computer science. For example, seeded extractors can be used to simulate randomized algorithms with access to only weak random sources, and are closely related to pseudorandom generators, error-correcting code and expanders. Independent source extractors can be used to generate high quality random bits for distributed computing and cryptography \cite{KalaiLRZ08}, \cite{KalaiLR09}, and are closely related to Ramsey graphs and other seedless extractors.

In cryptographic applications, however, one faces a new situation where the inputs of an extractor may be tampered by an adversary. For example, an adversary may tamper with the seed of a seeded extractor, or both sources of a two-source extractor. In this case, one natural question is how the output of  the tampered inputs will depend on the output of the initial inputs. In order to be resilient to adversarial tampering, one natural way is to require that original output of the extractor be (almost) independent of the tampered output. This leads to the notion of \emph{non-malleable extractors}, in both the seeded case and seedless case. These extractors not only are interesting in their own rights, but also have important applications in cryptography.

\begin{define}[Tampering Funtion]
For any function $f:S \rightarrow S$, $f$ has a fixed point at $s \in S$ if $f(s)=s$. We say $f$ has no fixed points in  $T \subseteq S$, if $f(t) \neq t$ for all $t \in T$. $f$ has no fixed points if $f(s) \neq s$ for all $s \in S$.
\end{define}

Seeded non-malleable extractors were introduced by Dodis and Wichs in \cite{DW09}, as a generalization of strong seeded extractors.

\begin{define}[Non-malleable extractor] A function $\snmExt:\{0,1\}^n \times \{ 0,1\}^d \rightarrow \{ 0,1\}^m$ is a seeded non-malleable extractor for min-entropy $k$ and error $\epsilon$ if the following holds : If $X$ is a  source on  $\{0,1\}^n$ with min-entropy $k$ and $\A : \{0,1\}^n \rightarrow \{0,1\}^n $ is an arbitrary tampering function with no fixed points, then
$$   |\snmExt(X,U_d) \scirc \snmExt(X,\A(U_d))  \scirc U_d- U_m \scirc  \snmExt(X,\A(U_d)) \scirc U_d | <\epsilon $$where $U_m$ is independent of $U_d$ and $X$.
\end{define}

The original motivation for seeded non-malleable extractors is to study the problem of privacy amplification with an active adversary. This is a basic problem in information theoretic cryptography, where two parties want to communicate with each other to convert their shared secret weak random source $X$ into shared secret nearly uniform random bits. However, the communication channel is watched by an adversary Eve, where we assume Eve has unlimited computational power and the two parties have local (non-shared) uniform random bits.

In the case where Eve is passive (i.e., can only see the messages but cannot change them), this problem can be solved by just applying a strong seeded extractor. However, in the case where Eve is active (i.e., can arbitrarily change, delete and reorder messages), the problem becomes much more complicated. The major goal here is to design a protocol that uses as few number of interactions as possible, and output a uniform random string $R$ that has length as close to $H_{\infty}(X)$ as possible (the difference is called \emph{entropy loss}). There has been extensive research on this problem (we give more details in Section $\ref{seeded_nm_intro}$). Along the line, a major progress was made by Dodis and Wichs  \cite{DW09}, who showed that seeded non-malleable extractors can be used to construct privacy amplification protocols with optimal round complexity and entropy loss.

This connection makes constructing non-malleable extractors a very promising approach to privacy amplification. However, all known constructions of such extractors (\cite{DLWZ11}, \cite{CRS12}, \cite{Li12a}, \cite{DY12}, \cite{Li12b}) require the entropy  of the weak source to be at least  $0.49n$. Moreover, all known constructions are essentially based on known two-source extractors, and the entropy requirement is exactly the same as the best known two-source extractor \cite{B2}. Thus the general feeling is that to construct explicit seeded non-malleable extractors for smaller entropy may be difficult, as is the situation for two-source extractors. In this work, somewhat surprisingly, we show that this is not the case. We dramatically improve all previous results and give explicit seeded non-malleable extractors that work for any min-entropy $k \ge \log^2 n$. Apart from applications to cryptography, this may of independent interest due to the connection found between seeded non-malleable extractors and two-source extractors in \cite{Li12b}.

We now discuss the seedless variant of non-malleable extractors. Cheraghchi and Guruswami \cite{CG14b}  introduced seedless non-malleable extractors as a natural generalization of seeded non-malleable extractors. Furthermore, they found an elegant connection between seedless non-malleable extractors and non-malleable codes, which are a generalization of error-correcting codes to handle a much larger class of tampering functions (rather than just bit erasure or modification). Informally, non-malleable codes are w.r.t a family of tampering functions $\mathcal{F}$,  and require that the decoding of any codeword that is tampered by a function $f \in \mathcal{F}$, is either the original message itself or something totally independent of the message (see Section $\ref{nmc_intro}$). Non-malleable codes have also been extensively studied recently (we provide more details in Section \ref{nmc_intro}), and Cheraghchi and Guruswami \cite{CG14b} showed a universal way of constructing explicit non-malleable codes by first constructing non-malleable seedless extractors.

In this paper we focus on one of the most interesting and well studied family of tampering functions, where the function tampers the original message independently in two separate parts. This is called the $2$-split-state model (see Section $\ref{nmc_intro}$ for a formal discussion). The corresponding seedless non-malleable extractor is then a generalization of two-source extractors, where both sources can be tampered. For ease of presentation,we present a  simplified definition here and we refer the reader to Section $\ref{formaldefs}$ for the formal definition.

\begin{define}[Seedless 2-Non-Malleable Extractor]\label{def:t1}
A function $\nmExt : \{ 0,1\}^{n} \times \{ 0,1\}^{n} \rightarrow \{ 0,1\}^m$ is a seedless 2-non-malleable extractor at min-entropy $k$  and error $\epsilon$ if it satisfies the following property: If $X$ and $Y$ are independent $(n,k)$-sources and   $\A=(f,g)$ is an arbitrary $2$-split-state tampering function, such that at least one of $f$ and $g$ has no fixed points, then   $$ |\nmExt(X,Y) \circ \nmExt(\A(X,Y)) - U_m \circ \nmExt(f(X),g(Y))| < \epsilon$$ where both $U_m$'s refer to the same uniform $m$-bit string.
\end{define}

Again, the connection in \cite{CG14b} makes constructing seedless $2$-non-malleable extractors a very interesting and promising approach to non-malleable codes in the $2$-split-state model. However, no explicit constructions of $2$-non-malleable extractors were known even when both sources are perfectly uniform. Indeed, finding an explicit construction of such extractors was left as an open problem in \cite{CG14b}, and none of the known constructions of seeded non-malleable extractors seem to satisfy this stronger notion. In this paper we solve this open problem and give the first explicit construction of $2$-non-malleable extractors. Furthermore we show that given any output of the extractor, we can efficiently sample uniformly from its pre-image. By the connection in \cite{CG14b} this also gives explicit non-malleable codes in the above mentioned well studied $2$-split-state model.

We note that our results about non-malleable codes in the $2$-split-state model do not improve the already nearly optimal construction in the recent work of Aggarwal et al.\ \cite{ADKO14}. However, our construction of seedless 2-non-malleable extractors is of independent interest, and provides a more direct way to construct non-malleable codes.\footnote{In \cite{ADKO14}, non-malleable  codes in the $2$-split-state model are constructed by  giving efficient  reductions from the $2$-split-state model to $t$-split-state model, and then using a known constructions of NM codes in the $t$-split-state model with almost optimal parameters \cite{CZ14}.}




Finally, as in the case of seeded non-malleable extractors \cite{CRS12}, we consider the situation where the sources can be tampered many times. For this, we introduce a natural generalization of seedless $2$-non-malleable extractors which we call seedless $(2,t)$-non-malleable extractors (i.e., the sources are tampered $t$ times). Correspondingly, in the case of non-malleable codes we also consider the situation where a codeword can be tampered many times. For this, we also introduce a natural generalization of non-malleable codes which we call \emph{one-many non-malleable codes} (see Section $\ref{nmc_intro}$). We initiate the study of these two objects in this paper and show that one-many non-malleable codes have several natural and interesting applications in cryptography.

We present a simplified definition of seedless $(2,t)$-non-malleable extractors here, and refer the reader to Section \ref{formaldefs} for the formal definition.

\begin{define}[Seedless (2,t)-Non-Malleable Extractor]\label{def:t}
A function $\nmExt : \{ 0,1\}^{n} \times \{ 0,1\}^{n} \rightarrow \{ 0,1\}^m$ is a seedless $(2,t)$-non-malleable  extractor at min-entropy $k$  and error $\epsilon$ if it satisfies the following property: If $X$ and $Y$ are independent $(n,k)$-sources and   $\A_1= (f_1,g_1),\ldots,\A_t=(f_t,g_t)$ are $t$ arbitrary $2$-split-state tampering functions, such that for each $i \in \{ 1,\ldots,t\}$ at least one of $f_i$ and $g_i$ has no fixed points,   then
\begin{align*}
 |\nmExt(X,Y), \nmExt(\A_1(X,Y)), \ldots, \nmExt(\A_t(X,Y)) - \\ U_m,  \nmExt(\A_1(X,Y)), \ldots, \nmExt(\A_t(X,Y))| < \epsilon,
 \end{align*}
  where both $U_m$'s refer to the same uniform $m$-bit string.
\end{define}

We provide explicit constructions of seedless $(2,t)$-non-malleable extractors for $t$ up to $n^{\delta}$ for a small enough constant $\delta$. Just as the connection between $2$-non-malleable extractors and regular non-malleable codes, we show that these extractors lead to explicit constructions of one-many non-malleable codes in the $2$-split-state model. We  note that as in the case of regular non-malleable codes, the construction based on $(2,t)$-non-malleable extractors may not be the only way to construct one-many non-malleable codes. However, it appears non-trivial to extend other existing constructions of non-malleable codes to satisfy this stronger notion. We discuss this in more details in Section $\ref{other_approach}$. 


We now formally define one-many non-malleable codes below.


\subsection{Non-malleable codes}\label{nmc_intro}

We introduce the notion of what we call one-many and many-many non-malleable codes, generalizing  the notion of non-malleable codes introduced by Dziembowski, Pietrzak and  Wichs  \cite{DPW10}. Since the introduction of non-malleable codes, there has been a flurry of recent work on finding explicit constructions, resulting in applications to tamper-resilient cryptography \cite{DPW10}, robust versions of secret sharing schemes \cite{ADL13}, and connections to the seemingly unrelated area of derandomization \cite{CG14b}. We discuss prior work in detail in Section $\ref{nm_prior}$. 

 We briefly motivate the notion of non-malleable codes. Traditional error-correcting codes  encode a message $m$ into a longer codeword $c$ enabling recovery of  $m$ even after  part of $c$ is corrupted. We can view this corruption as a tampering function $f$ acting on the codeword, where $f$ is from some small allowable family $\mathcal{F}$ of tampering functions.  The strict requirement of retrieving  the encoded message $m$ imposes restrictions on the kind of tampering functions that can be handled. One might hope to achieve a weaker goal of only detecting errors, possibly with high probability. However  the notion of error detection fails to work with respect to the family of constant functions since one cannot hope to detect errors against a function that always outputs some fixed codeword.

The notion of non-malleable codes is an elegant  generalization of error-detecting codes. Informally, a non-malleable code with respect to a tampering function family $\mathcal{F}$ is equipped with a randomized encoder $\Enc$ and a deterministic decoder $\Dec$ such that $\Dec(\Enc(m))=m$ and for any tampering function $f\in \mathcal{F}$ the following holds: for any message $m$,  $\Dec(f(\Enc(m)))$ is either the message $m$ or is $\epsilon$-close (in statistical distance) to a distribution $D_f$ independent of $m$.  The parameter $\epsilon$ is called the error. Thus, in some sense, either the message arrives correctly, or, the message is entirely lost and becomes gibberish. A formal definition of non-malleable codes is given below.

\newcommand{\rep}{\textnormal{replace}}

First we define the replace function $\rep:\{ 0,1\}^* \times \{ 0,1\}^* \rightarrow \{0,1\}^*$. If the second input to $\rep$ is a single value $s$, replace all occurrences  of $\same$ in the first input with $s$ and output the result. If the second input to $\rep$ is a set $(s_1, \ldots, s_n)$, replace all occurrences of $\same_i$ in the first input with $s_i$ for all $i$ and output the result.

\begin{define}[Coding schemes] Let $\Enc:\{0,1\}^k \rightarrow \{0,1\}^n$ and $\Dec:\{0,1\}^n \rightarrow \{0,1\}^k \cup \{ \perp \}$ be functions such that $\Enc$ is a randomized function (i.e.\ it has access to a private randomness) and $\Dec$ is a deterministic function. We say that $(\Enc,\Dec)$ is a coding scheme with block length~$n$ and message length $k$ if for all $s \in \{0,1\}^k $, $\Pr[\Dec(Enc(s))=s]=1$ (the probability is over the randomness in $\Enc$).
\end{define}

\begin{define}[Non-malleable codes] A coding scheme $(\Enc,\Dec)$  with block length $n$ and message length $k$ is a non-malleable code with respect to a family of tampering functions  $\mathcal{F} \subset \mathcal{F}_n$  and error~$\epsilon$ if  for every $f \in \mathcal{F}$ there exists a random variable $D_f$ on $\{ 0,1\}^k \cup \{ \same \}$ which is independent of the randomness in $\Enc$  such that for all messages $s \in \{0,1\}^k$, it holds that $$  |\Dec(f(\Enc(s))) - \rep(D_f,s)| \le \epsilon $$
\end{define}

The rate of a non-malleable code $\C$ is given by $\frac{k}{n}$. Observe that to construct non-malleable codes, it is still necessary to restrict the class of tampering functions. This follows since the tampering function could then use the function $\Dec$ to decode the message $m$,  get a message $m^{\prime}$ by flipping all the bits in $m$, and use the encoding function to pick any codeword in $\Enc(m^{\prime})$. However presumably, the class of tampering functions can now be much richer than what was possible for error correction and error detection.

\paragraph{Tampering Multiple Codewords.} Observe that the above definition envisions the adversary receiving a single codeword $\Enc(s)$ and outputting a single tampered codeword $f(\Enc(s))$. We refer to this as the ``one-one" setting. While indeed this is very basic, we argue that this does not capture scenarios where the adversary may be getting multiple codewords as input or be allowed to output multiple codewords. As an example, consider the following.

Say there is an auction where each party can submit its bid, and, the item goes to the highest bidder. An honest party, wishing to bid for value $s$, encodes its bid using NM codes and sends $\Enc(s)$. This indeed would prevent an adversary (which belongs to an appropriate class of tampering functions) from constructing his own bid by tampering $\Enc(s)$ and coming up with $\Enc(s+1)$, which would completely compromise the sanity of the auction process. However what if the adversary can submit \emph{two} bids out of which exactly one is guaranteed to be a winning bid? For example, the adversary can submit bids to $r$ and $2s - r$ (for some $r$ not known to the adversary). This is not ruled out by NM codes!

Towards that end, we introduce a stronger notion which we call \emph{one-many NM codes}. Intuitively, this guarantees the following. Consider the set of codewords output by the adversary. We require that even the joint distribution of the encoded value be independent of the value encoded in the input. A formal definition is given below:

\begin{define} [One-Many Non-malleable codes] A coding scheme $(\Enc,\Dec)$  with block length $n$ and message length $k$ is a non-malleable code with respect to a family of tampering functions  $\mathcal{F} \subset (\mathcal{F}_n)^t$  and error~$\epsilon$ if  for every $(f_1, \ldots f_t) \in \mathcal{F}$, there exists a random variable $D_{\vec{f}}$ on $(\{ 0,1\}^k \cup \{ \same \})^t$ which is independent of the randomness in $\Enc$  such that for all messages $s \in \{0,1\}^k$, it holds that $$  |(\Dec(f_1(X)), \ldots, \Dec(f_t(X))) - \rep(D_{\vec{f}},s)| \le \epsilon $$
Where $X = \Enc(s)$. We refer to $t$ as the tampering degree of the non-malleable code.
\label{1-manynmcode}
\end{define}

We argue that one-many non-malleable codes is a basic notion which is interesting to study independent of concrete applications. However later we point out some concrete applications to non-malleable secret sharing (where one wishes to store multiple secrets), and, to \emph{witness signatures}.

An expert in cryptography by now would have noticed this is analogous to the well studied notion of one-many non-malleable \emph{commitments} in the literature. Even though both notions deal with related concerns, we note non-malleable codes and non-malleable commitment are fundamentally different objects with the latter necessarily based on complexity assumptions. To start with, we prove a simple impossibility result for one-many non-malleable codes (whereas for one-many non-malleable commitments, a corresponding \emph{positive} result is known \cite{PassR08}).

\begin{lemma} One-many non-malleable codes which work for any arbitrary tampering degree and $\epsilon < 1/4$ cannot exist for a large class of tampering functions.
\end{lemma}
\begin{proof}
The class of tampering functions which we consider are the ones where each function is allowed to read any one bit $X_i$ of its choice from the input code $X$, and output a fresh encoding of $X_i$. Most natural tampering functions (including split state ones \cite{DPW10} \cite{CG14a}) considered in the literature fall into this class. Assume that the encoded value $s$ has at least 4 possibilities (length 2 bits or higher). The case of a single bit $s$ is discussed later.

Recall that $n$ is the length of the code. We set $t = n$.  Let $X = \Enc(s)$ be the input codeword where $s$ is chosen at random. We consider $n$ tampering functions where $F_i$ simply reads $X_i$ and outputs a fresh encoding $W_i = \Enc(X_i)$. Now consider $(\Dec(f_1(X)), \ldots, \Dec(f_n(X)))$. Observe that this is exactly the bits of the string $X$. If the distinguisher applies the decode procedure on $X$, it will recover $s$. Now consider any possible output $(d_1, \ldots, d_n)$ of $D_{\vec{f}}$. Now note that there cannot exist $d_i$ which is $\same$. This is because otherwise it will be replaced by $s$ (see Definition \ref{1-manynmcode}) which is at least 2 bits while $\Dec(W_i)$ is just a single bit. This in turn implies that the value $\rep(D_{\vec{f}},s)$ (from Definition \ref{1-manynmcode}) is independent of $s$ and $X$. Thus a distinguisher (given access to $s$) can easily have an advantage exceeding $\epsilon$.

For a single bit $s$, we modify our tampering functions to encode two bits: $W_i = \Enc(X_i || 0)$. Then again we can argue that neither of $d_i$ will be $\same$ since then it will be replaced by $s$ which is only one bit. This in turn again implies that $\rep(D_{\vec{f}},s)$ is independent of $s$ and $X$. This concludes the proof.

\end{proof}

We also introduce a natural generalization which we call \emph{many-many non-malleable codes}. This refers to the situation where the adversary is given multiple codewords as input.

\begin{define}[Many-Many Non-malleable codes] A coding scheme $(\Enc,\Dec)$  with block length $n$ and message length $k$ is a non-malleable code with respect to a family of tampering functions  $\mathcal{F} \subset (\mathcal{F}_n)^t$  and error~$\epsilon$ if  for every $(f_1, \ldots f_t) \in \mathcal{F}$, there exists a random variable $D_{\vec{f}}$ on $(\{ 0,1\}^k \cup \{ \same_i \}_{i \in [u]})^t$ which is independent of the randomness in $\Enc$  such that for all vector of messages $(s_1, \ldots, s_u)$, $s_i \in \{0,1\}^k$, it holds that $$  |(\Dec(f_1(\vec{X})), \ldots, \Dec(f_t(\vec{X}))) - \rep(D_{\vec{f}}, (s_1, \ldots, s_u))| \le \epsilon $$

Where $X_i = \Enc(s_i)$ and $\vec{X} = (X_1, \ldots, X_u)$
\label{many-manynmcode}
\end{define}

\begin{lemma} One-many non-malleable codes with tampering degree $t$ and error $\epsilon$ are also many-many non-malleable codes for tampering degree $t$ and error $u\epsilon$ (where $u$ is as in Definition \ref{many-manynmcode}).
\end{lemma}
\begin{proof}
This proof relies on a simple hybrid argument and the fact that all sources $X_1, \ldots, X_u$ are independent. We only provide a proof sketch here. Assume towards contradiction that there exists a one-many code with error $\epsilon$, which, under the many-many tampering adversary has error higher than $u.\epsilon$. That is, the adversary $\vec(f)$ is given as input $(X_1, \ldots, X_u)$ which are encodings of $(s_1, \dots, s_u)$ respectively. This is referred to as the hybrid 0. Now consider the following hybrid experiment. In the $i$-th hybrid experiment, the code $X_i$ is changed to be an encoding  of 0 (as opposed to be an encoding of $s_i$). We claim that in this experiment, the error changes by at most $\epsilon$. This is because otherwise we can construct a one-many tampering adversary with error higher than $\epsilon$. To construct such an adversary $\vec(f^i)$, each $f^i_j$ has ${X_k}_{k \neq i}$ hardcoded in it and takes $X_i$ as input. This would show an adversary against which one-many non-malleable codes have an error higher than $\epsilon$.

By the time we reach $(u-1)$-th hybrid experiment, the error could only have reduced by at most $(u-1)\epsilon$. However in the $(u-1)$-th hybrid experiment, the error can at most be $\epsilon$ since it corresponds to the one-many setting. Hence, the error in the hybrid 0 could have been at most $u.\epsilon$. This concludes the proof.

\end{proof}

\paragraph{Non-malleable codes in the split-state model}

An important and well studied family of tampering functions (which is also relevant to the current work) is the family of tampering functions in the $C$-split-state model, for $C \ge 2$. In this model, each tampering function $f$ is of the form $(f_1,\ldots,f_{C})$ where   $f_{i} \in \mathcal{F}_{n/C}$, and for any codeword $x = (x_1,\ldots,x_C) \in (\{ 0,1\}^{n/C})^C$ we define $(f_1,\ldots,f_C)(x_1,\ldots,x_C) = (f_1(x_1),\ldots,f_{C}(x_C))$.  Thus each $f_i$  independently  tampers a fixed partition of the codeword. Non-malleable codes in this model can also be viewed as \emph{non-malleable secret sharing}. This is because the strings $(x_1,\ldots,x_C)$ can be seen as the shares of $s$ and tampering each share individually does not allow one to ``maul" the shared secret $s$.

There has been a lot of recent work  on constructing explicit and efficient non malleable codes in the $C$-split-state model.  Since $C=1$ includes all of $\mathcal{F}_{n}$, the best one can hope for is $C = 2$. A Monte-Carlo construction of non-malleable codes in this model was given in the original paper on non-malleable codes \cite{DPW10} for $C=2$ and then improved in \cite{CG14a}. However, both of these constructions are inefficient. For $C=2$, these Monte-Carlo constructions imply existence of codes of rate close to $\frac{1}{2}$ and corresponds to the hardest case. On the other extreme, when $C=n$,  it corresponds to the case of bit tampering where each function $f_i$ acts independently on a particular bit of the codeword. By a recent line of work \cite{DKO13} \cite{ADL13} \cite{CG14b} \cite{CZ14} \cite{ADKO14}, we now have almost optimal constructions of non-malleable codes in the $C$-state-state model, for any $C\ge 2$.

We remark that all the prior work in the information theoretic setting has only considered the construction of what we call \emph{one-one NM codes in the split state model}\footnote{A notion called as continuous non-malleable codes was considered in \cite{FMNV14} where a codeword is tampered multiple times, but the experiment stops whenever an error message is encountered. The constructions provided for continuous non-malleable codes are in the computational setting. We discuss this in Section $\ref{nm_prior}$}. That is, the adversary is only given as input a single codeword and outputs a single codeword. Our new notions seek to remedy that fact.

\paragraph{Many-many non-malleable secret sharing.}  Consider the example of non-malleable secret sharing. What if there are shares of multiple secrets which the adversary can tamper with? What if the adversary is allowed to output shares of multiple secrets? For example, say there are two secret and two devices. Each device stores one share of each of the secrets. Say that an adversary is able to tamper with the data stored on each device individually (or infect each of them with a virus). Then, the current notion of one-one NM codes does not rule out a non-trivial relationship between two resulting secrets and the two original secrets we start with. It is conceivable that what we need here is a \emph{two-two} non-malleable secret sharing. Our many-many non-malleable codes directly lead to such a many-many non-malleable secret sharing. 

\paragraph{Subsequent Work.} The notion of one-many non-malleable codes has been used to construct \emph{witness signatures} \cite{GKJ15}. Very roughly, witness signatures allow any party with a witness to some NP statement, to sign a message such that anyone can verify that the message was indeed signed using a valid witness to the NP statement. On the other hand, the signatures should still be unforgeable: that is, producing a signature on a new message (even given several message, signature pairs) should be as hard as computing a witness to the NP statement itself. There is no setup or any key generation involved. Witness signatures can be seen as an analogue of the notion of witness encryption \cite{GargGSW13} for signatures. The notion of witness signatures was introduced in \cite{GKJ15} who used one-many non-malleable codes to propose a construction in the tamper proof hardware model. The fact that non-malleable codes satisfy one-many security (as opposed to just one-one) was crucial in their work.

\subsection{Summary of results}
Our first main result is an explicit construction of a $(2,t)$-seedless non-malleable extractor. We note that  prior to this work, such a construction was not known for even $t=1$ and full min-entropy.
 \begin{thm1}  There exists a constant  $\gamma>0$ such that for all $n>0$ and $t \le n^{\gamma}$,  there exists an efficient seedless $(2,t)$-NM extractor at min-entropy $n-n^{\gamma}$ with error $2^{-n^{\Omega(1)}}$ and output length $m= n^{\Omega(1)}$.
 \end{thm1}

 Next, we show that it is possible to efficiently sample almost uniformly from the pre-image of any output of this extractor. We prove this in Section \ref{efficiency}. Combining this with Theorem $\ref{connection_t}$ and a hybrid argument, we immediately have the following result.

 \begin{thm2}  There exists a constant $\gamma>0$ such that for all $n>0$ and $t \le n^{\gamma}$,  there exists an efficient construction of   one-many non-malleable codes in the $2$-split state model with tampering degree $t$, relative rate $n^{\Omega(1)}/n$, and error $2^{-n^{\Omega(1)}}$.
 \end{thm2}

We next improve the min-entropy rate requirements of seeded non-malleable extractors. As mentioned above, prior to this work, the best known such construction worked for min-entropy rate $0.499$ \cite{Li12b}. We have the following result.

\begin{thm3}  There exists a constant $c$ such that for all $n>0$ and  $\epsilon>0$, and $k \ge c \log^{2}\l(\frac{n}{\epsilon}\r)$, there  exists an explicit   construction of a seeded non-malleable extractor $\snmExt:\{ 0,1\}^n \times \{ 0,1\}^d \rightarrow \{ 0,1\}^{m}$, with  $m = \Omega(k)$ and $d =  O\l(\log^{2}\l(\frac{n}{\epsilon}\r)\r)$.
 \end{thm3}
We in fact have a more general result, and can handle $t$-adversarial functions in the seeded non-malleable case as well, which improves a result of \cite{CRS12}. We refer the reader to Section $\ref{seeded_nm}$ for more details.

Combined with the protocol developed in \cite{DW09}, this immediately gives the following result about privacy amplification, which matches the best known result in \cite{Li15b} but has a simpler protocol.

\begin{thm4}
There exists a constant $C$ such that for any $\eps>0$ with $k \geq C(\log n+\log(1/\eps))^2$, there exists an explicit 2-round privacy amplification protocol with an active adversary for $(n, k)$ sources, with security parameter $\log(1/\eps)$ and entropy loss $O(\log n+\log (1/\eps))$.
 \end{thm4}

\subsection{Other possible approaches to construct one-many non-malleable codes}\label{other_approach}
Since a major part of this paper is devoted to constructing explicit seedless $(2,t)$-non-malleable extractors (and providing  efficient sampling algorithms for almost uniformly sampling from the pre-image of any output), and one of the major motivation for such explicit extractors is to construct one-many non-malleable codes in the $2$-split-state model, a natural question is whether existing constructions of non-malleable codes in the $2$-split state model can be modified to satisfy the stronger notion of one-many non-malleable codes.

Our first observation is that not every construction of a one-one non-malleable code satisfies the stronger notion of being a one-many non-malleable code. Intuitively this is because, say $\Enc$ and $\Dec$ are the encoding and decoding function of some non-malleable code against some class of tampering functions $\mathcal{F}$. Thus, for $f_1,f_2 \in \mathcal{F}$, and  any message $m$, $\Dec(f_1(\Enc(m)))$ is close to $D_{f_{i}}$, $i=1,2$. But it is possible this does not rule out the possibility that, for instance  $\Dec(f_2(\Enc(m))) =  \Dec(f_1(\Enc(m))) + m+1$. Clearly,  this code is not non-malleable against two adversaries from $\mathcal{F}$, and hence is not one-many.

We now take a specific example. Suppose $\C$ is  a one-one NM code against $2$-split-state adversaries. We construct another code $\C^{\prime}$ where the message $m$ is broken into $m_1$ and $m_2$ using additive secret sharing. Then one encodes both $m_1$ and $m_2$ separately using the encoder of $\C$ (and includes both as part of the code, each encoding being equally divided in two halves). It is easy to show  that $\C^{\prime}$ is still one-one NM.
 
On the other hand, if the two adversaries act in the following way: one adversary can take the encoding of $m_1$ and put an  encoding of $1$ on his own. That will be the first output code (to message $m_1+1$). Next, the other adversary can take the encoding of $m_2$ and  put an encoding of $0$ on its own. That will be second output code (to message $m_2$). Now it can be seen that the two output code sum to $m + 1$. Thus, $\C^{\prime}$ is not one-many in the $2$-split-state model.
 
It turns out that existing constructions of non-malleable codes in the split-state model either fail to satisfy stronger notion of one-many, or at least it appears non-trivial to generalize the proofs of non-malleability against multiple adversaries. We briefly discuss these approaches, and why it appears non-trivial to extend them to handle multiple adversaries. A first approach could be to generalise the reductions in the recent work of Aggarwal et al.\cite{ADKO14}, and possibly show that one-many non-malleable codes in the $2$-split-state model can be reduced to the problem of constructing one-many non-malleable codes in the bit-tampering model. However, each known construction of a NM code in the bit-tampering model \cite{CG14b} \cite{AGMPP14} follows the general approach of starting out with an initial non-malleable code in the $2$-split-state model (which is also an NM code against bit-wise tampering) of possibly low rate, and then amplifies the rate to almost optimal. Thus, it is not clear how to use this approach to construct one-many NM codes in the $2$-split-state model.

Another approach could be to show that the non-malleable codes constructed by Aggarwal et al.\ \cite{ADL13} generalize to handle many adversaries. However, from a careful examination of their proof it turns out that it is crucially used that the inner product function is an extractor for weak sources at min-entropy rate slightly greater than $\frac{1}{2}$. It turns out that this fact is tailor made for exactly one adversary, and for handling $t>1$ adversaries one needs that the inner product function is an extractor for min-entropy rate approximately $\frac{1}{t+1}$, which is not true. Thus, it is not clear how to extend their approach as well. 

Finally, a third approach could be to extend the construction of the seedless non-malleable extractor for $10$ sources in the work of Chattopadhyay and Zuckerman \cite{CZ14}. However, from a careful examination of the proof it follows that the crucial step of first constructing a seedless non-malleable condenser based on a sum-product theorem fails to generalize when there are more than one adversary.

Thus, it appears that our new explicit constructions of seedless $(2,t)$-non-malleable extractors are a necessity for constructing one-many non-malleable codes in the split-state model.

\subsection{Related work on privacy amplification} \label{seeded_nm_intro}As mentioned above, seeded non-malleable extractors were introduced by Dodis and Wichs in \cite{DW09}, to study the problem of privacy amplification with an active adversary. 


The goal is roughly as follows. We pick a security parameter $s$, and if the adversary Eve remains passive during the protocol then the two parties should achieve shared secret random bits that are $2^{-s}$-close to uniform. On the other hand, if Eve is active, then the probability that Eve can successfully make the two parties output two different strings without being detected is at most $2^{-s}$. We refer the readers to \cite{DLWZ11} for a formal definition.

Here, while one can still design protocols for an active adversary, the major goal is to design a protocol that uses as few number of interactions as possible, and output a uniform random string $R$ that has length as close to $H_{\infty}(X)$ as possible (the difference is called \emph{entropy loss}). When the entropy rate of $X$ is large, i.e., bigger than $1/2$, there exist protocols that take only one round (e.g., \cite{MW07}, \cite{dkrs}). However these protocols all have very large entropy loss. On the other hand, \cite{DW09} showed that when the entropy rate of $X$ is smaller than $1/2$, then no one round protocol exists; furthermore the length of $R$ has to be at least $O(s)$ smaller than $H_{\infty}(X)$. Thus, the natural goal is to design a two-round protocol with such optimal entropy loss. There has been a lot of effort along this line \cite{MW07}, \cite{dkrs}, \cite{DW09}, \cite{RW03}, \cite{KR09}, \cite{ckor}, \cite{DLWZ11}, \cite{CRS12}, \cite{Li12a}, \cite{Li12b}. However, all protocols before the work of \cite{DLWZ11} either need to use $O(s)$ rounds, or need to incur an entropy loss of $O(s^2)$. 

In \cite{DW09}, Dodis and Wichs showed that the previously defined  seeded non-malleable extractor can be used to construct 2-round privacy amplification protocols with optimal entropy loss. They further showed that seeded non-malleable extractors exist when $k>2m+3\log(1/\eps) + \log d + 9$ and $d>\log(n-k+1) + 2\log (1/\eps) + 7$. However, they were not able to construct such extractors. The first explicit construction of seeded non-malleable extractors appeared in \cite{DLWZ11}, with subsequent improvements in \cite{CRS12}, \cite{Li12a}, \cite{DY12}. However, all these constructions require the entropy rate of the weak source to be bigger than $1/2$. In another paper, Li \cite{Li12b} gave the first explicit non-malleable extractor that breaks this barrier, which works for entropy rate $1/2-\delta)$ for some constant $\delta>0$. This is the best known seeded non-malleable extractor to date. Further, \cite{Li12b} also showed a connection between seeded non-malleable extractors and two-source extractors, which suggests that constructing explicit seeded non-malleable extractors with small seed length for smaller entropy may be hard.

In a different line of work, Li \cite{Li12a} introduced the notion of \emph{non-malleable condenser}, which is a weaker object than seeded non-malleable extractor. He then constructed explicit non-malleable condensers for entropy as small as $k=\polylog(n)$ in \cite{Li15b} and used them to give the first two-round privacy amplification protocol with optimal entropy loss, subject to the constraint that $k \geq s^2$.

\subsection{Related work on non-malleable codes}\label{nm_prior}
We give a summary of known constructions of non-malleable codes. As remarked above, all known explicit constructions of non-malleable codes in the information theoretic setting are in framework of what we call as one-one non-malleable codes. That is, the adversary is only given as input a single code and outputs a single code. 

Since the introduction of non-malleable codes by Dziembowski, Pietrzak and Wichs \cite{DPW10}, the most well studied model is the $C$-split-state model introduced above. By a recent line of work \cite{DKO13} \cite{ADL13} \cite{CG14b} \cite{CZ14} \cite{ADKO14}, we now have almost optimal constructions of non-malleable codes in the $C$-state-state model, for any $C\ge 2$. 

In the model of global tampering, Agrawal et al. \cite{AGMPP14} constructed efficient non-malleable codes with rate $1-o(1)$ against a class of tampering functions slightly more general than the family of permutations. 

A notion related to the many-many setting we consider in this work, called as continuous non-malleable codes, was introduced by Faust et al.\ \cite{FMNV14}. In a continuous non-malleable code, the codewords was allowed to be tampered multiple times, but the tampering experiment stops whenever an error message is detected. Thus this model is  weaker than the notion we consider. The constructions provided of continuous non-malleable codes in \cite{FMNV14} are under computational assumptions.

There were also some other conditional results. Liu and Lysyanskaya \cite{LL12} constructed efficient constant  rate non-malleable codes in the split-state model against computationally bounded adversaries under strong cryptographic assumptions. The work of Faust et al.\ \cite{FMVW13} constructed  almost optimal non-malleable codes against the class of polynomial sized circuits  in the CRS framework. \cite{CCP12}, \cite{CCP11}, \cite{CSKM11}, and \cite{FMNV14} considered non-malleable codes in other models. 

The recent work of Chandran et al.\ \cite{CGMPU15} found interesting connections between non-malleable codes in a model slightly more general than the split-wise model and non-malleable commitment schemes.

\subsubsection*{Organization}
We give an overview of all our explicit constructions in Section \ref{section:overview}. We introduce some preliminaries in Section $\ref{section:prelims}$, and formally define seeded and seedless non-malleable extractors in Section \ref{formaldefs}. We use Section $\ref{section:connection}$ to present the connection between seedless $(2,t)$-non-malleable extractors and one-many non-malleable codes in the $2$-split-state model. In Section $\ref{section:proof}$, we present an explicit construction of a seedless $(2,t)$-non-malleable extractor. An explicit construction of a seeded non-malleable extractor construction at polylogarithmic min-entropy is presented in Section $\ref{seeded_nm}$. Finally, we use Section $\ref{efficiency}$ to give efficient encoding and decoding algorithms for the resulting one-many non-malleable codes.

\section{Overview of Our Constructions} \label{section:overview}
In this section, we give an overview of the main ideas involved in our constructions. The main ingredient in all our constructions is an explicit seedless $(2,t)$-non-malleable extractor. Further, we give  efficient algorithms for almost uniformly sampling from the pre-image of any output of this extractor. The explicit construction of many-many non-malleable codes  in the $2$-split state model with tampering degree $t$ then follow in a straightforward way using the connection via Theorem $\ref{connection_t}$. 

It turns out that by a  simple modification of the construction of our seedless non-malleable extractor, we also have an explicit construction of a seeded non-malleable extractor which works for any min-entropy $k\ge \log^{2} n$. We will give an overview of how to achieve this as well.

\subsection{A Seedless $(2,t)$-Non-Malleable Extractor}\label{sec:overview_t}

  Let $\gamma$ be a small enough constant and $C$ a large enough one. Let $t=n^{\gamma/C}$. 

We construct an explicit function $\nmExt: \{ 0,1\}^{n} \times \{ 0,1\}^{n} \rightarrow \{ 0,1\}^{m}$, $m=n^{\Omega(1)}$ which satisfies the following property: If $X$ and $Y$ be  independent $(n,n-n^{\gamma})$-sources on $\{0,1 \}^n$, and $\A_1=(f_1,g_1),\ldots,\A_t=(f_t,g_t)$ are arbitrary $2$-split sate tampering functions such that  for any $i \in [t]$, at least one of $f_i$ or $g_i$ has no fixed points, the following holds: 
\begin{align*} 
|\nmExt(X,Y) \circ \nmExt(\A_1(X,Y)) \circ \ldots \nmExt(\A_t(X,Y)) - \\ U_{m} \circ \nmExt(\A_1(X,Y)) \circ \ldots \nmExt(\A_t(X,Y))| \le \epsilon,
\end{align*}
where $\epsilon = 2^{-n^{\Omega(1)}}$.

By a  convex combination argument (Lemma $\ref{lemma:final_convex}$), we show that if $\nmExt$ satisfies the property above, then it is indeed a seedless $(2,t)$-non-malleable extractor (Definition $\ref{def:t}$).

We introduce some notation.

\textbf{Notation:} For any function $H$, and $V = H(X,Y)$, we use $V^{(i)}$ to denote the random variable $H(\A_i(X,Y))$. If $Z_a,Z_{a+1},\ldots,Z_{b}$ are random variables, we use $Z_{[a,b]}$ to denote the random variable $(Z_{a},\ldots,Z_b)$. 
For any bit string $z$, let $z_{\{ h\}}$ denote the symbol in the $h$'th co-ordinate of $x$.  For a string $x$ of length $m$, and $T \subseteq [m]$, let $x_{\{ T\}}$  be the projection of $x$ onto  the co-ordinates indexed by $T$. For a string $x$ of length $m$, define the string $\slice(x,w)$ to be the prefix of $x$ with length $w$. 

The high level idea of the non-malleable extractor is as follows. Initially we have two independent sources $(X, Y)$ and $t$ tampered version $\{\A_i(X, Y)\}$, which can depend arbitrarily on $(X, Y)$. We would like to gradually break the dependence of $\{\A_i(X, Y)\}$ on $(X, Y)$, until at the end we get an output $\nmExt(X, Y)$ which is independent of all $\{\nmExt(\A_i(X,Y))\}$.

Towards this end, we would first like to create something from $(X, Y)$ that can distinguish from $\{\A_i(X, Y)\}$. More specifically, we will obtain a small string $Z$ of length $n^{\Omega(1)}$ from $(X, Y)$, such that with high probability $Z$ is different from all $\{Z^{(i)}\}$ obtained from $\{\A_i(X, Y)\}$. Next, we will run some iterative steps of extraction from $(X, Y)$, with each step based on one bit of $Z$. The crucial property we will have here is that whenever we reach a bit of $Z$ which is different from the corresponding bits of $\{Z^{(i)}, i \in S\}$ for some subset $S \subseteq [t]$, in that particular step the output of our extraction from $(X, Y)$ will be (close to) uniform and independent of all the corresponding outputs obtained from $\{A_i(X, Y), i \in S\}$. Furthermore this will remain true in all subsequent steps of extraction. Therefore, since $Z$ is different from all $\{Z^{(i)}, i \in [t]\}$, we know that at the end our output $\nmExt(X, Y)$ will be independent of all $\{\nmExt(\A_i(X,Y)), i \in [t]\}$.\ We now elaborate about the two steps in more details below.

\textbf{Step 1:} Here we use the sources $X$ and $Y$ to obtain a random variable $Z$, such that for each $i \in [t]$,  $Z \neq Z^{(i)}$ with probability at least $1-2^{-n^{\Omega(1)}}$. Thus by a union bound with probability $1-2^{-n^{\Omega(1)}}$ we have that $Z$ is different form all $\{Z^{(i)}, i \in [t]\}$. To obtain $Z$, we first take two small slices  $X_1$ and $Y_1$ from the sources $X$ and $Y$ respectively, with size at least $3n^{\gamma}$; and use the strong inner product $2$-source extractor $\IP$ to generate an almost uniform random variable $V=\IP(X,Y)$. Now we take an explicit asymptotically good binary linear error correcting code, and obtain encodings $(E(X), E(Y))$ of $(X, Y)$ respectively. We now use $V$ to pseudorandomly sample $n^{\Omega(1)}$ bits from $E(X)$ to obtain $X_2$, and we do the same thing to obtain $Y_2$ from $E(Y)$. We use known constructions of an averaging sampler $\samp$ \cite{Z97} \cite{Vad04} (see Definition \ref{samp:vad})  to do this (in fact, we can even sample completely randomly since $V$ is close to uniform). 

Now define $$Z= X_1 \circ Y_1 \circ X_2 \circ Y_2.$$ The length of $Z$ is $\ell = n^{\beta}$ bits for some small constant $\beta$.  Fix some $i \in [t]$. We  claim that $Z \neq Z^{(i)}$ with probability at least $1-2^{-n^{\Omega(1)}}$. To see this, assume without loss of generality  that  $f_i$ has no fixed points. If $X_1 \neq X_1^{(i)}$ or $Y_1 \neq Y_1^{(i)}$, then  we have $Z \neq Z^{(i)}$. Now suppose $X_1 = X_1^{(i)}$ and $Y_1 = Y_1^{(i)}$. Thus, $V = V^{(i)}$. We fix $X_1$, since $\IP$ is a strong extractor (Theorem $\ref{strong_ip}$), $V$ is still close to uniform and now it is a function of $Y$, and thus independent of $X$. Since $X \neq X^{(i)}$, by the property of the code, we know that $E(X)$ and $E(X^{(i)})$ must differ in at least a constant fraction of co-ordinates. Thus, if we uniformly (or pseudorandomly) sample $n^{\Omega(1)}$ bits from these coordinates, then with probability $1-2^{-n^{\Omega(1)}}$ the sampled strings will be different. 

We can now fix $Z,\{Z^{(i)}: i \in [t]\}$, such that $Z \neq Z^{(i)}$ for any $i$. Since the size of each $Z^{(i)}$ is small, we have that conditioned on this fixing, the sources $X$ and $Y$ are still independent and have min-entropy at least $n-O(t\ell)$ each (with high probability).

\textbf{Step 2:} Here our goal is to gradually break the dependence of $\{\A_i(X, Y)\}$ on $(X, Y)$, until at the end we get an output $\nmExt(X, Y)$ which is independent of all $\{\nmExt(\A_i(X,Y))\}$. To achieve this, our crucial observation is that while many other techniques in constructing non-malleable seeded extractors (such as those in \cite{DLWZ11}, \cite{CRS12}, \cite{Li12b} fail in the case where both sources are tampered, the powerful technique of alternating extraction still works. Thus, we will be relying on this technique, which has been used a lot in recent studies of extractors and privacy amplification \cite{DW09}, \cite{Li12a}, \cite{Li12b}, \cite{Li13a}, \cite{Li13b}, \cite{Li15}. We now briefly recall the details. The alternating extraction protocol is an interactive protocol between two parties, Quentin and Wendy, using two strong seeded extractors $\Ext_q$, $\Ext_w$. Assume initially Wendy has a weak source $X$ and Quentin has another source $Q$ and a short uniform random string $S_1$.\footnote{In fact, $S_1$ can be a slightly weak random source as well.} Suppose that $X$ is independent of $(Q, S_1)$. In the first round, Quentin sends $S_1$ to Wendy, Wendy computes $R_1 = \Ext_{w}(X,S_1) $ and sends it back to Quentin, and Quentin then computes $S_2= \Ext_q(Q,R_1)$. Continuing in this way, in round $i$, Quentin sends $S_{i}$, Wendy computes the random variables $R_{i} = \Ext_w(X,S_{i})$ and sends it to Quentin, and Quentin then computes the random variable $S_{i+1} = \Ext_q(Q,R_{i})$. This is done for some $u$ steps, and each of the random variables $R_i,S_i$ is of length $m$. Thus,  the following sequence of random variables is generated: $$ S_1, R_1 = \Ext_{w}(X,S_1), S_2 = \Ext_{q}(Q,R_1),\ldots,S_{u} = \Ext_{q}(Q,R_{u-1}),R_{u} = \Ext_{w}(X,S_{u}).$$
Also define the following  look-ahead extractor: $$\laExt(X,(Q,S_1)) =  R_1,\ldots,R_{u}$$
Now suppose we have $t$ tampered versions of $X$: $X^{(1)},\ldots,X^{(t)}$, which can depend on $X$ arbitrarily; and $t$ tampered versions of $(Q, S_1)$: $(Q^{(1)},S_1^{(1)}),\ldots,(Q^{(t)},S_1^{(t)})$, which can depend on $(Q, S_1)$ arbitrarily. Let $\laExt(X,(Q,S_1))=R_1,\ldots,R_u$, and for $h \in [t]$, let $\laExt(X^{(h)},(Q^{(h)},S_1^{(h)})) = R_1^{(h)},\ldots,R_t^{(h)}$. As long as $(X,X^{(1)},\ldots,X^{(t)})$ is independent of $((Q,S_1),(Q^{(1)},S_1^{(1)}),\ldots,(Q^{(t)},S_1^{(t)}))$ and $t, u, m$ are small compared to the entropy of $X$ and $Q$, one can use induction together with standard properties of strong seeded extractors to show that the following holds: for any $j \in [u]$, 
\begin{align*} R_{j},\{ R_i^{(h)} : i \in [j-1], h \in [t]\},\{ (Q^{(h)},S_1^{(h)}): h \in [t] \} \\ \approx U_m,\{ R_i^{(h)} : i \in [j-1], h \in [t]\},\{ (Q^{(h)},S_1^{(h)}): h \in [t] \}\end{align*}

Based on this property, we describe two different approaches to achieve our goal in Step $2$. The first approach was our initial construction, while the second approach is inspired by new techniques in a recent work of Cohen \cite{C15}. It turns out the second approach is simpler and more suitable for our application to many-many non-malleable codes, thus we only provide the formal proof for the second approach in this paper (see Section~\ref{section:proof}). Recall that the high level idea in both approaches is that we will proceed bit by bit based on the previously obtained string $Z$, which is different from all $\{Z^{(i)}, i \in [t]\}$. Whenever we reach a bit of $Z$ which is different from the corresponding bits of $\{Z^{(i)}, i \in S\}$ for some subset $S \subseteq [t]$, in that particular step the output of our extraction from $(X, Y)$ will be (close to) uniform and independent of all the corresponding outputs obtained from $\{A_i(X, Y), i \in S\}$. Furthermore this will remain true in all subsequent steps of extraction. We will achieve this by running some alternating extraction protocol for $\ell$ times, where $\ell$ is the length of $Z$. Each time the alternating extraction will be between $X$ and a new $(Q_h, S_{1, h})$ obtained from $Y$, where we take $S_{1, h}$ to be a small slice of $Q_h$.

\textbf{Construction $1$:}\footnote{formal proofs of the claims in the sketch of Construction $1$ are not provided in this paper.}
Our first approach is based on a generalization of the techniques in \cite{Li13b}. Here we fist achieve an intermediate goal: whenever we reach a bit of $Z$ which is different from the corresponding bits of $\{Z^{(i)}, i \in S\}$ for some subset $S \subseteq [t]$, the output of our extraction from $(X, Y)$ will \emph{have some entropy} conditioned on all the corresponding outputs obtained from $\{A_i(X, Y), i \in S\}$. Suppose at step $h$ ($1 \leq h \leq \ell$) we have obtained $Q_h$ from $Y$ (in the first step we can take a small slice of $Y$ to be $Q_1$) and use it to run an alternating extraction protocol with $X$. We run the alternating extraction for $t+2$ rounds and obtain outputs $R_{h, 1},\ldots,R_{h, t+2}$. The crucial idea is to use the $h$'th bit of $Z$, to set a random variable $W_h$ as either $(R_{h,1},\ldots,R_{h,t+1})$ or $R_{h,t+2}$ (appended with an appropriate number of $0$'s to make them the same length). 

Now consider the subset $S \subseteq [t]$ where the $h$'th bit of $Z$ is different from the $h$'th bit of $\{Z^{(i)}, i \in S\}$. If $W_h=(R_{h,1},\ldots,R_{h,t+1})$ then for all $i \in S$, we have $W_h^{(i)}=R_{h,t+2}^{(i)}$. Since $S$ has at most $t$ elements, the size of $\{W_h^{(i)}\}$ is at most $tm$. Note that $W_h$ has size $(t+1)m$ and is close to uniform. Thus $W_h$ has entropy roughly $m$ conditioned on $\{W_h^{(i)}, i \in S\}$ (here we can ignore the appended $0$'s in $W_h^{(i)}$ since they won't affect the entropy in $W_h$). On the other hand, if  then $W_h=R_{h,t+2}$ then for all $i \in S$, we have $W_h^{(i)}=R_{h,1}^{(i)},\ldots,R_{h,t+1}^{(i)}$. By the property of alternating extraction we have that $W_h$ is close to uniform conditioned on $\{W_h^{(i)}, i \in S\}$.

We can now go from having conditional entropy to being conditional uniform, as follows. We first convert $W_h$ into a somewhere random source by applying an optimal seeded extractor and trying all possible choices of the seed. One can show that conditioned on previous random variables generated in our algorithm, $W_h$ is now a deterministic function of $X$ and thus independent of $Y$ and $Q_h$. We now take another optimal seeded extractor and use each row in this somewhere random source as a seed to extract a longer output from $Q_h$. In this way we obtain a new somewhere random source. If we choose parameters appropriately we can ensure that the size of $W_h$ is much smaller than the entropy of $Q_h$, and thus the number of rows in this new somewhere random source is much smaller than its row length. Therefore, by using an extractor from \cite{BRSW06} we can use this somewhere random source to extract a close to uniform output $V_h$ from $X$. Since $W_h$ has entropy at least $m$ conditioned on $\{W_h^{(i)}, i \in S\}$, as long as the size of $V_h$ is small, using standard arguments one can show that $V_h$ will be close to uniform conditioned on all $\{V_h^{(i)}, i \in S\}$. 

We now go into the next step of alternating extraction, where we will take a strong seeded extractor and use $V_h$ to extract a uniform string $Q_{h+1}$ from $Y$. We will then use $X$ and $Q_{h+1}$ to do the alternating extraction for next step. The point here is that whenever we have $V_h$ is close to uniform conditioned on all $\{V_h^{(i)}, i \in S\}$ for some $S \subseteq [t]$, we can show that $Q_{h+1}$ is close to uniform conditioned on all $\{Q_{h+1}^{(i)}, i \in S\}$. Thus in the next step of alternating extraction, we can first fix all $\{Q_{h+1}^{(i)}, i \in S\}$, and then fix all the $\{R_{h+1, j}^{(i)}, i \in [S], j \in [t+2]\}$, and all the $\{V_{h+1}^{(i)}, i \in S\}$ (these will now be deterministic functions of $X$). Conditioned on this fixing $Q_{h+1}$ is still close to uniform, and $X$ still has a lot of entropy left (as long as the size of each $R_{h+1, j}^{(i)}$ and $V_{h+1}^{(i)}$ is small). Therefore, in this step $V_{h+1}$ will be close to uniform even conditioned on all $\{V_{h+1}^{(i)}, i \in S\}$, i.e., once we have independence it will continue to hold in subsequent steps. Thus our goal is achieved.

\textbf{Construction $2$:} Here we replace our approach in Construction $1$ with a more direct approach, by using the idea of ``flip-flop" alternating extraction introduced in a recent paper by Cohen \cite{C15}, which is again based on the techniques developed in \cite{Li13b}. Again, assume we are now looking at the $h$'th bit of $Z$, and we have obtained $Q_h$ from $Y$. 

Now each step of alternating extraction will consist of two sub steps of alternating extraction, with each sub step taking two rounds. In the first sub step, we use $X$ and $Q_h$ to perform an alternating extraction for two rounds and output $R_{h, 1}, R_{h, 2}$. If the $h$'th bit of $Z$ is $0$, we take $V_h=R_{h, 1}$; otherwise we take $V_h=R_{h, 2}$. Now we will take a strong seeded extractor $\Ext$ and use $V_h$ to extract $\overline{Q}_h=\Ext(Y, V_h)$ from $Y$. We then use $\overline{Q}_h$ and $X$ to perform the second sub step of alternating extraction, which again runs for two rounds and outputs $\overline{R}_{h, 1}, \overline{R}_{h, 2}$. Now if the $h$'th bit of $Z$ is $0$, we take $\overline{V}_h=\overline{R}_{h, 2}$; otherwise we take $\overline{V}_h=\overline{R}_{h, 1}$. One can see that this is indeed in a ``flip-flop" manner.

The idea is as follows. Consider the $h$'th bit of $Z$, and let $S \subseteq [t]$ be such that for all $i \in S$, we have $Z_{\{h\}} \neq Z_{\{h\}}^{(i)}$. Now consider the $h$'th step of alternating extraction. If $Z_{\{h\}}=0$, then in the first sub step of alternating extraction, $V_h=R_{h, 1}$; while for all $i \in S$, we have $V_h^{(i)}=R_{h, 2}^{(i)}$. Now it's possible that $V_h$ depends on $\{V_h^{(i)}, i \in S\}$, and thus $\overline{Q}_h$ also depends on $\{\overline{Q}_h^{(i)}, i \in S\}$. However, when we go into the second sub step of alternating extraction, we will choose $\overline{V}_h=\overline{R}_{h, 2}$; while for all $i \in S$, we have $\overline{V}_h^{(i)}=\overline{R}_{h, 1}^{(i)}$. Thus by the property of alternating extraction, we have that $\overline{V}_h$ is close to uniform conditioned on all $\{\overline{V}_h^{(i)}, i \in S\}$. 

On the other hand, if $Z_{\{h\}}=1$, then in the first sub step of alternating extraction, $V_h=R_{h, 2}$; while for all $i \in S$, we have $V_h^{(i)}=R_{h, 1}^{(i)}$. Thus in this sub step, by the property of alternating extraction, we have that $V_h$ is close to uniform conditioned on all $\{V_h^{(i)}, i \in S\}$. Therefore we also have that $\overline{Q}_h$ is close to uniform conditioned on all $\{\overline{Q}_h^{(i)}, i \in S\}$, and they are deterministic functions of $Y$ given $V_h$ and $\{V_h^{(i)}, i \in S\}$. Thus, when we go into the second sub step of alternating extraction, we can first fix all $\{\overline{Q}_h^{(i)}, i \in S\}$ and $\overline{Q}_h$ is still close to uniform. Now all $\{\overline{V}_h^{(i)}, i \in S\}$ will be deterministic functions of $X$, and thus we can further fix them. As long as the size of each $\overline{R}_{h, j}$ is small, conditioned on this fixing $X$ still has a lot of entropy left. Therefore $\overline{Q}_h$ can still be used to perform an alternating extraction with $X$, and this gives us that $\overline{V}_h=\overline{R}_{h, 1}$ is close to uniform. That is, again we get that $\overline{V}_h$ is close to uniform conditioned on all $\{\overline{V}_h^{(i)}, i \in S\}$.

Once we have this property, we can go into the next step of alternating extraction. We will now take a strong seeded extractor and use $\overline{V}_h$ to extract $Q_{h+1}$ from $Y$, and then use $X$ and $Q_{h+1}$ to perform the next step of alternating extraction. Since $\overline{V}_h$ is close to uniform conditioned on all $\{\overline{V}_h^{(i)}, i \in S\}$, we also have that $Q_{h+1}$ is close to uniform conditioned on all $\{Q_{h+1}^{(i)}, i \in S\}$. Thus by the same argument above, we can first fix all $\{Q_{h+1}^{(i)}, i \in S\}$ and all $\{V_{h+1}^{(i)}, i \in S\}$, and conditioned on this fixing $Q_{h+1}$ is still close to uniform. Therefore $V_{h+1}$ will be close to uniform conditioned on all $\{V_{h+1}^{(i)}, i \in S\}$. Thus, going into the second sub step, we will also have that $\overline{Q}_{h+1}$ is close to uniform conditioned on all $\{\overline{Q}_{h+1}^{(i)}, i \in S\}$. Thus again we can first fix all $\{\overline{Q}_{h+1}^{(i)}, i \in S\}$ and all $\{\overline{V}_{h+1}^{(i)}, i \in S\}$, and conditioned on this fixing $\overline{Q}_{h+1}$ is still close to uniform. Therefore we get that $\overline{V}_{h+1}$ is close to uniform conditioned on all $\{\overline{V}_{h+1}^{(i)}, i \in S\}$, i.e., once we have independence it will continue to hold in subsequent steps. Thus our goal is achieved.

\subsection{An Explicit Seeded Non-Malleable Extractor for Polylogarithmic Min-Entropy}
Let $\gamma>0$ be a small constant. For any $\epsilon>0$, let  $k \ge O\l( \log^{2+ \gamma}\l(\frac{n}{\epsilon}\r)\r)$, $t\le k^{\gamma/2}$ and $d  =  O\l(t ^2 \log^{2}(\frac{n}{\epsilon})\r)$. We construct a function $\snmExt: \{ 0,1\}^n \times \{ 0,1\}^n \rightarrow \{ 0,1\}^{m}$, $m=O\l(\log\l(\frac{n}{\epsilon}\r)\r)$, such that the following holds: If $X$ is a $(n,k)$-source, $Y$ is an independent uniform seed of length $d$, and  $\A_1,\ldots,\A_{t}$ are arbitrary functions with no fixed points, then the following holds: \begin{align*} \snmExt(X,Y),\snmExt(X,\A_1(Y)),\ldots,\snmExt(X,\A_t(y)) \\ \approx_{\epsilon}U_m,\snmExt(X,\A_1(Y)),\ldots,\snmExt(X,\A_t(y)) \end{align*} 
 
 We now describe our construction, which is essentially a simple modification of our seedless non-malleable extractor construction.

\textbf{Step 1:} Let $Y_1$ be a small slice of $Y$. Compute $V = \Ext(X,Y_1)$, where $\Ext$ is a strong seeded extractor. Now we use $V$ to randomly sample bits from $E(Y)$, where $E$ is the encoder of an asymptotically good error correcting code. Let the sampled bits be $Y_2$. We define $$Z = Y_1 \circ Y_2$$
We show  that with high probability $Z \neq Z^{(i)}$ for all $i \in [t]$. We provide a  brief sketch of the argument. Fix any $i \in [t]$. If $Y_1 \neq Y_1^{(i)}$, then clearly $Z \neq Z^{(i)}$. Now suppose $Y_1 = Y_1^{(i)}$. We fix $Y_1$, and since $\Ext$ is a strong seeded extractor, it follows that $V$ is still close to uniform, and is a deterministic function of $X$, thus independent of $Y, \{Y^{i}, i \in [t]\}$. Therefore $V$ can be used to sample bits from $Y$. Since $\A_i$ has no fixed points, it follows that $Y \neq Y^{(i)}$. Thus $E(Y)$ and $E(Y^{(i)})$ must differ in at least a constant fraction of coordinates. Therefore with high probability $Y_2 \neq Y_2^{(i)}$. By a union bound,  with high probability $Z \neq Z^{(i)}$ for all $i \in [t]$.

\textbf{Step 2:} As long as the size of $(Y_1, V, Y_2)$ is small, we can show that conditioned on the fixing of these variables, $X$ and $Y$ are still independent. Moreover both $X$ and $Y$ only lose a small amount of entropy. Now we can use any of Construction $1$ and Construction $2$ above to finish the extraction. The same argument will show that at the end $\snmExt(X,Y)$ will be close to uniform conditioned on all $\{\snmExt(A,\A_i(Y)), i \in [t]\}$.


We refer the reader to Section $\ref{seeded_nm}$ for more details.

\paragraph{Comparison to the LCB in \cite{C15}} Our second approach in constructing non-malleable two-source extractors is inspired by the work of \cite{C15}. Especially, we use the idea of ``flip-flop" alternating extraction introduced there. However, there are also some differences between our construction and the ``Local Correlation Breaker" constructed in \cite{C15}, which are worth pointing out. 

First, in our construction, both sources $X$ and $Y$ are tampered. This results in $t$ random variables $X^{(1)},\ldots,X^{(t)}$ that are arbitrarily correlated with $X$, and $t$ random variables $Y^{(1)},\ldots,Y^{(t)}$ that are arbitrarily correlated with $Y$. In contrast, in the case of Local Correlation Breaker constructed in \cite{C15}, there are only correlated random variables with one source, while the other source is not tampered. In this sense, our construction can actually be viewed as given a stronger version of the LCB. 

Second, the way to obtain a string that distinguishes the correlated parts is quite different. In the case of the LCB, one can simply use the index of each row in the somewhere random source. On the other hand, in our case we do not have such an index, since the only access we have are the two sources $X$ and $Y$. Thus, we have to take extra efforts to create such a string from these two sources, by using error correcting codes and random sampling.

\subsection{Efficient Algorithms for Many-Many Non-Malleable Codes}
The above construction gives a $(2,t)$ non-malleable extractor. However, for our application to constructing explicit many-many non-malleable codes, given any output of the extractor we need to efficiently sample (almost) uniformly from its pre-image. To do this using the construction described above is highly non-trivial. Therefore, in order to make it easy to efficiently sample from the pre-image of an output (i.e., ``inverting" the extractor), we use additional ideas to modify the non-malleable extractor. We now briefly describe the main ideas that we use. Recall that $t$ is the number of tampered versions of the sources, and $\ell$ is the length of the string $Z$ we obtained.

\textbf{Idea 1:} Since our construction of the $(2,t)$ non-malleable extractor involves multiple steps of alternating extraction, we need to first invert the extractors used in these steps. For this purpose, we will use linear seeded strong extractors in all alternating extraction steps. A linear seeded strong extractor is an extractor such that for any fixed seed, the output is a linear function of the input. Thus for any fixed seed, in order to sample uniformly from an output's pre-image, we can just sample uniformly according to a system of linear equations, which can be done efficiently. 

\textbf{Idea 2:} Next, we will divide the sources $X$ and $Y$ into blocks. In each step of alternating extraction, we will also divide $Q_h$ and $\overline{Q}_h$ into blocks. Then, whenever we use an extractor to extract from $X$, $Q_h$ or $\overline{Q}_h$, we will use a completely new block of $X$, $Q_h$ or $\overline{Q}_h$. When we apply an extractor to $Y$ to generate $Q_h$ or $\overline{Q}_h$, we will also use completely new blocks of $Y$ to do this. This ensures that we do not have to deal with multiple compositions of extractors on the same string. That is, different applications of extractors are used on different parts of the inputs; so to invert them we can invert each part separately. Note that each alternating extraction takes at most $2$ rounds, so it suffices to divide $Q_h$ and $\overline{Q}_h$ into two blocks. 

Here, we need to choose the parameters appropriately.\ Let the size of each $S_{h, j}$ and $R_{h, j}$ produced in alternating extraction be roughly $d$, and the size of each block of $Q_h$ and $\overline{Q}_h$ be $n_q$. Since in the analysis of each alternating extraction we need to fix $O(t)$ tampered versions of $(S_{h, j}, \overline{S}_{h, j})$ and $(R_{h, j}, \overline{R}_{h, j})$, we need to have $n_q \geq \Theta(td)$. Now in the analysis of the entire non-malleable extractor, we need to fix $O(t)$ tampered versions of $Q_h$ and $\overline{Q}_h$, and $O(t\ell)$ tampered versions of $(S_{h, j}, \overline{S}_{h, j})$. The total size of this is $O(t \ell d)$. Thus we can take all $t, \ell, d$ to be some small enough $n^{\Omega(1)}$ such that the total entropy loss of $X$ and $Y$ is some small $n^{\Omega(1)}$. Note that $X$ and $Y$ initially have almost full entropy. Therefore, we can divide $X$ and $Y$ into $O(\ell)$ blocks (or even $O(t\ell)$ blocks, for a reason we will explain below), such that even conditioned on the fixing of all $\{S_{h, j}, R_{h, j}, \overline{S}_{h, j}, \overline{R}_{h, j}, Q_h, \overline{Q}_h\}$ and all previous blocks, each block still has entropy rate say at least $0.9$ (this can be achieved as long as $n/(t \ell) \gg t \ell d$). This ensures that each time we apply an extractor, we can use new blocks of $X$ and $Y$.

\textbf{Idea 3:} However, there is one issue with inverting a linear seeded extractor. The problem is that the pre-image size for different seeds may not be the same. For example, if we have a linear seeded extractor that outputs $m$ bits from an $n$-bit input, then one can show that for most seeds the pre-image size is $2^{n-m}$, while for some seed the pre-image size can be $2^n$. If we first generate the seed uniformly and then sample uniformly from the pre-image given each seed, then the overall distribution is not uniform over the entire pre-image, due to the above mentioned size difference. To rectify this, we construct a new linear seeded extractor $\iExt:\{ 0,1\}^{n} \times \{ 0,1\}^{d} \rightarrow \{ 0,1\}^{m}$ with $m=d/2$ that works for entropy rate $0.9$ sources. Moreover $\iExt$ has the property that given any output, for any fixed seed the pre-image size is the same. The idea is as follows. We first take $0.1d$ bits from the seed and use an average sampler to sample $0.9d$ distinct bits from the source. Since we are using a sampler and the source has entropy rate $0.9$, an argument in \cite{Vad04} shows that with high probability conditioned on the $0.1d$ bits of the seed, the sampled $0.9d$ bits from the source also has entropy rate roughly $0.9$. Now we take the rest $0.9$ bits of the seed and the sampled $0.9d$ bits from the source and apply the inner product two-source extractor (or just use leftover hash lemma), which can output $d/2$ uniform random bits. Now the point is that given any output and any fixed seed, the pre-image of the inner product part has the same size,\footnote{Except when the seed is $0$, but we can deal with this by adding a $1$ to both the source and the seed.} and now the pre-image of any sampled bits also have the same size (since the pre-image is just the sampled bits adding any possible choice of the other $n-0.9d$ bits). 

Note that each time we apply $\iExt$, the output length becomes half of the seed length. Thus in the alternating extraction if we start with seed length $d$, then after one sub step of alternating extraction, the output length will become $\Omega(d)$ since the sub step takes at most $2$ rounds. We will truncate the output if necessary to keep it to be the same length, no matter we choose $R_{h,1}$ or $R_{h, 2}$ (since they have different sizes). Now we need to use this output to extract $\overline{Q}_h$ or $Q_{h+1}$ from $Y$. Since the size of $\overline{Q}_h$ or $Q_{h+1}$ is $\Theta(td)$, we will take $\Theta(t)$ new blocks from $Y$ and apply $\iExt$ to them using the \emph{same} seed, and then concatenate the outputs. Since the blocks of $Y$ form a block source, and $\iExt$ is a strong seeded extractor, one can show that the concatenated outputs is close to uniform. We then do the same thing for the next sub step of alternating extraction. Since we need to repeat alternating extraction for $O(\ell)$ steps, we need to divide $Y$ into $O(t\ell)$ blocks; while we can divide $X$ into only $O(\ell)$ blocks.

\textbf{Idea 4:} Now given any output, our sampling strategy is as follows. We first uniformly generate $X_1$ and $Y_1$, from whom we can compute $V=\IP(X_1, Y_1)$. Then we know which bits of the codeword we are sampling. We then uniformly generate these sampled bits $X_2, Y_2$ and thus we obtain $Z$. Once we have $Z$, we will now uniformly generate all $\{S_{h, j}, R_{h, j}, \overline{S}_{h, j}, \overline{R}_{h, j}\}$ produced in alternating extractions. Based on $Z$ and these variables, we can now generate all the blocks of $X$ used and all the $\{Q_h, \overline{Q}_h\}$ by inverting $\iExt$. Finally, based on $\{Q_h, \overline{Q}_h\}$ we can generate all the blocks of $Y$ used by again inverting $\iExt$. 

This almost works except for the following problem. The blocks of $X$ and $Y$ generated must also satisfy the linear equations imposed by $X_2, Y_2$, which are the bits sampled from the codewords of encodings of $X$ and $Y$ by using a linear error correcting code. However, it is unclear what is the dependence between the linear equations imposed by $X_2, Y_2$ and the other linear equations that we obtain earlier. Of course, if they are linearly independent then we are in good shape. 

To solve this problem, our crucial observation is that if $\ell$ is small and the number of blocks is large enough (say we divide the rest of $X$ into $O(\ell)$ blocks and the rest of $Y$ into $O(t\ell)$ blocks for a large enough constant in $O( \cdot )$), then the entire alternating extraction steps only consume say half of the bits of $X$ and $Y$. Thus, whatever linear equations we obtain from these steps are only imposing constraints to say the first half bits of $X$ and $Y$. Therefore, we can hope that the encodings of $X$ and $Y$ use all the bits of $X$ and $Y$, and thus the linear equations imposed by these encodings will be linearly independent of the equations we obtain from alternating extraction (i.e., the second half bits act as ``free variables").

We indeed succeed with this idea. More specifically, we are going to divide the rest of $X$ and $Y$ (the parts excluding $X_1$ and $Y_1$, which has length $n-n^{\Omega(1)}$) into chunks of length $b=\lceil \log n \rceil$. We will now view each chunk as an element in the field $\F_{2^b}$. We then take say $0.9n$ bits and view it as a string in $\F^{0.9n/b}$. We can now use Reed-Solomn code (RS-code for short) in $\F_{2^b}$ to encode this string into a codeword in $\F^{2^b}$. Note that $2^b>n >0.9n/b$, so this encoding is feasible, and it has distance rate $(2^b-0.9n/b)/(2^b)>0.9$. Now, instead of using $V=\IP(X_1, Y_1)$ to sample $n^{\Omega(1)}$ bits, we will sample $n^{\Omega(1)}$ field elements from the encoding of $X$ and $Y$, and then view them as bit strings. Since the RS-code has distance rate $0.9$, again we have that if two strings are different, then with probability $1-2^{-n^{\Omega(1)}}$, the sampled strings of their encodings will also be different. Moreover, the sampled bit string now has length roughly $n^{\Omega(1)}\log n$, which is still small enough.

Now we can continue with our sampling strategy. As before we first generate all the blocks of $X$ and $Y$ used in all alternating extraction steps. This only consists of the first half bits of $X$ and $Y$. Now, any fixing of these bits can be viewed equivalently as fixing the first $0.5n/b$ field elements in a message. Thus we are still left with $0.4n/b$ free field elements, and we have $n^{\Omega(1)}$ linear equations in $\F_{2^b}$ according to the RS-code. As long as the number of free variables is larger than the number of equations (i.e., $0.4n/b>n^{\Omega(1)}$), the property of the RS-encoding ensures that this set of linear equations are linearly independent. Thus, for any fixed first half bits of $X$ and $Y$, the pre-image according to the linear equations imposed by the sampled bits $X_2, Y_2$ has the same size.

\textbf{Summary.} Now we are basically done. Again, given any output, our sampling strategy is as follows. We first uniformly generate $X_1$ and $Y_1$, from whom we can compute $V=\IP(X_1, Y_1)$. Then we know which co-ordinates of the codeword we are sampling. We then uniformly generate these sampled bits $X_2, Y_2$ and thus we obtain $Z$. Once we have $Z$, we will now uniformly generate all $\{S_{h, j}, R_{h, j}, \overline{S}_{h, j}, \overline{R}_{h, j}\}$ produced in alternating extractions. Based on $Z$ and these variables, we can now generate all the blocks of $X$ used and all the $\{Q_h, \overline{Q}_h\}$ by inverting $\iExt$. Based on $\{Q_h, \overline{Q}_h\}$ we can generate all the blocks of $Y$ used by again inverting $\iExt$. Finally, we use the linear equations imposed by $X_2, Y_2$ to generate the rest of the bits in $X$ and $Y$.

To show that we are indeed sampling uniformly from the output's pre-image, we will establish the following two facts.

\noindent \textbf{Fact 1:} For any fixed $Z=z$, any choice of $\{s_{h, j}, r_{h, j}, \overline{s}_{h, j}, \overline{r}_{h, j}\}$ gives the same pre-image size of $(x, y)$. This follows directly from the fact that our linear seeded extractor has the same pre-image size for any seed, and the argument about the linear equations imposed by the RS-code above.

\noindent \textbf{Fact 2:} For different $Z=z$, and different choice of $\{s_{h, j}, r_{h, j}, \overline{s}_{h, j}, \overline{r}_{h, j}\}$, the pre-image size of $(x, y)$ is also the same. This follows because the ``flip-flop" alternating extraction has a symmetric manner. More specifically, no matter each bit of $z$ is 0 or 1, we will use two sub steps of alternating extraction, with each step taking two rounds of alternating extraction. Thus by symmetry no matter each bit of $z$ is 0 or 1, the pre-image size of the blocks of $X$ and $\{q_h, \overline{q}_h\}$ is the same. Moreover, although depending on the $h$'th bit $z$, we may choose either $r_{h, 1}$ or $r_{h, 2}$ (or either $\overline{r}_{h, 1}$ or $\overline{r}_{h, 2}$), we truncate them if necessary to the same size. So when we generate the blocks of $Y$ using them and $\{q_h, \overline{q}_h\}$, the pre-image of the blocks of $Y$ will also have the same size. Thus,  the pre-image size of the blocks of $X$ and $Y$ used for this bit is the same. Therefore, for different $Z=z$ and different $\{s_{h, j}, r_{h, j}, \overline{s}_{h, j}, \overline{r}_{h, j}\}$, the pre-image size is also the same. 

Now the conclusion that we are sampling uniformly from the output's pre-image follows from the above two facts, and the observation that any $(x, y)$ in the output's pre-image produces exactly one sequence of  $z, \{s_{h, j}, r_{h, j}, \overline{s}_{h, j}, \overline{r}_{h, j}\}$.

\section{Preliminaries} \label{section:prelims}
\subsection{Notations}
We use capital letters to denote distributions and their support, and corresponding small letters to denote a sample from the source. Let $[m]$  denote the set $\{1,2,\hdots,m \}$, and  $U_r$ denote the uniform distribution over $\{0,1\}^r$. For a string $x$ of length $m$, define the string $\slice(x,w)$ to be the prefix of length $w$ of $x$. For any $i \in [m]$, let $x_{\{i\}}$ denote the symbol in the $i$'th co-ordinate of $x$, and for any $T \subseteq [m]$, let $x_{\{ T\}}$ denote the projection of $x$ to the co-ordinates indexed by $T$.

\subsection{Min Entropy, Flat Distributions}
The min-entropy of a source $X$ is defined to be $ H_{\infty}(X) = \min_{s \in \support(X)}\l\{1/\log(\Pr[X=s])\r\}$.
A distribution (source) $D$ is flat if it is uniform over a set $S$.
A $(n,k)$-source is a distribution on $\{ 0,1\}^n$ with min-entropy $k$.
It is a well known fact that any $(n,k)$-source is a convex combination of flat sources supported on sets of size $2^k$.

\subsection{Statistical Distance, Convex Combination of Distributions and Probability Lemmas}
\begin{define}[Statistical distance] Let $D_1$ and $D_2$ be two distributions  on a set $ S$. The statistical distance between $D_1$ and $D_2$ is defined to be: 
$$|D_1 - D_2|  = \max_{T \subseteq S} |D_1(T) - D_2(T)| = \frac{1}{2} \sum_{s \in S }|\Pr[D_1=u] - \Pr[D_2=u]|$$
 $D_1$ is $\epsilon$-close to $D_2$ if $|D_1 - D_2| \le \epsilon$. 
\end{define}

\begin{define}[Convex combination]  A distribution $D$ on a set $S$ is a convex combination of distributions $D_1,\ldots,D_l$ on $S$ if there exists non-negative constants (called weights) $w_1,\ldots,w_{\ell}$ with $\sum_{i=1}^l w_i =1$ such that  $\Pr[D=s] =\sum_{i=1}^l w_i \cdot \Pr[D_i =s] $ for all $s \in S$. 
We use the notation $D = \sum_{i=1}^l w_i \cdot D_i $ to denote the fact that $D$ is a convex combination of the distributions $D_1,\ldots,D_{\ell}$ with weights $w_1,\ldots,w_{\ell}$.
\end{define}
\begin{define}
For random variables $X$ and $Y$, we use $X|Y$ to denote a random variable with distribution:   $\Pr[(X|Y) = x] = \sum_{y \in \support(Y)}\Pr[Y=y] \cdot \Pr[X= x| Y=y ]$.
\end{define}
We record the following lemma which follows from the above definitions.
\begin{lemma}\label{sd_convex}Let $X$ and $Y$ be distributions on a set $S$ such that $X = \sum_{i=1}^l w_i \cdot X_i$ and $Y = \sum_{i=1}^l w_i \cdot Y_i $. Then $|X-Y| \le \sum_{i} w_i \cdot |X_i - Y_i|$.
\end{lemma}
\subsection{Seeded and Seedless Extractors}\label{ext_def}
\begin{define}[Strong seeded extractor]\label{def_seeded}A function $\Ext: \{ 0,1\}^{n} \times \{ 0,1\}^{d} \rightarrow \{0,1 \}^{m}$ is called a strong seeded extractor for min-entropy $k$ and error $\epsilon$ if for any $(n,k)$-source $X$ and an independent uniformly random string $U_d$, we have $$ |\Ext(X,U_d) \circ U_d- U_{m} \circ U_d| <\epsilon,$$ where $U_m$ is independent of $U_d$. Further if the function $\Ext(\cdot,u)$ is a linear function over $\F_2$ for every $u \in \{ 0,1\}^d$, then $\Ext$ is called a linear seeded extractor.
\end{define}

\begin{define} [Independent Source Extractor]
A function $\IExt: (\zo^n)^t \to \zo^m$ is an extractor for independent $(n, k)$ sources that uses $t$ sources and outputs $m$ bits with error $\ep$, if for any $t$ independent $(n, k)$ sources $X_1, X_2, \cdots, X_t$, we have

\[|\IExt(X_1, X_2, \cdots, X_t)-U_m| \leq \ep.\]
In the special case where $t=2$, we say $\IExt$ is a \emph{two-source extractor}.
\end{define}

\subsection{Conditional Min-Entropy}
\begin{define} The average conditional min-entropy is defined as $$ \widetilde{H}_{\infty}(X|W) = \log \l( E_{w \leftarrow W}\l[\max_{x} \Pr[X=x | W=w] \r] \r) = - \log E\l[ 2^{-H_{\infty}(X|W=w)} \r] $$
\end{define}
We recall some results on conditional min-entropy from \cite{DORS08}.
\begin{lemma}[\cite{DORS08}] For any $s>0$, $\Pr_{w \leftarrow W}\l[H_{\infty}(X|W=w) \ge \widetilde{H}_{\infty}(X|W)-s\r] \ge 1- 2^{-s}$.
\end{lemma}
\begin{lemma}[\cite{DORS08}]\label{lemma:entropy_loss} If a random variable $B$ can take at most $\ell$ values, then $\widetilde{H}_{\infty}(A|B) \ge H_{\infty}(A) - \ell$.
\end{lemma}
It is sometimes convenient to work with average case seeded extractors, where if a source $X$ has average case conditional min-entropy $\tilde{H}_{\infty}(X|Z)\ge k$ then the output of the extractor is uniform even  when $Z$ is given. 
\begin{lemma}[\cite{DORS08}]\label{lemma:cond_ext} For any $\delta>0$, if $\Ext$ is a $(k,\epsilon)$-extractor then it is also a $(k+ \log\l(\frac{1}{\delta}\r), \epsilon+\delta)$ average case extractor.
\end{lemma}
The following result on conditional min-entropy was proved in \cite{MW07}.
\begin{lemma}\label{lemma:entropy_loss_1} Let $X,Y$ be random variables such that the random variable $Y$ takes at $\ell$ values. Then 
\begin{align*}
 \Pr_{y \sim Y}\l[ H_{\infty}(X| Y = y) \ge H_{\infty}(X) - \log \ell -\log\l(\frac{1}{\epsilon}\r)\r] > 1-\epsilon.
 \end{align*}
\end{lemma} 

We also need the following lemma from \cite{Li12b}.
\begin{lemma}\label{cor:condition}Let $X,Y$ be random variables with supports $S,T \subseteq V$ such that $(X,Y)$ is $\epsilon$-close to a distribution with min-entropy $k$. Further suppose that the random variable $Y$ can take at most $\ell$ values. Then 
\begin{align*}
 \Pr_{y \sim Y}\bigg[ (X| Y = y) \text{ is $2\epsilon^{1/2}$-close to a source with } \text{ min-entropy }  k - \log \ell - \log\l(\frac{1}{\epsilon}\r)\bigg] \ge 1 - 2\epsilon^{1/2}.
 \end{align*}
\end{lemma}
\subsection{Somewhere Random Sources}
\begin{define}A source $X$ is a $t \times k$ somewhere random source if it  comprises of $t$ rows on $\{ 0,1\}^{k}$ such that at least one of the rows is uniformly distributed. The rows may have arbitrary correlations  among themselves.
\end{define}
\subsection{Some Known Extractor Constructions}

We use  explicit  constructions of strong linear seeded extractors \cite{Tr01} \cite{RRV02}.
\begin{thm}[\cite{Tr01}\cite{RRV02}]\label{thm:lsext} For every $n,k,m \in \mathbb{N}$ and $\epsilon>0$ such that $m \le k \le n$, there exists an explicit linear strong seeded extractor $\LSExt: \{ 0,1\}^n \times \{ 0,1\}^{d} \rightarrow \{0,1 \}^{m}$ for min-entropy $k$, error $\epsilon$, and $d= O\l(\frac{\log^{2}(n/\epsilon)}{\log(k/m)}\r)$.
\end{thm}

The following is an explicit construction of a strong seeded extractor with optimal parameters \cite{GUV}.
\begin{thm}\label{guv} For any constant $\alpha>0$, and all integers $n,k>0$ there exists a polynomial time computable  strong seeded extractor $\Ext: \{ 0,1\}^n \times \{ 0,1\}^d \rightarrow \{ 0,1\}^m$   with $d = O(\log n + \log (\frac{1}{\epsilon}))$ and $m = (1-\alpha)k$.
\end{thm}  

We use the following strong seeded extractor constructed by Zuckerman \cite{Zuck07} that achieves seed length $\log(n)+O(\log(\frac{1}{\epsilon}))$ to extract from any source with constant min-entropy.
\begin{thm}[\cite{Zuck07}]\label{zuck7} For all constant $\alpha,\delta,\epsilon>0$ and for all $n>0$ there exists an efficient construction of a strong seeded extractor $\Ext:\{ 0,1\}^n \times \{ 0,1\}^d \rightarrow \{ 0,1\}^m$ with $m\ge (1-\alpha)n$ and $D=2^{d} = O(n)$.
\end{thm}
We recall a folklore construction of a two-source extractors based on the inner product function \cite{CG88}. We include a proof for completeness.
\begin{thm}[\cite{CG88} ]\label{strong_ip} For all $m,r>0$, with $q=2^{m}, n =rm$, let $X,Y$ be independent sources on $\F_q^r$ with min-entropy $k_1,k_2$ respectively. Let $\IP$ be the inner product function over the field $\F_q$.  Then, we have: $$|\IP(X,Y), X - U_m, X| \le \epsilon, \hspace{0.5cm} |\IP(X,Y), Y - U_m, Y| \le \epsilon$$ where $\epsilon =2^{\frac{-(k_1+k_2-n-m)}{2}}$.
\end{thm}
\begin{proof} Let $X,Y$ be uniform on  sets $A,B \subseteq \F_q^r$ respectively, with $|A|=2^{k_1}$ and $|B|=2^{k_2}$. Let $\psi$ be any non-trivial additive character of the finite field $\F_q$. For short, we use $\cdot$ to denote the standard inner product over $\F_q$. We have 
\begin{align*}
\sum_{y \in B}|\sum_{x \in A}\psi(x\cdot y)| &\le (|B|)^{\frac{1}{2}}\l(\sum_{y \in \F_q^r}\sum_{x,x^{\prime}\in A}\psi((x-x^{\prime})\cdot y)\r)^{\frac{1}{2}} \\
&\le |B|^{\frac{1}{2}}\l(\sum_{x,x^{\prime}\in A}\l(\sum_{y \in \F_q^r}\psi((x-x^{\prime})\cdot y)\r)\r)^{\frac{1}{2}}  
\end{align*}
where the first inquality follows by an application of the Cauchy-Schwartz inequality.
Further, whenever $x \neq x^{\prime}$, we have $$\sum_{y \in \F_q^r}\psi((x-x^{\prime})\cdot y) = 0. $$
Thus, continuing with our estimate, we have 
\begin{align*}
\sum_{y \in B}|\sum_{x \in A}\psi(x\cdot y)| &\le |B|^{\frac{1}{2}}\l(  |A|q^r  \r)^{\frac{1}{2}} 
                							=2^{\frac{n+k_1+k_2}{2}}
\end{align*}
Thus, $$ E_{Y}|\E_{X}\psi(\IP(X,Y))| \le 2^{\frac{n-k_1-k_2}{2}}$$
Using Vazirani's XOR Lemma (see \cite{Rao07} for a proof), it now follows that 
$$| \IP(X,Y),Y - U_m,Y | \le 2^{\frac{n+m-k_1-k_2}{2}}$$
It can be similarly shown that $|\IP(X,Y), X - U_m, X| \le  2^{\frac{n+m-k_1-k_2}{2}}$.
\end{proof}

\section{Seeded  and Seedless Non-Malleable Extractors}\label{formaldefs}
We give formally introduce seedless $(2,t)$-non-malleable extractors in this section. We first recall the definition of seeded $t$-non-malleable extractors from \cite{CRS12}, which generalizes the definition introduced in \cite{DW09}. 

\begin{define}[$t$-Non-malleable Extractor] A function $\snmExt:\{0,1\}^n \times \{ 0,1\}^d \rightarrow \{ 0,1\}^m$ is a seeded $t$-non-malleable extractor for min-entropy $k$ and error $\epsilon$ if the following holds : If $X$ is a  source on  $\{0,1\}^n$ with min-entropy $k$ and $\A_1 : \{0,1\}^n \rightarrow \{0,1\}^n,\ldots,  \A_t : \{0,1\}^n \rightarrow \{0,1\}^n$ are  arbitrary tampering function with no fixed points, then
\begin{align*}
|\snmExt(X,U_d) \scirc \snmExt(X,\A_1(U_d)) \scirc \ldots \scirc \snmExt(X,\A_t(U_d))  \scirc U_d \\ -  U_m \scirc \snmExt(X,\A_1(U_d)) \scirc \ldots \scirc \snmExt(X,\A_t(U_d))  \scirc U_d | <\epsilon 
\end{align*}
where $U_m$ is independent of $U_d$ and $X$.
\end{define}

We now proceed to define seedless non-malleable extractors, which were introduced by Cheraghchi and Guruswami in \cite{CG14b}.

We need the following functions.
\[
 \cpy(x,y) =
  \begin{cases}
   x & \text{if } x \neq \same \\
   y       & \text{if } x  = \same

  \end{cases}
\]

\[
 \cpy^{(t)}((x_1,\ldots,x_{t}),(y_1,\ldots,y_{t})) = (\cpy(x_1,y_1),\ldots,\cpy(x_t,y_t))
 \]
\begin{define}[Seedless Non-Malleable Extractor] A function $\nmExt : \{0,1\}^{n} \rightarrow \{0,1\}^{m}$ is  a seedless non-malleable extractor with respect to a class of sources $\mathcal{X}$ and a family of tampering functions $\mathcal{F}$ with error $\epsilon$ if for every distribution $X \in \mathcal{X}$ and every tampering function $f \in \mathcal{F}$, there exists a random variable $D_{X,f}$  on $\{ 0,1\}^m \cup \{ \same\}$ which is independent of the source $X$  such that   $$|\nmExt(X) \scirc \nmExt(f(X)) - U_m \scirc \cpy(D_{X,f},U_m) | \le \epsilon $$  where both $U_m$'s refer to the same uniform $m$-bit string.
\end{define}

When the class of tampering functions are $2$-split-state, the definition of seedless non-malleable extractors specializes as follows.
\begin{define}[Seedless 2-Non-Malleable Extractor]\label{def:t1}
A function $\nmExt : \{ 0,1\}^{n} \times \{ 0,1\}^{n} \rightarrow \{ 0,1\}^m$ is a seedless 2-non-malleable   extractor at min-entropy $k$  and error $\epsilon$ if it satisfies the following property: If $X$ and $Y$ are independent $(n,k)$-sources and   $\A=(f,g)$ is an arbitrary $2$-split-state tampering function,  then there exists a random variable $D_{f,g}$ on $\{ 0,1\}^{m} \cup \{ \same\}$ which is independent of the sources $X$ and $Y$,   such that  $$ |\nmExt(X,Y) \circ \nmExt(\A(X,Y)) - U_m \circ \cpy(D_{f,g},U_m)| < \epsilon$$ where both $U_m$'s refer to the same uniform $m$-bit string.
\end{define}

In this work, we introduce the following natural generalization where the sources $X,Y$ are tampered by $t$ tampering functions, each of which is from the $2$-split-state family.
\begin{define}[Seedless (2,t)-Non-Malleable Extractor]\label{def:t}
A function $\nmExt : \{ 0,1\}^{n} \times \{ 0,1\}^{n} \rightarrow \{ 0,1\}^m$ is a seedless $(2,t)$-non-malleable  extractor at min-entropy $k$  and error $\epsilon$ if it satisfies the following property: If $X$ and $Y$ are independent $(n,k)$-sources and   $\A_1= (f_1,g_1),\ldots,\A_t=(f_t,g_t)$ are $t$ arbitrary $2$-split-state tampering functions,  then there exists a random variable $D_{\vec{f},\vec{g}}$ on $\l(\{ 0,1\}^{m} \cup \{ \same\}\r)^{t}$ which is independent of the sources $X$ and $Y$,   such that  $$ |\nmExt(X,Y), \nmExt(\A_1(X,Y)), \ldots, \nmExt(\A_t(X,Y)) - U_m, \cpy^{(t)}(D_{\vec{f},\vec{g}},U_m)| < \epsilon$$ where both $U_m$'s refer to the same uniform $m$-bit string.
\end{define}

\section{Non-malleable codes via Seedless non-malleable extractors}\label{section:connection}
The following theorem  is a straightforward generalization of the connection found between non-malleable codes and seedless non-malleable extractors \cite{CG14b}.
\begin{thm}\label{connection_t} Let $\nmExt:(\{0,1\}^{n})^2 \rightarrow \{0,1\}^{m}$  be a polynomial time computable seedless $(2,t)$-non-malleable extractor for min-entropy $n$ with error $\epsilon$. Then there exists a one-many non-malleable code with an efficient decoder in the $2$-split-state model with tampering degree $t$, block length $=2n$, relative rate  $\frac{m}{2n}$, and error $=\epsilon2^{mt+1}$.
\end{thm}

The one-many non-malleable codes in the $2$-split-state model  is define in the following way: For any message $s \in \{ 0,1\}^m$, the encoder $\Enc(s)$ outputs a uniformly random string from the set $\nmExt^{-1}(s) \subset \{ 0,1\}^{2n}$. For any codeword $c \in \{0,1\}^{2n}$, the decoder $\Dec$ outputs $\nmExt(c)$. Thus for the encoder to be efficient, one need to sample almost uniform from $\nmExt^{-1}(s)$.

\section{An Explicit Seedless $(2,t)$-Non-Malleable Extractor }\label{section:proof}
We first set up some tools that we use in our extractor construction.
\subsection{Averaging Samplers} 
In our contsruction, we need to pseudorandomly sample a subset $T$ in $[n]$ such that it intersects any large enough subset with high probability. It turns out that a stronger sampling problem has been extensively studied with the following stronger requirement: For any function $f:[n] \rightarrow [0,1]$, the average of $f$ on the sampled subset $T$ is close to its actual mean with high probability. Such sampling procedures are known as averaging samplers. We use the definition from \cite{Vad04}.

\begin{define}[Averaging sampler \cite{Vad04}]\label{def:samp} A function $\samp: \{0,1\}^{r} \rightarrow [n]^{t}$ is a $(\mu,\theta,\gamma)$ averaging sampler if for every function $f:[n] \rightarrow [0,1]$ with average value $\frac{1}{n}\sum_{i}f(i) \ge \mu$, it holds that 
$$ \Pr_{i_1,\ldots,i_t \leftarrow \samp(U_{R})}\l[  \frac{1}{t}\sum_{i}f(i) \le \mu - \theta \r] \le \gamma.$$
$\samp$ has distinct samples if for every $x \in \{ 0,1\}^{r}$, the samples produced by $\samp(x)$ are all distinct.
\end{define}
The following theorem proved by Zuckerman \cite{Z97} essentially shows that seeded extractors are equivalent to averaging samplers.
\begin{thm}[\cite{Z97}] \label{thm:seed_samp}Let $\Ext:\{0,1\}^{n} \times \{ 0,1\}^{d} \rightarrow \{ 0,1\}^{m}$ be a strong seeded extractor for min-entropy $k$ and error $\epsilon$. Let $\{ 0,1\}^{n}= \{ s_1,\ldots,s_{2^d}\}$. Then  $\samp(x) = \l(\Ext(x,s_1) \circ s_1,\ldots,\Ext(x,s_{2^d}) \circ s_{2^{d}} \r) $ is a $(\mu,\theta,\gamma)$ averaging sampler with distinct samples for any $\mu>0$ , $\theta= \epsilon$ and $\gamma= 2^{k-n}$. 
\end{thm}

Using known constructions of strong seeded extractors, we have the following corollary.
\begin{cor}\label{cor:samp} For any constants $\delta_{\samp},\nu_{\samp}>0$,  there exist constants  $\alpha,\beta<\nu_{\samp}$ such that for all $n>0$ and any $r \ge n^{\alpha}$ there exists a  polynomial time computable function $\samp: \{ 0,1\}^r \rightarrow [n]^{t_{\samp}}$  $ t_{\samp} = O(n^{\beta})$ satisfying the following property: for any set $S \subset [n]$ of size $\delta_{\samp} n$, $$\Pr[|\samp(U_r) \cap S| \ge 1] \ge 1 - 2^{-\Omega(n^{\alpha})}.$$ Further $\samp$ has distinct samples.
\end{cor}
\begin{proof} We set the parameter  $\alpha$ as follows. Let $\Ext: \{ 0,1\}^{n^{\alpha}} \times \{ 0,1\}^d \rightarrow \{ 0,1\}^m$ be the strong linear seeded extractor for min-entropy $k = \frac{n^{\alpha}}{2}$ and error $\epsilon = \frac{\delta}{2}$ from Theorem \ref{thm:lsext}. Thus $t=2^d=O(n^{\alpha c})$ for some constant $c$. We choose $\alpha < {\nu}_{\samp}$ small enough such that $c \alpha < \nu_{samp}$ (and set $\beta=c \alpha)$.  The result now follows by using Theorem $\ref{thm:seed_samp}$.
\end{proof}
\subsection{Alternating Extraction}\label{section:altext}
We recall the method of alternating extraction, which we use as a crucial component in our construction.  

The alternating extraction protocol takes in two integer parameters $u,m>0$. Assume that there are two parties, Quentin with a source $Q$ and a uniform seed $S_1$ (which may be correlated with $Q$), and Wendy with source $W$. Further suppose that $(Q,S_1)$ is kept as a secret from Wendy and $W$ is kept a secret from Quentin.The protocol is an interactive process between Quentin and Wendy, and runs for $u$ steps. 

 Let $\Ext_q$, $\Ext_w$ be strong seeded extractors. In the first step, Quentin sends $S_1$ to Wendy, Wendy computes $R_1 = \Ext_{w}(X,S_1) $ and sends it back to Quentin, and Quentin then computes $S_2= \Ext_q(Q,R_1)$. Continuing in this way, in step $i$, Quentin sends $S_{i}$, Wendy computes the random variables $R_{i} = \Ext_w(X,S_{i})$ and sends it to Quentin, and Quentin then computes the random variable $S_{i+1} = \Ext_q(Q,R_{i})$. This is done for $u$ steps. Each of the random variables $R_i,S_i$ is of length $m$. Thus,  the following sequence of random variables is generated: $$ S_1, R_1 = \Ext_{w}(X,S_1), S_2 = \Ext_{q}(Q,R_1),\ldots,S_{u} = \Ext_{q}(Q,R_{u-1}),R_{u} = \Ext_{w}(X,S_{u}).$$

\textbf{Look-Ahead Extractor} We define the following  look-ahead extractor: $$\laExt(X,(Q,S_1)) =  R_1,\ldots,R_{u}$$

In our application of  the alternating extraction protocol,  the initial seed $S_1$ is not  guaranteed to be uniform but only has high min-entropy\footnote{another way to handle this is to use the extractor from \cite{Raz05}, but we avoid this to ensure invertibility of the final extractor.}. We first  prove a lemma which shows that strong seeded extractors work even when the seed is not uniform but has high enough min-entropy.
\begin{lemma}\label{lemma:almost_full} Let $\Ext: \{ 0,1\}^{n} \times \{ 0,1\}^d \rightarrow \{ 0,1\}^m$ be a strong seeded extractor for min-entropy $k$, and error $\epsilon$. Let $X$ be a $(n,k)$-source and let $Y$ be a source on $\{ 0,1\}^{d}$ with min-entropy $d-\la$. Then, $$ |\Ext(X,Y) \circ Y- U_m \circ Y| \le 2^{\la} \epsilon.  $$
\end{lemma}
\begin{proof} 
Since $Y$ is a source with min-entropy $d-\la$, we can assume it is uniform on a set $A$ of size $2^{d-\la}$. Thus 
\begin{align*}
|\Ext(X,Y) \circ Y- U_m \circ Y| &= \frac{1}{2^{d-\la}} \sum_{y \in A}|\Ext(X,y) - U_m| \\
 						&\le  \frac{1}{2^{d-\la}} \sum_{y \in \{0,1\}^d}|\Ext(X,y) - U_m|  \\
						&\le \frac{1}{2^{d-\la}} 2^d \epsilon = 2^{\la} \epsilon
\end{align*}
where the last inequality uses the fact that $\Ext$ is a strong seeded extractor.
\end{proof}
\textbf{Notation:} If $Z_a,Z_{a+1},\ldots,Z_{b}$ are random variables, we use $Z_{[a,b]}$ to denote the random variable $Z_{a},\ldots,Z_b$. 

We  now prove a general lemma which establishes a strong property satisfied by the alternating extraction protocol. The proof uses ideas from a result proved by Li on alternating extraction \cite{Li13b}, and in fact generalizes this result. 

\begin{lemma}\label{main_altext} Let $X$ be a  $(n_{w},k_{w})$-source and let $X^{(1)},\ldots,X^{(t)}$ be random variables on $\{ 0,1\}^{n_w}$ that are arbitrarily correlated with $X$. Let $Y=(Q,S_1), Y^{(1)} = (Q^{(1)}, S_1^{(1)}), \ldots, Y^{(t)} = (Q^{(t)}, S_1^{(t)})$ be arbitrarily correlated random variables  that are  independent of $(X,X^{(1)},X^{(2)},\ldots,X^{(t)})$. Suppose that $Q$ is a $(n_q,k_q)$-source,  $S_1$ is a $(m,m-\la)$-source,  $Q^{(1)},\ldots,Q^{(t)}$ are each on $n_q$ bits, and $S^{(1)},\ldots,S^{(t)}$ are each on $m$ bits. Let $\Ext_q,\Ext_w$ be strong seeded extractors that extract $m$ bits at min-entropy $k$ with error $\epsilon$ and seed length $m$. Let $\laExt$ be the look-ahead extractor for an alternating extraction protocol with parameters $u,m$, with $\Ext_q,\Ext_w$ being the strong seeded extractors used by Quentin and Wendy respectively. Let $\laExt(X,Y) = R_1,\ldots,R_u$ and for $ j \in [t]$,  $\laExt(X^{(j)},Y^{(j)}) = R^{(j)}_1,\ldots,R^{(j)}_u$. If $k_w,k_q \ge k+ u (t+1) m+2 \log(\frac{1}{\epsilon})$, then the following holds for each $ i \in [u]$: 
$$R_{i}, R_{[1,i-1]}, R_{[1.i-1]}^{(1)}, \ldots, R_{[1,i-1]}^{(t)}, Q, Q^{(1)}, \ldots, Q^{(t)}   \approx_{\epsilon_i}  U_m, R_{[1,i-1]}, R_{[1.i-1]}^{(1)}, \ldots, R_{[1,i-1]}^{(t)},Q, Q^{(1)}, \ldots, Q^{(t)}$$  where $\epsilon_i = O( u \epsilon+ 2^{\la}\epsilon)$.
\end{lemma}
\begin{proof} We in fact prove the following stronger claim.
\begin{claim} For each $ i \in [u]$ the following hold:
\begin{align*}
R_{i}, R_{[1,i-1]}, R_{[1,i-1]}^{(1)}, \ldots, R_{[1,i-1]}^{(t)}, S_{[1,i]}, S_{[1,i]}^{(1)}, \ldots, S_{[1,i]}^{(t)}, Q, Q^{(1)}, \ldots, Q^{(t)}  \\  \approx_{\epsilon_i}   U_m, R_{[1,i-1]}, R_{[1,i-1]}^{(1)}, \ldots, R_{[1,i-1]}^{(t)}, S_{[1,i]}, S_{[1,i]}^{(1)}, \ldots, S_{[1,i]}^{(t)}, Q, Q^{(1)}, \ldots, Q^{(t)} 
\end{align*} and 
\begin{align*}
S_{i+1}, S_{[1,i]}, S_{[1,i]}^{(1)}, \ldots, S_{[1,i]}^{(t)}, R_{[1,i]}, R_{[1,i]}^{(1)}, \ldots, R_{[1,i]}^{(t)}, X, X^{(1)}, \ldots, X^{(t)}  \\  \approx_{\epsilon_i+2\epsilon}   U_m, S_{[1,i]}, S_{[1,i]}^{(1)}, \ldots, S_{[1,i]}^{(t)}, R_{[1,i]}, R_{[1,i]}^{(1)}, \ldots, R_{[1,i]}^{(t)}, X, X^{(1)}, \ldots, X^{(t)} 
\end{align*}
where $\epsilon_i =  4 (i-1) \epsilon+ 2^{\la}\epsilon$.
Further, conditioned on $R_{[1,i-1]}, R_{[1,i-1]}^{(1)}, \ldots, R_{[1,i-1]}^{(t)}, S_{[1,i]}, S_{[1,i]}^{(1)}, \ldots, S_{[1,i]}^{(t)}$, (a)  $(X,X^{(1)},\ldots,X^{(t)})$ is independent from $(Y,Y^{(1)},\ldots,Y^{(t)})$, (b) $X,Q$ each have  average conditional min-entropy at least $(u-i)(t+1)m+k+2\log\l(\frac{1}{\epsilon}\r)$ and (c) $R_i,R_i^{(1)},\ldots,R_{i}^{(t)}$ are deterministic functions of $(X,X^{(1)},\ldots,X^{(t)})$. 
\end{claim}
\begin{proof} We prove this claim by induction on $i$. 

Let $i=1$. Since $R_1 = \Ext_w(X,S_1)$, and $\Ext_w$ is a strong-seeded extractor, it follows by Lemma $\ref{lemma:almost_full}$ that $ \Ext_w(X,S_1),S_1 \approx_{\epsilon_1} U_m,S_1$, where $\epsilon_1 = 2^{\la}\epsilon$. Thus we can fix $S_1$, and $R_1$ is still $\epsilon_1$-close to uniform on average. We note that $R_1$ is a deterministic function of $X$. Since the random variables $S_1^{(1)},\ldots,S_1^{(t)},Q,Q^{(1)},\ldots,Q^{(t)}$ are deterministic functions of $Y,Y^{(1)},\ldots,Y^{(t)}$ and thus uncorrelated with $X$, we have $$ R_1,S_1,S_1^{(1)},\ldots,S_1^{(t)},Q,Q^{(1)},\ldots,Q^{(t)} \approx_{\epsilon_1} U_m,S_1,S_1^{(1)},\ldots,S_1^{(t)},Q,Q^{(1)},\ldots,Q^{(t)}. $$
We fix the random variables $S_1,S_1^{(1)},\ldots,S_1^{(t)}$. By Lemma \ref{lemma:entropy_loss}, the source $Q$ has average conditional min-entropy at least $k_q-m(t+1) = k+ (u-1)m(t+1) + 2\log\l(\frac{1}{\epsilon}\r)$ after this fixing. Using Lemma $\ref{lemma:cond_ext}$ it follows that $\Ext_q$ is a $(k+\log\l(\frac{1}{\epsilon}\r),2\epsilon)$ strong average case extractor. We also note that $R_1,R_1^{(1)},\ldots,R_1^{(t)}$ are now deterministic functions of $X,X^{(1)},\ldots,X^{(t)}$. Thus recalling that $S_2 = \Ext_{q}(Q,R_1)$, we have $S_2,R_1 \approx_{(2\epsilon+\epsilon_1)} U_m,R_1$, since $R_1$ is $\epsilon_1$-close to uniform and using the fact that by Lemma $\ref{lemma:cond_ext}$ $\Ext_w$ is a $(k+\log\l(\frac{1}{\epsilon}\r),2\epsilon)$ strong average case extractor. Thus on fixing $R_1$,  $S_2$ is $(2\epsilon+\epsilon_1)$-close to $U_m$ on average  and is a deterministic function of $Y$.  Since the random variables $R_1^{(1)},\ldots,R_1^{(t)}$  are deterministic functions of $X,X^{(1)},\ldots,X^{(t)}$,  we thus have \begin{align*}S_2,S_1,S_1^{(1)},S_1^{(t)},R_{1},R_{1}^{(1)},\ldots, R_1^{(t)},X,X^{(1)},\ldots,X^{(t)} \\ \approx_{\epsilon_1+2\epsilon} U_m,S_1,S_1^{(1)},S_1^{(t)},R_{1},R_{1}^{(1)},\ldots, R_1^{(t)},X,X^{(1)},\ldots,X^{(t)}
\end{align*}

Further, it still holds that $(X,X^{(1)},\ldots,X^{(t)})$ is independent from $(Y,Y^{(1)},\ldots,Y^{(t)})$. This proves the base case of our induction.

Now suppose that the claim is true for $i$ and we will prove it for $i+1$. Fix the random variables $R_{[1,i-1]}, R_{[1,i-1]}^{(1)}, \ldots, R_{[1,i-1]}^{(t)}, S_{[1,i]}, S_{[1,i]}^{(1)}, \ldots, S_{[1,i]}^{(t)}$. By induction hypothesis, it follows that $X,Q$ each have average conditional min-entropy at least $(u-i)m(t+1) + k + 2\log\l(\frac{1}{\epsilon}\r)$ after this fixing. We now fix the random variables $R_i,R_i^{(1)},\ldots,R_{i}^{(t)}$ (these random variables are deterministic functions of $X,X^{(1)},\ldots,X^{(t)}$ by induction hypothesis).  Thus by Lemma $\ref{lemma:entropy_loss}$, the source $X$ has conditional min-entropy at least $(u-i)(t+1)m + k + 2 \log\l(\frac{1}{\epsilon} \r) - (t+1)m  = (u-i-1)(t+1)m + k + 2 \log\l(\frac{1}{\epsilon} \r) $ after this fixing.   

Since $S_{i+1}=\Ext_{q}(Q,R_{i})$ is now independent of $X$ and $(\epsilon_i+2\epsilon)$-close to $U_m$ on average (by induction hypothesis), it follows that $\Ext_w(X,S_{i+1}),S_{i+1} \approx_{\epsilon_i+4 \epsilon} U_m,S_{i+1}$. Thus on fixing $S_{i+1}$, $R_{i+1}= \Ext_w(X,S_{i+1})$ is $(\epsilon_i+4 \epsilon)$-close to $U_m$ on average, and is a deterministic function of $X$. We  also fix the random variables $S_{i+1}^{(1)},\ldots,S_{i+1}^{(t)}$. Since we have fixed  the random variables $R_i^{(1)},\ldots,R_{i}^{(t)}$, thus $S_{i+1}^{(1)},\ldots,S_{i+1}^{(t)}$ are deterministic functions of $Y,Y^{(1)},\ldots,Y^{(t)}$. Hence $R_{i+1}$ is still $\epsilon_{i+1}$-close to uniform on average and a deterministic function of $X$ after this fixing. Thus, 
\begin{align*}
R_{i+1}, R_{[1,i]}, R_{[1,i]}^{(1)}, \ldots, R_{[1,i]}^{(t)}, S_{[1,i+1]}, S_{[1,i+1]}^{(1)}, \ldots, S_{[1,i+1]}^{(t)}, Q, Q^{(1)}, \ldots, Q^{(t)}  \\  \approx_{\epsilon_{i+1}}   U_m, R_{[1,i]}, R_{[1,i]}^{(1)}, \ldots, R_{[1,i]}^{(t)}, S_{[1,i+1]}, S_{[1,i+1]}^{(1)}, \ldots, S_{[1,i+1]}^{(t)}, Q, Q^{(1)}, \ldots, Q^{(t)}. 
\end{align*} 

The source $Q$ has conditional min-entropy at least $(u-i)(t+1)m + k + 2 \log\l(\frac{1}{\epsilon} \r) - (t+1)m  = (u-i-1)(t+1)m + k + 2 \log\l(\frac{1}{\epsilon} \r) $. 

Recall that $S_{i+2}=\Ext_q(Q,R_{i+1})$. Since $\Ext_q$ is a $(k+\log\l(\frac{1}{\epsilon}\r),2\epsilon)$ strong average case extractor, it follows that $\Ext_q(Q,R_{i+1}),R_{i+1} \approx_{\epsilon_{i+2}+2 \epsilon} U_m$.  Since  the random variables  $R_{i+1}^{(1)},\ldots,R_{i+1}^{(t)}$ are deterministic functions of $X,X^{(1)},\ldots,X^{(t)}$ (recall that we have fixed $S_{i+1}^{(1)},\ldots,S_{i+1}^{(t)}$), it follows that 
\begin{align*}
S_{i+2}, S_{[1,i+1]}, S_{[1,i+1]}^{(1)}, \ldots, S_{[1,i+1]}^{(t)}, R_{[1,i+1]}, R_{[1,i+1]}^{(1)}, \ldots, R_{[1,i+1]}^{(t)}, X, X^{(1)}, \ldots, X^{(t)}  \\  \approx_{\epsilon_{i+1}+2\epsilon}   U_m, S_{[1,i+1]}, S_{[1,i+1]}^{(1)}, \ldots, S_{[1,i+1]}^{(t)}, R_{[1,i+1]}, R_{[1,i+1]}^{(1)}, \ldots, R_{[1,i+1]}^{(t)}, X, X^{(1)}, \ldots, X^{(t)}. 
\end{align*}
Also, we maintain at each step  that $(X,X^{(1)},\ldots,X^{(t)})$ is independent from $(Y,Y^{(1)},\ldots,Y^{(t)})$. This completes the proof.
\end{proof}
\end{proof}
\begin{remark}We note that if instead of using a strong seeded extractor to generate $R_1$ (recall  $R_1= \Ext_w(X,S_1)$), we used the extractor constructed by Raz \cite{Raz05}, then the error achieved is $O(u\epsilon)$.
\end{remark}
\subsection{Construction of Some Key Components}
In this section,  we construct functions  which are key ingredients in all our explicit extractor constructions. It is based on a new way of using the technique of alternating extraction, and is inspired by  a recent elegant work of Cohen \cite{C15} on constructing local correlation breakers.

We define the following function which is inspired by the ``flip-flop" method introduced by Cohen \cite{C15}.
\RestyleAlgo{boxruled}
\LinesNumbered
\begin{algorithm}[ht]
  \caption{2$\laExt(x,y,q_i,b)$\label{alg}  \vspace{0.1cm}\newline \textbf{Input:} Bit strings $x,y,q_i$ of length $n_w,n_y,n_q$ respectively, and a bit $b$. \newline \textbf{Output:} A bit string of length $n_q$. \newline \textbf{Subroutine:} Let $\Ext_q: \{0,1 \}^{n_q} \times \{ 0,1\}^{m} \rightarrow \{ 0,1\}^{m}$ be a strong seeded  extractor set to extract from min-entropy  $k$ with error $\epsilon$ and seed length $m$. Let $\Ext_w: \{0,1 \}^{n_w} \times \{ 0,1\}^{m} \rightarrow \{ 0,1\}^{m}$ be a strong seeded  extractor set to extract from min-entropy  $ k$ with error $\epsilon$ and seed length $d$.\newline
 Let $\laExt:  \{0,1 \}^{n_w} \times \{ 0,1\}^{n_q+m}\rightarrow \{ 0,1\}^{2m}$  be the look ahead extractors defined in Section $\ref{section:altext}$ for an alternating extraction protocol  with parameters $m,u=2$ (recall $u$ is the number of steps in the protocol,  $m$ is the length of each random variable that is communicated between the players), and using $\Ext_q,\Ext_w$ as the strong seeded extractors. \newline Let $\Ext : \{0,1\}^{n_y} \times \{ 0,1\}^{m} \rightarrow \{ 0,1\}^{n_q}$ be a strong seeded extractor set to extract from min-entropy $k_1$ with error $\epsilon$.
 }
Let $s_{i,1} = \slice(q_i,m)$
 
Let $\laExt(x,(q_i,s_{i,1}))=r_{i,1},r_{i,2}$
  
if $b=0$, let $\overline{q_{i}} = \Ext(y,r_{i,1})$

\hspace{0.72cm}else let $\overline{q_{i}} = \Ext(y,r_{i,2})$
 
endif 
 
Let $\overline{s_{i,1}}= \slice(\overline{q_{i}},m)$.
 
 Let $\laExt(x,(\overline{q_{i}},\overline{s_{i,1}}))= \overline{r_{i,1}},\overline{r_{i,2}}$.
 
if $b=0$, let $q_{i+1} = \Ext(y,\overline{r_{i,2}})$

\hspace{0.72cm}else let $q_{i+1} = \Ext(y,\overline{r_{i,1}})$
 
endif 

Ouput $q_{i+1}$.
\end{algorithm}

We now prove the following lemma. 
\begin{lemma}\label{base:lemma} Let $b,\{b^{(h)} : h \in [j] \}$ be $j+1$ bits such that for all $h \in [j]$, $b \neq b^{(h)}$.  
Let $X$ be a  $(n_{w},k_w)$-source and let $\{ X^{(h)}: h \in [j]\}$ be random variables on $\{ 0,1\}^{n_w}$ that are arbitrarily correlated with $X$. Let $Y,  \{ Y^{(h)} : h \in [j] \}$ be  arbitrarily correlated random variables  that are  independent of  $(X,\{ X^{(h)}: h \in [j]\})$. Suppose that $Y$ is a $(n_y,k_y)$-source, $k_y= n_y - \la$,  each random variable in $\{ Y^{(h)} : h \in [j] \}$ is on $n_y$ bits. Let $Q_{i}$ be some function of $Y$ on $n_q$ bits with min-entropy at least $n_q-\la$, and for each $h \in [j]$, let $Q^{(h)}$ be an an arbitrary function of $Y,\{ Y^{(a)}:a \in [j]\}$ on $n_q$ bits.

Let $2\laExt$ be the function computed by Algorithm $1$. Let $2\laExt(X,Y,Q_i,b)=Q_{i+1}$, and for $h \in [j]$, let $2\laExt(X^{(h)},Y^{(h)},Q_i^{(h)},b^{(h)})=Q_{i+1}^{(h)}$. Suppose $k_y \ge \max\{k , k_1\} + 10\l(jn_q+ jm+\log\l(\frac{1}{\epsilon}\r)\r)$, $k_w \ge  k  + 10\l( jm+\log\l(\frac{1}{\epsilon}\r)\r)$, and $n_q \ge k+  10jm+2 \log(\frac{1}{\epsilon}) + \la $.  

Then  with probability at least $1-\epsilon^{\prime}$, where $\epsilon^{\prime} = O(2^{\la} \epsilon)$, over the fixing of the random variables $Q_i, \{ Q_{i}^{(h)} : h \in [j] \},  R_{i,1},R_{i,2}, \{ R^{(h)}_{i,1}, R^{(h)}_{i,2} : h \in [j]  \}, \overline{Q}_i, \{ \overline{Q}_{i}^{(h)}: h \in [j] \}, \overline{R}_{i,1},\overline{R}_{i,2}, \{ \overline{R}^{(h)}_{i,1},\overline{R}^{(h)}_{i,2}  : h \in [j]  \}, \{ Q_{i+1}^{(h)} : h \in [j] \} $ : (a) $Q_{i+1}$ is $\epsilon^{\prime}$-close to $U_{n_q}$ and is a deterministic function of $Y$ (b) The random variables $(X,\{ X^{(h)}: h \in [j]\})$ and $(Y,\{ Y^{(h)}: h \in [j]\})$  are independent (c) $X$ has min-entropy at least $k_w - 10\l( jm+\log\l(\frac{1}{\epsilon}\r)\r)$ and $Y$ has min-entropy at least $k_y-10\l( jn_q+ jm+\log\l(\frac{1}{\epsilon}\r)\r)$. 
\end{lemma}
\begin{proof} 
\textbf{Notation:} For any function $H$, if $V = H(X,Y)$, let $V^{(a)}$  denote the random variable $H(X^{(a)},Y^{(a)})$. 

We split the proof into two cases, depending on $b$.

\noindent
Case $1$: Suppose $b=1$.  By Lemma \ref{main_altext}, it follows that \begin{align*} R_{i,2},\{ R_{i,1}^{(h)}: h \in [j]\},Q_i,\{ Q_{i}^{(h)} : h \in [j] \} \\ \approx_{\epsilon_1}  U_m,\{ R_{i,1}^{(h)}: h \in [j]\},Q_i,\{ Q_{i}^{(h)} : h \in [j] \}\end{align*}, where  $\epsilon_1=c2^{\la}\epsilon$, for some constant $c$. Thus, we can fix $\{ R_{i,1}^{(h)}: h \in [j]\},Q_i,\{ Q_{i}^{(h)} : h \in [j] \}$, and with probability at least $1-O(\epsilon_1)$, $R_{i,2}$ is $O(\epsilon_1)$-close to $U_m$. Note that $R_{i,2}$ is now a deterministic function of $X$. Further, by Lemma $\ref{lemma:entropy_loss_1}$, $Y$ loses min-entropy at most $(j+1)n_q + \log\l(\frac{1}{\epsilon}\r)$ with probability at least $1-\epsilon$ due to this fixing. Since on fixing $Q_i,\{ Q_{i}^{(h)} : h \in [j] \}$, the random variables $\{ R_{i,1}^{(h)}: h \in [j]\}$ are deterministic function of  $X,\{X^{(h)} : h \in [j] \}$, the source $X$  loses min-entropy at most  $j m + \log\l(\frac{1}{\epsilon}\r)$ with probability at least $1-\epsilon$ due to this fixing. We now note that the random variables $\{ \overline{Q}_{i}^{(h)}: h \in [j]\}$ are deterministic functions of $Y,\{ Y^{(h)}: h \in [j]\}$. Thus, we fix $\{ \overline{Q}_{i}^{(h)}: h \in [j] \}$,  and by Lemma $\ref{lemma:entropy_loss_1}$, $Y$ loses min-entropy at most $jn_q+\log\l(\frac{1}{\epsilon}\r)$ with probability at least $1-\epsilon$ due to this fixing. Since $\Ext$ extracts from min-entropy $k_1$, and $k_y$ was chosen large enough, it follows that the random variable $\overline{Q}_{i}$ is $(\epsilon+\epsilon_1)$-close to $U_{n_q}$ with probability at least $1-O(\epsilon_1)$ even after the fixing. Further,  we fix $R_{i,2}$ since $\Ext$ is a strong seeded extractor, and by Lemma $\ref{lemma:entropy_loss_1}$, $X$ loses min-entropy at most $ m + \log\l(\frac{1}{\epsilon}\r)$ with probability at least $1-\epsilon$ due to this fixing. Thus $\overline{Q}_i$ is now a deterministic function of $Y$. We now fix the random variables $\{ R_{i,2}^{(h)}: h \in [j]\}$, noting that they are deterministic functions of $X$ and hence does not affect the distribution of $\overline{Q}_i$. $X$  loses min-entropy at most  $j m + \log\l(\frac{1}{\epsilon}\r)$ with probability at least $1-\epsilon$ due to this fixing.

We now note that the random variables $\{\overline{R}_{i,1}^{(h)},\overline{R}_{i,2}^{(h)}: h \in [j]\}$ are deterministic function of $X,\{X^{((j))} : j \in [h]\}$ since we have fixed $\{\overline{Q}_i^{(h)} : h \in [j]\}$. Thus, we can fix  $\{\overline{R}_{i,1}^{(h)},\overline{R}_{i,2}^{(h)}: h \in [j]\}$ and $X$ loses min-entropy at most $2jm+\log\l( \frac{1}{\epsilon}\r)$ with probability at least $1-\epsilon$. Thus it follows by Lemma $\ref{main_altext}$ that $|\overline{R}_{i,1},\overline{Q}_{i} - U_{m},\overline{Q}_i| < \epsilon+O(\epsilon_1) $. We fix  $\overline{Q}_i$ and $Y$ loses min-entropy at most $n_q+\log\l(\frac{1}{\epsilon}\r)$ using Lemma $\ref{lemma:entropy_loss_1}$. Finally, we note that $\{ Q_{i+1}^{(h)} : h \in [j] \} $ is now a deterministic function of $Y,\{ Y^{(h)}: h \in [j]\}$. Thus, we can fix $\{ Q_{i+1}^{(h)} : h \in [j] \} $ variables and $Y$ loses min-entropy at most $jn_q+\log\l(\frac{1}{\epsilon}\r)$ with probability at least $1-\epsilon$ due to this fixing. Further, $\overline{R}_{i,1}$ is now a deterministic function of $X$. It  follows that $Q_{i+1}$ is $O(\epsilon_1+\epsilon)$-close to $U_{n_q}$ since $k_y$ is chosen large enough. We further fix $\overline{R}_{i,1}$ noting that $\Ext$ is a strong extractor and  $X$ loses min-entropy at most $ m + \log\l(\frac{1}{\epsilon}\r)$ with probability at least $1-\epsilon$ due to this fixing. 

\noindent 
Case $2$: Now suppose $b=0$. We fix the random variables $Q_i,\{ Q_{i}^{(h)} : h \in [j] \}$. Conditioned on this fixing, it follows by Lemma $\ref{main_altext}$ that $|R_{i,1}-U_m|<\epsilon_1$, $\epsilon_1 = O(2^{\la}\epsilon)$, with probability at least $1-\epsilon$. Since $\Ext$ is a strong seeded extractor (and $k_y$ is large enough) and $R_{i,1}$ is a deterministic function of $X$, it follows that  $|\overline{Q}_i,R_{1,i} - U_{n_q},R_{i,1}|<\epsilon + \epsilon_1$ with probability at least  $\epsilon$.  We fix $R_{i,1}$, and observe that $\overline{Q}_{i}$ is now a deterministic function of $Y$. We can now fix $\{ R_{i,1}^{(h)},R_{i,2}^{(h)}: h \in [j]\}$ since $\{ R_{i,2}^{(h)}: h \in [j]\}$  is a deterministic function of $X,\{ X^{(h)}: h \in [j]\}$, and hence does not affect the distribution of $\overline{Q}_{i}$. As a result of these fixings, it is clear that $(X,\{ X^{(h)}: h \in [j]\})$ is independent of $(Y_i,\{ Y^{(h)}: h \in [j]\})$. Further $X$ loses min-entropy of at most $2(j+1)m+ \log\l(\frac{1}{\epsilon}\r)$ with probability at least $1-\epsilon$, and  $Y$ loses min-entropy of at most $2(j+1)n_q+ (j+1)m + 3\log\l(\frac{1}{\epsilon}\r)$ with probability at least $1-3\epsilon$. 
Note that now $\overline{Q}_i,\{ Q_i^{(h)}: h \in [j]\}$ are deterministic functions of $Y,\{ Y^{(h)}: h \in [j]\}$, and $\overline{Q}_{i}$ is $O(\epsilon_1)$-close to $U_{n_q}$. 
By Lemma \ref{main_altext}, it follows that $$ \overline{R}_{i,2},\{ \overline{R}_{i,1}^{(h)}: h \in [j]\},\overline{Q}_i,\{ \overline{Q}_i^{(h)}: h \in [j]\} \approx_{\epsilon_2} U_m,\{ \overline{R}_{i,1}^{(h)}: h \in [j]\},\overline{Q}_i,\{ \overline{Q}_i^{(h)}: h \in [j]\}$$ where $\epsilon_2=c(\epsilon_1+ \epsilon+\epsilon)$, for some constant $c$. Thus, we can fix $\{ \overline{R}_{i,1}^{(h)}: h \in [j]\}, \overline{Q}_{i},\{ \overline{Q}_{i}^{(h)}: h \in [j]\}$ and with probability at least $1-O(\epsilon_2)$, $\overline{R}_{i,2}$ is $O(\epsilon_2)$-close to $U_m$. Note that $\overline{R}_{i,2}$ is now a deterministic function of $X$. Further, by Lemma $\ref{lemma:entropy_loss_1}$, $Y$ loses min-entropy at most $(j+1)n_q +  \log\l(\frac{1}{\epsilon}\r)$ with probability at least $1-\epsilon$ due to this fixing. Since on fixing $\overline{Q}_i,\{\overline{Q}_i^{(h)}: h \in [j] \} $, the random variables $\{ \overline{R}_{i,1}^{(h)}: h \in [j]\}$ are deterministic functions of  $X, \{ X^{(h)}: h \in [j]\}$, the source $X$  loses min-entropy at most  $j m + \log\l(\frac{1}{\epsilon}\r)$ with probability at least $1-\epsilon$ due to this fixing. We now note that the random variables $\{ Q_{i+1}^{(h)} : h \in [j] \} $ are deterministic functions of $Y,\{Y^{(h)}:h \in [j]\} $. Thus, we fix $\{ Q_{i+1}^{(h)} : h \in [j] \}$  and by Lemma $\ref{lemma:entropy_loss_1}$, $Y$ loses min-entropy at most $(j+1)n_q+\log\l(\frac{1}{\epsilon}\r)$ with probability at least $1-\epsilon$ due to this fixing. Since $\Ext$ extracts from min-entropy $k_1$,  (and $k_y$ is large enough) it follows that random variable $Q_{i+1}$ is $O(\epsilon_2)$-close to $U_{n_q}$ even after the fixing. Further,  we fix $\overline{R}_{i,2}$ since $\Ext$ is a strong seeded extractor, and by Lemma $\ref{lemma:entropy_loss_1}$, $X$ loses min-entropy $ m + \log\l(\frac{1}{\epsilon}\r)$ with probability at least $1-\epsilon$ due to this fixing. Further   $Q_{i+1}$ is now a deterministic function of $Y$. Thus we can fix the random variables $\{ \overline{R}_{i,2}^{(h)}: h\in [j]\}$ since they are deterministic function sod $X$ and does not affect the distribution of $Q_{i+1}$.  $X$ loses min-entropy at most  $ m + \log\l(\frac{1}{\epsilon}\r)$ with probability at least $1-\epsilon$ due to this fixing. 
 This completes the proof.
\end{proof}
We now construct a function that is a crucial ingredient in our non-malleable extractor constructions. (Recall that for any string $z$, we use $z_{\{h\}}$ to denote the symbol in the $h$'th co-ordinate of $z$.)
\RestyleAlgo{boxruled}
\LinesNumbered
\begin{algorithm}[ht]
  \caption{$\nmExt_1(x,y,z)$\label{alg}  \vspace{0.1cm}\newline \textbf{Input:} Bit strings $x,y,z$ of length $n_w,n_y,\ell$ respectively. \newline \textbf{Output:} A bit string of length $n_q$. }
Let $q_1 = \slice(y,n_q)$ 

\For{$h=1$ to $\ell$} { 

$q_{h+1} = 2\laExt(x,y,q_h,z_{\{h\}})$

}
 
Ouput $q_{\ell+1}$.
\end{algorithm}

\begin{lemma}\label{main_lemma_1} Let $z,z^{(1)},\ldots,z^{(t)}$ each be $\ell$ bit strings such that for all $i \in [t]$, $z \neq z^{(i)}$. Let $X$ be a  $(n_{w},k_w)$-source and let $X^{(1)},\ldots,X^{(t)}$ be random variables on $\{ 0,1\}^{n_w}$ that are arbitrarily correlated with $X$. Let $Y, Y^{(1)}, , \ldots, Y^{(t)}$ be  random variables  on $n_y$ bits  that are  independent of $(X,X^{(1)},X^{(2)},\ldots,X^{(t)})$. Suppose that $Y$ is a $(n_y,k_y)$-source, $k_y=n_y-\la$.

Let $\nmExt_1$ be the function computed by Algorithm $2$. Let $\nmExt_1(X,Y,z)=Q_{\ell+1}$, and for $h \in [t]$, let $\nmExt_1(X^{(h)},Y^{(h)},z^{(h)})=Q_{\ell+1}^{(h)}$. Suppose $k_y\ge \max\{k , k_1\} + 20\ell\l(tn_q+ tm+\log\l(\frac{1}{\epsilon}\r)\r) $, $k_w \ge  k  + 20\ell\l( tm+\log\l(\frac{1}{\epsilon}\r)\r)$ and $n_q \ge k+  10tm+2 \log(\frac{1}{\epsilon}) + \la $.
Then, we have  $$Q_{\ell+1},Q_{\ell+1}^{(1)},\ldots,Q^{(t)}_{\ell+1} \approx_{\epsilon^{\prime}} U_{n_q}, Q_{\ell+1}^{(1)},\ldots,Q^{(t)}_{\ell+1} $$ where $\epsilon_{\ell}^{\prime } = O((2^{\la}+\ell)\epsilon)$.
\end{lemma}
\begin{proof}\textbf{Notation:} For any function $H$, if $V = H(X,Y)$, let $V^{(a)}$  denote the random variable $H(X^{(a)},Y^{(a)})$.

For  $h \in  [\ell]$, define the sets $$\mathcal{\ind}_h = \{ i \in  [t]: z_{\{ h\}} \neq z^{(i)}_{\{ h\}}\} , \hspace{1cm}\overline{\ind}_{h} = [t]\setminus \ind_{h},$$ $$ \ind_{[h]} = \cup_{i=1}^{h}\ind_h, \hspace{1cm}\overline{\ind}_{[h]} = [t]\setminus \ind_{[h]}.$$ 

We record a simple claim.
\begin{claim}\label{all_included}For each $i \in [t]$, there exists $h \in [\ell]$ such that $i \in \ind_h$.
\end{claim}
\begin{proof} Recall that we have fixed $Z,Z^{(1)},\ldots,Z^{(t)}$ such that $Z \neq Z^{(i)}$ for any $ i \in [t]$. Thus it follows that for each $i \in [t]$, there exists some $h \in [\ell]$ such that $Z_{\{h\}} \neq Z^{(i)}_{\{h\}}$, and hence $i \in \mathcal{\ind}_h$.
\end{proof}

We now prove our main claim, which combined with Lemma $\ref{base:lemma}$ and a simple inductive argument proves Lemma $\ref{main_lemma_1}$.
\begin{claim}\label{main_claim_1}For any $h \in \{0,1,\ldots,\ell\}$, suppose the following holds:

With probability at least $1-\epsilon_h$ over the fixing of the random variables  $\{ Q_i : i \in [h]\}, \{ Q_{i}^{(j)} : i \in [h], j \in [t] \},  \{ R_{i,1}, R_{i,2} : i \in [h] \}, \{ R^{(j)}_{i,1}, R^{(j)}_{i,2} : i \in [h], j \in [t]   \}, \{ \overline{Q}_i : i \in [h] \}, \{ \overline{Q}_{i}^{(j)}: i \in [h], j \in [t] \}, \{ \overline{R}_{i,1},\overline{R}_{i,2} : i \in [h] \}, \{ \overline{R}^{(j)}_{i,1},\overline{R}^{(j)}_{i,2} : i \in [h], j \in [t]  \}, \{ Q_{i+1}^{(j)} : j \in \ind_{[h]} \}$: (a) $Q_{h+1}$ is $\epsilon_h$-close to a source with min-entropy at least $n_q -\la$ and is a deterministic function of $Y$ (b) $\{ Q_{h+1}^{(j)}: j \in \overline{\ind}_{[h]}\}$ is a deterministic function of $Y,\{Y^{(j)} : j \in [t] \}$ (c) The random variables $(X,\{ X^{(j)}: j \in [t]\})$ and $(Y,\{ Y^{(j)}: j \in [t]\})$  are independent (d) $X$ has min-entropy at least $k_w - 10 h \l( tm+\log\l(\frac{1}{\epsilon}\r)\r)>  k  + 10\l( tm+\log\l(\frac{1}{\epsilon}\r)\r)$ and $Y$ has min-entropy at least $k_y-10h\l(tn_q + tm+\log\l(\frac{1}{\epsilon}\r)\r)> \max\{k , k_1\} + 10\l(tn_q+ tm+\log\l(\frac{1}{\epsilon}\r)\r)$.

Then, the following holds: 

 Let $\epsilon_{h+1} = \epsilon_{h}+c2^{\la}\epsilon$ for some constant $c$. With probability at least $1-\epsilon_{h+1}$ over the fixing of the random variables $\{ Q_i : i \in [h+1]\}, \{ Q_{i}^{(j)} : i \in [h+1], j \in [t] \},   \{ R_{i,1}, R_{i,2} : i \in [h+1] \}, \{ R^{(j)}_{i,1}, R^{(j)}_{i,2} : i \in [h+1], j \in [t]   \}, \{ \overline{Q}_i : i \in [h+1] \}, \{ \overline{Q}_{i}^{(j)}: i \in [h+1], j \in [t] \}, \{ \overline{R}_{i,1},\overline{R}_{i,2} : i \in [h] \}, \{ \overline{R}^{(j)}_{i,1},\overline{R}^{(j)}_{i,2} : i \in [h+1], j \in [t]  \}, \{ Q_{i+1}^{(j)} : j \in \ind_{[h+1]} \}$: (a) $Q_{h+2}$ is $\epsilon_{h+1}$-close to $U_{n_q}$ and is a deterministic function of $Y$ (b) $\{ Q_{h+2}^{(j)}: j \in \overline{\ind}_{[h+1]}\}$ is a deterministic function of $Y,\{Y^{(j)} : j \in [t] \}$ (c) The random variables $(X,\{ X^{(j)}: j \in [t]\})$ and $(Y,\{ Y^{(j)}: j \in [t]\})$  are independent (d) $X$ has min-entropy at least $k_w - 10 (h+1) \l( tm+\log\l(\frac{1}{\epsilon}\r)\r)$ and $Y$ has min-entropy at least $k_y-10(h+1)\l( tm+\log\l(\frac{1}{\epsilon}\r)\r)  $. 
\end{claim}
\begin{proof} We fix the random variables $\{ Q_i : i \in [h]\}, \{ Q_{i}^{(j)} : i \in [h], j \in [t] \},  \{ R_{i,1}, R_{i,2} : i \in [h] \}, \{ R^{(j)}_{i,1}, R^{(j)}_{i,2} : i \in [h], j \in [t]   \}, \{ \overline{Q}_i : i \in [h] \}, \{ \overline{Q}_{i}^{(j)}: i \in [h], j \in [t] \}, \{ \overline{R}_{i,1},\overline{R}_{i,2} : i \in [h] \}, \{ \overline{R}^{(j)}_{i,1},\overline{R}^{(j)}_{i,2} : i \in [h], j \in [t]  \}, \{ Q_{i+1}^{(j)} : j \in \ind_{[h]} \}$ such that (a), (b), (c), (d) holds (this happens with probability at least $1-\epsilon_{h}$. We also fix the random variables $\{ R_{h+1,\psi_1(z^{(j)}_{h+1})}^{(j)}: j \in \ind_{[h]}\}$, noting that they are deterministic functions of $X$. Thus $X$ has min-entropy at least $k_w - 10 h \l( jm+\log\l(\frac{1}{\epsilon}\r)\r) - tm- \log\l(\frac{1}{\epsilon}\r)$ with probabilitiy at least $1-\epsilon$. Further, $Q$ has min-entropy at least $k_y-10h\l(tn_q + tm+\log\l(\frac{1}{\epsilon}\r)\r)$.
The claim now follows directly from Lemma $\ref{base:lemma}$.
\end{proof}
To complete the proof of Lemma $\ref{main_lemma_1}$, we now note that the hypothesis of Claim $\ref{main_claim_1}$ is indeed satisfied when $h=0$. Thus, by  $\ell$ applications of Claim $\ref{main_claim_1}$, it follows that the $Q_{\ell+1}$ is $\epsilon^{\prime}_{\ell}$-close to $U_{n_q}$, where $\epsilon_{\ell}^{\prime} = O(2^{\la}\epsilon+\ell \epsilon )$. This follows since for all applications of Claim $\ref{main_claim_1}$ except the first time, $Q_{h}$ is $\epsilon_h$-close to uniform, and hence the parameter $\la=0$. This concludes the proof of Lemma $\ref{main_lemma_1}$. 
\end{proof}
\subsection{An Explict Seedless $(2,t)$-Non-Malleable Extractor Construction}
We are now ready to present our construction. We first set up the various ingredients developed so far  with appropriate parameters.

\textbf{Subroutines and Parameters}
\begin{enumerate}
\item Let $\gamma$ be  a small enough constant and $C$ a large one. Let $t=n^{\gamma/C}$.
\item Let $n_1 = n^{\beta_1}$, $\beta_1 = 10 \gamma$. Let $\IP:\{ 0,1\}^{n_1} \times \{ 0,1\}^{n_1} \rightarrow \{ 0,1\}^{n_2}$, $n_2 = \frac{n_1}{10}$, be the strong two-source extractor from Theorem $\ref{strong_ip}$.
\item Let $\mathcal{C}$ be an explicit $[\frac{n}{\alpha},n,\frac{1}{10}]$-binary linear error correcting code  with encoder $E: \{ 0,1\}^{n} \rightarrow \{0,1 \}^{\frac{n}{\alpha}}$. Such explicit codes are known, for example from the work of Alon et al. \cite{ABNNR92}.
\item Let $\samp:\{ 0,1\}^{n_2} \rightarrow \l[\frac{n}{\alpha}\r]$ be the sampler from Corollary $\ref{cor:samp}$ with parameters $\delta_{\samp} = \frac{1}{10}$ and $\nu_{\samp} = \beta_1$. Let the number of samples  $t_{\samp} = n^{\beta_2}$. Thus, $\beta_2 \le \beta_1$.
\item Let $\ell = 2(n^{\beta_1} + n^{\beta_2})$. Thus $\ell \le n^{11\gamma}$.
\item  We set up the parameters for the components used by $2\laExt$ (computed by Algorithm $1$) as follows.
\begin{enumerate}
\item Let $n_3 = n^{\beta_3}, n_4= n^{\beta_4}$, with $\beta_3=100 \gamma$ and $\beta_4=50 \gamma$.

Let $\Ext_q: \{0,1 \}^{n_3} \times \{ 0,1\}^{n_4} \rightarrow \{ 0,1\}^{n_4}$ be the strong seeded linear extractor from Theorem $\ref{thm:lsext}$ set to extract from min-entropy $k_q=\frac{n_3}{4}$ with error $\epsilon = 2^{-\Omega(n^{\gamma_q})}$, $\gamma_q = \frac{\beta_4}{2}$. Thus, by Theorem \ref{thm:lsext}, we have that the seed length $d_q = O\l(\frac{\log^2(n_3/\epsilon)}{\log(k_q/n_4)}\r) =  O(n^{2\gamma_q}) = n_4$.

 Let $\Ext_w: \{0,1 \}^{n} \times \{ 0,1\}^{n_4} \rightarrow \{ 0,1\}^{n_4}$ be the strong linear seeded  extractor  from Theorem $\ref{thm:lsext}$ set to extract from min-entropy  $k_w=\frac{n}{2}$ with error $\epsilon = 2^{-\Omega(n^{\gamma_q})}$.
  
 \item Let $\laExt:  \{0,1 \}^{n} \times \{ 0,1\}^{n_3}\rightarrow \{ 0,1\}^{2n_4}$  be the look ahead extractor used by $2\laExt$ (recall that the  parameters in the alternating extraction protocol are set as $m=n_4,u=2$ where $u$ is the number of steps in the protocol, $m$ is the length of each random variable that is communicated between the players, and $\Ext_q,\Ext_w$ are the strong seeded extractors used in the protocol.).
 
 \item Let $\Ext: \{ 0,1\}^{n} \times \{0,1\}^{n_4} \rightarrow \{ 0,1\}^{n_3}$ be the linear strong seeded extractor from Theorem $\ref{thm:lsext}$ set to extract from min-entropy $\frac{n}{2}$ with seed length $n_4$ and error $2^{-\Omega(n^{\beta_4/2})}$.
\end{enumerate}
\item Let $\nmExt_1$ be the function computed by Algorithm $2$, which uses the function $2\laExt$ set up as above.  
\end{enumerate}
\RestyleAlgo{boxruled}
\LinesNumbered
\begin{algorithm}[ht]
  \caption{nmExt(x,y)\label{alg}  \vspace{0.1cm}\newline \textbf{Input:} Bit strings $x,y$, each of length $n$. \newline \textbf{Output:} A bit string of length $n_4$.}
 Let $x_1 = \slice(x,n_1)$, $y_1=\slice(y,n_1)$. Compute $v = \IP(x,y)$.
 
 Compute $T = \samp(v) \subset [\frac{n}{\alpha}]$. 
 
 Let $z= x_1\circ x_2 \circ y_1\circ y_2$ where $x_2= (E(x))_{\{T\}}, y_2= (E(y))_{\{T\}}$.
 
 Output $\nmExt_1(x,y,z)$.
 \end{algorithm}
 
 We now state our main theorem.
 \begin{thm}\label{seedless_main} Let $\nmExt$ be the function computed by Algorithm $3$. Then $\nmExt$  is a seedless $(2,t)$-non-malleable extractor with error $2^{-n^{\Omega(1)}}$.
 \end{thm}
We establish the following two lemmas, from which the above theorem is direct. 

 \begin{lemma}\label{main_lemma} $\nmExt:\{ 0,1\}^{n} \times \{ 0,1\}^{n} \rightarrow \{ 0,1\}^{n_4}$ satisfies the following property $\mathcal{P}_n$: If $X,Y$ are independent $(n,n-n^{\gamma})$-sources and $\A_1=(f_1,g_1),\ldots,\A_t =(f_t,g_t)$ are arbitrary $2$-split-state tampering functions, such that for each $i \in [t]$, at least one of $f_i,g_i$ has no fixed points, then the following holds:\begin{align*} 
|\nmExt(X,Y), \nmExt(\A_1(X,Y)), \ldots, \nmExt(\A_t(X,Y)) - \\ U_{n_4}, \nmExt(\A_1(X,Y)), \ldots, \nmExt(\A_t(X,Y))| \le \epsilon,
\end{align*}
where $\epsilon =2^{-n^{\Omega(1)}}$.
 \end{lemma}

\begin{lemma}\label{lemma:final_convex} Suppose $\nmExt:\{ 0,1\}^{n} \times \{ 0,1\}^{n} \rightarrow \{ 0,1\}^{n_4}$,  satisfies  property $\mathcal{P}_n$ (from Lemma $\ref{main_lemma}$).
Then, $\nmExt$ is a seedless $(2,t)$-non-malleable extractor with error $(2^{-n^{\gamma}}+\epsilon)2^{2t}$.
 \end{lemma}
 \textbf{Notation:} For any function $H$, if $V = H(X,Y)$, let $V^{(i)}$  denote the random variable $H(\A_i(X,Y))$. 

\begin{proof}[Proof of Lemma $\ref{main_lemma}$]
We begin by proving the following claim.
\begin{claim}\label{lemma:notequal} With probability at least $1-2^{-n^{\Omega(1)}}$, $Z \neq Z^{(i)}$ for each $i \in [t]$.
\end{claim}
\begin{proof} Pick an arbitrary $i \in [t]$. Without loss of generality, suppose $f_i$ has no fixed points. If $X_1 \neq X_1^{(i)}$ or $Y_1 \neq Y_1^{(i)}$, then $Z \neq Z^{(i)}$. Now suppose $X_1 = X_1^{(i)}$ and $Y_1 = Y_1^{(i)}$. We fix $X_1$, and note that since $\IP$ is a strong extractor (Theorem $\ref{strong_ip}$), $V$ is $2^{-\Omega(n_1)}$-close to $U_{n_2}$  after this fixing (with probability at least $1-2^{-\Omega(n_1)}$). Also note that $V=V^{(i)}$.

Since  $f_i$ has no fixed points, it follows that since $E$ is an encoder of a code with relative distance distance $\frac{1}{10}$, $\Delta(E(X),E(X^{(i)})) \ge \frac{n}{10 \alpha}$. Let $D = \{ j \in \l[\frac{n}{\alpha}\r]: E(X)_{\{j\}} \neq E(X^{(i)})_{\{j\}} \}$. Thus $|D| \ge \frac{n}{10 \alpha}$.   Using  Corollary $\ref{cor:samp}$, it follows that with probability at least $1-2^{-\Omega(n_1)}$, $| D \cap \samp(V)| \ge 1$, and thus  $X_2 \neq X_2^{(i)}$ (since $\samp(V)=\samp(V^{(i)})$). This proves the claim.
\end{proof}
We fix $Z,Z^{(1)},\ldots,Z^{(t)}$ such that $Z \neq Z^{(i)}$ for any $i \in [t]$ (from the lemma above, this occurs with probability $1-2^{-n^{\Omega(1)}}$). We note that by the Lemma $\ref{lemma:notequal}$ and Lemma $\ref{lemma:entropy_loss_1}$, each of the sources $X$ and $Y$  still has  min-entropy at least $n-n^{\gamma}-(t+1)\ell-n^{\gamma/10}>n-n^{12 \gamma}$ with probability at least $1-2^{-n^{\gamma/10}}$. 

Lemma $\ref{main_lemma}$ now follows directly from Lemma $\ref{main_lemma_1}$ by noting that the following hold by our choice of parameters:
\begin{itemize}
\item $n-n^{12\gamma}> \frac{n}{2} + 20(n^{\beta_1} + n^{\beta_2})(n^{\gamma/C}(n^{\beta_3}+n^{\beta_4})+ n^{\beta_4})$
\item $n^{\beta_3} > \frac{4}{3}(20tn^{\beta_4}+n^{12\gamma})$

\item $2^{n^{12\gamma}}2^{-\Omega(n^{\beta_4/2})}<2^{-\Omega(n^{\beta_4/4})}$.
\end{itemize}
This concludes the proof.
\end{proof}

  \begin{proof}[Proof of Lemma $\ref{lemma:final_convex}$]Let $\A_1=(f_1,g_1),\ldots,\A_t=(f_t,g_t)$ be arbitrary $2$-split-state adversaries. We partition $\{ 0,1\}^{n}$ in two different ways based on the fixed points of the tampering functions. 
 
 For any $R \subseteq [t]$, define $$  W^{(R)}= \{ x \in \{0,1\}^n: \text{$f_i(x)=x$ if $i \in R$, and $f_i(x) \neq x$ if $i \in [t]\setminus R$}\}.$$
 
 Similary, for any $S \subseteq [t]$, define $$ V^{(S)}  = \{ y \in \{0,1\}^n: \text{$g_i(y)=y$ if $i \in S$, and $g_i(y) \neq y$ if $i \in [t]\setminus S$}\}.$$
 Thus the sets $W^{(R)}, R \subseteq [t]$ defines a partition of $\{0,1 \}^n$. Similarly $V^{(S)}, S \subseteq [t]$ defines a partition of $\{ 0,1\}^n$. For $R,S \subseteq [t]$, let $X^{(R)}$ be a random variable uniform on $W^{(R)}$, and $Y^{(S)}$ be a random variable uniform on $V^{(S)}$. 
 
 Let $U_{n_4}$ be uniform on $\{ 0,1\}^{n_4}$ and independent of $X^{R},Y^{S}$, for all $R,S \subseteq [t]$.  
 
 Define $$D_{\vec{f},\vec{g}}^{(R,S)} = (U_{n_4},Z_{1}^{(R,S)},\ldots,Z_{t}^{(R,S)}) $$ where we define the random variable 
 
 \[
 Z_i^{(R,S)} =
  \begin{cases}
   \nmExt(f_i(X^{(R)}),g_i(Y^{(S)}))  & \text{if }  i \in [t] \setminus(R \cap S) \\
   \same       & \text{if }i \in R \cap S
  \end{cases}
\]

 Define the distribution: $$D_{\vec{f},\vec{g}} = \sum_{R,S}\alpha_{R,S} D_{\vec{f},\vec{g}}^{(R,S)}$$, where $\alpha_{R,S} = \frac{|W^{(R,S)}||V^{(R,S)|}}{2^{2n}}$.

We first prove the following claim.
\begin{claim} Let 
\begin{align*}
\Delta_{R,S} = \alpha_{R,S} |\nmExt(X^{(R)},Y^{(S)}),\nmExt(f_1(X^{(R)}),g_1(Y^{(S)})),\ldots, \\ \nmExt(f_t(X^{(R)}),g_t(Y^{(S)})) - D_{\vec{f},\vec{g}}^{(R,S)} |. 
\end{align*}
Then, for every $R,S \subseteq [t]$, $\Delta_{R,S} \le 2^{-n^{\gamma}} + \epsilon$.
\end{claim} 
\begin{proof} If $|W^{(R)}| \le 2^{n-n^{\gamma}}$, it follows that $\alpha_{R,S} \le 2^{-n^{\gamma}} $, and hence the claim follows. Thus, assume that $H_{\infty}(X^{(R)}) \ge n-n^{\gamma}$. Using a similar argument, we can assume that $H_{\infty}(Y^{(S)}) \ge n-n^{\gamma}$. 

Let $\overline{R \cap S}= [t]\setminus (R \cap S)  = \{ i_1,\ldots,i_j\}$. It follows that for any $c \in \overline{R \cap S}$, at least one the following is true: $(1)$ $f_{c}$ has no fixed points on $W^{(R)}$ $(2)$ $g_{c}$ has no fixed points on $V^{(S)}$.  Thus, invoking Lemma $\ref{main_lemma}$, we have
\begin{align*} |\nmExt(X^{(R)},Y^{(S)}),\nmExt(f_{i_1}(X^{(R)}),g_{i_1}(Y^{(S)})),\ldots, \nmExt(f_{i_j}(X^{(R)}),g_{i_j}(Y^{(S)}))   \\ - U_{n_4},\nmExt(f_{i_1}(X^{(R)}),g_{i_1}(Y^{(S)})),\ldots,  \nmExt(f_{i_j}(X^{(R)}),g_{i_j}(Y^{(S)})) | \le  \epsilon
\end{align*}
The claim now follows by observing that for each $c \in R \cap S$, $f_{c}$ and $g_{c}$ are the identity functions on the sets $W^{(R)}$ and $V^{(S)}$ respectively.
\end{proof}
Let $X,Y$ be independent and uniformly random on $\{ 0,1\}^{n}$. Thus, we have 
\begin{align*}
|\nmExt(X,Y), \nmExt(\A_1(X,Y)), \ldots, \nmExt(\A_t(X,Y))  \\ - U_{n_4}, \cpy^{(t)}(D_{\vec{f},\vec{g}},U_{n_4})| =  \sum_{R,S \subseteq [t]} \Delta_{R,S}  \le 2^{2t}(\epsilon+2^{-n^{\gamma}}).
\end{align*}
Thus $\nmExt$ is a $(2,t)$-non-malleable extractor with error $(\epsilon+2^{-n^{\gamma}})2^{2t}$.
 \end{proof}
 
 \section{An Explict Seeded Non-Malleable Extractor at Polylogarithmic Min-entropy}\label{seeded_nm}
 
 \textbf{Subroutines and Parameters}
\begin{enumerate}
\item Let $\gamma$ be  a small enough constant and $C$ a large one. Let $t,k,d$ be parameters such that $t\le k^{\gamma/2}$.
\item Let $n_1 =  \log\l(\frac{tn}{\epsilon} \r)$. Let $\Ext_s:\{ 0,1\}^{n} \times \{ 0,1\}^{n_1} \rightarrow \{ 0,1\}^{n_1}$ be the strong seeded extractor from Theorem $\ref{guv}$ set to extract from min-entropy $2n_1$ and error $2^{-\Omega(n_1)}$.
\item Let $\mathcal{C}$ be an explicit $[\frac{d}{\alpha},d,\frac{1}{10}]$-binary linear error correcting code  with encoder $E: \{ 0,1\}^{d} \rightarrow \{0,1 \}^{\frac{d}{\alpha}}$. Such explicit codes are known, for example from the work of Alon et al. \cite{ABNNR92}.
\item Let $\Ext_{\samp}: \{0,1 \}^{n_1} \times \{0,1\}^{d_1}\rightarrow\{ 0,1\}^{n_2}$ be the strong seeded extractor from Theorem $\ref{zuck7}$ set to extract from min-entropy $\frac{n_1}{2}$ with error $\frac{1}{20}$ and output length $n_2$, such that $N_2D_1= \frac{d}{\alpha}$, where $N_2=2^{n_2}$ and $D_1 =2^{d_1}$. Let $\{ 0,1\}^{d_1} = \{ s_1,\ldots,s_{D_1}\}$. Define $\samp:\{ 0,1\}^{n_1}\rightarrow [\frac{d}{\alpha}]^{D_1}$ as: $\samp(x) = \l(\Ext(x,s_1) \circ s_1,\ldots,\Ext(x,s_{D_1}) \circ s_{D_1} \r) $. By Theorem $\ref{zuck7}$, we have $D_1=c_1n_1$, for some constant $c_1$.
\item Let $\ell = n_1+D_1=(c_1+1)n_1$.
\item  We set up the parameters for the components used by $2\laExt$ (computed by Algorithm $1$) as follows.
\begin{enumerate}
\item Let $n_3 = c_3t\ell, n_4= 10\ell$, for some large enough constant $c_3$.

Let $\Ext_q: \{0,1 \}^{n_3} \times \{ 0,1\}^{n_4} \rightarrow \{ 0,1\}^{n_4}$ be the strong seeded  extractor from Theorem $\ref{guv}$ set to extract from min-entropy $k_q=\frac{n_3}{4}$ with error $\epsilon = 2^{-\Omega(n_4)}$. 

 Let $\Ext_w: \{0,1 \}^{n} \times \{ 0,1\}^{n_4} \rightarrow \{ 0,1\}^{n_4}$ be the strong seeded  extractor  from Theorem $\ref{guv}$ set to extract from min-entropy $\frac{k}{2}$ with error $\epsilon = 2^{-\Omega(n_4)}$.
 
 \item Let $\laExt:  \{0,1 \}^{n} \times \{ 0,1\}^{n_3+n_4}\rightarrow \{ 0,1\}^{2n_4}$  be the look ahead extractor used by $2\laExt$. Recall that the  parameters in the alternating extraction protocol are set as $m=n_4,u=2$ where $u$ is the number of steps in the protocol, $m$ is the length of each random variable that is communicated between the players, and $\Ext_q,\Ext_w$ are the strong seeded extractors used in the protocol.
 
 \item Let $\Ext: \{ 0,1\}^{d} \times \{0,1\}^{n_4} \rightarrow \{ 0,1\}^{n_3}$ be the strong seeded extractor from Theorem $\ref{guv}$ set to extract from min-entropy $\frac{d}{2}$ with seed length $n_4$ and error $2^{-\Omega(n_4)}$.
 
\end{enumerate}
\item Let $\nmExt_1$ be the function computed by Algorithm $2$, which uses the function $2\laExt$ set up as above.  

 \item Let $n_5 = \frac{k}{100t}$. Let $\Ext_1: \{ 0,1\}^{n} \times \{0,1\}^{n_4} \rightarrow \{ 0,1\}^{n_5}$ be the strong seeded extractor from Theorem $\ref{guv}$ set to extract from min-entropy $\frac{k}{4}$ with seed length $n_4$, error $2^{-\Omega(n_4)}$.

\end{enumerate}

\RestyleAlgo{boxruled}
\LinesNumbered
\begin{algorithm}[ht]
  \caption{snmExt(x,y)\label{alg}  \vspace{0.1cm}\newline \textbf{Input:} Bit strings $x,y$,  of length $n,d$ respectively. \newline \textbf{Output:} A bit string of length $n_4$.}
 $y_1=\slice(y,n_1)$. Compute $v = \Ext_s(x,y_1)$.
 
 Compute $T = \samp(v) \subset [\frac{n}{\alpha}]$. 
 
 Let $z= y_1 \circ y_2$ where $y_2= (E(y))_{\{T\}}$.
 
 Output $\Ext_1(x,\nmExt_1(x,y,z))$.
 \end{algorithm}
 
 We now state our main theorem.
 
 \begin{thm}\label{main_thm_seeded} Let $\snmExt:\{ 0,1\}^{n} \times \{ 0,1\}^{d} \rightarrow \{ 0,1\}^{n_5}$ be the function computed by Algorithm $4$. Then $\snmExt$ satisfies the following property:  For any $\epsilon>0$, $k \ge C \log^{2+ \gamma}\l(\frac{n}{\epsilon}\r)$, $t\le k^{\gamma/2}$ and $d \ge C  t ^2 \log^{2}\l(\frac{n}{\epsilon}\r)$, if $X$ is a $(n,k)$-source, and $Y$ is an independent and uniform distribution on $\{ 0,1\}^d$, and $\A_1\ldots,\A_t$ are arbitrary tampering functions, such that for each $i \in [t]$, $\A_i$ has no fixed points, then the following holds:\begin{align*} 
|\snmExt(X,Y), \snmExt(X,\A_1(Y)), \ldots, \snmExt(X,\A_t(Y)),Y - \\ U_{n_5}, \snmExt(X,\A_1(Y)), \ldots, \snmExt(X,\A_t(Y)),Y| \le O(\epsilon),
\end{align*}
 \end{thm}
 \textbf{Notation:} For any function $H$, if $V = H(X,Y)$, let $V^{(i)}$  denote the random variable $H(X,\A_i(Y))$. 

 \begin{proof}
We first prove  the following claim.
\begin{claim}\label{lemma:snotequal} With probability at least $1-\epsilon$, $Z \neq Z^{(i)}$ for each $i \in [t]$.
\end{claim}
\begin{proof} Pick an arbitrary $i \in [t]$.  If $Y_1 \neq Y_1^{(i)}$, then we have $Z \neq Z^{(i)}$. Now suppose  $Y_1 = Y_1^{(i)}$. We fix $Y_1$, and note that since $\Ext_s$ is a strong extractor (Theorem $\ref{strong_ip}$), $B$ is $2^{-\Omega(n_1)}$-close to $U_{n_1}$.

Since  $\A_i$ has no fixed points, it follows that since $E$ is an encoder of a code with relative distance distance $\frac{1}{10}$, $\Delta(E(Y),E(Y^{(i)})) \ge \frac{d}{10 \alpha}$. Let $\D = \{ j \in \l[\frac{d}{\alpha}\r]: E(Y)_{\{j\}} \neq E(Y^{(i)})_{\{j\}} \}$. Thus $|\D| \ge \frac{d}{10 \alpha}$.   Using  Theorem $\ref{thm:seed_samp}$, it follows that with probability at least $1-\epsilon$, $| \D \cap \samp(V)| \ge 1$, and thus  $Y_2 \neq Y_2^{(i)}$ (since $\samp(V)=\samp(V^{(i)})$). The claim now follows by a simple union bound.
\end{proof}
We fix $Z,Z^{(1)},\ldots,Z^{(t)}$ such that $Z \neq Z^{(i)}$ for any $i \in [t]$ (from the lemma above, this occurs with probability $1-\epsilon$). We note that by the Lemma $\ref{lemma:snotequal}$ and Lemma $\ref{lemma:entropy_loss_1}$, the source $X$  has  min-entropy at least $k-2n_1$   and the source $Y$ has min-entropy at least $d-2\ell$ with probability at least $1-\epsilon$.

Lemma $\ref{main_lemma}$ now follows directly from Lemma $\ref{main_lemma_1}$ by noting that the following hold by our choice of parameters:
\begin{itemize}
\item $\frac{d}{2} > 20 \ell (t(n_3 +n_4) + \log(\frac{1}{\epsilon}) )$
\item $k-2n_1 \ge \frac{n_3}{4} + 20 \ell (tn_4 + \log(\frac{1}{\epsilon})  )$
\item $n_3-2n_1 \ge \frac{4}{3}(10tn_4 + 2\log(\frac{1}{\epsilon}))$
\end{itemize}
This concludes the proof.
\end{proof}

\section{Efficient Encoding and Decoding Algorithms for One-Many Non-Malleable Codes} \label{efficiency}
In this section, we construct efficient algorithms for almost uniformly sampling from the pre-image of any output of a modified version of the  $(2,t)$-non-malleable extractor constructed in Section $\ref{section:proof}$. Combining this with Theorem $\ref{connection_t}$ and Theorem $\ref{seedless_main}$ gives us efficient constructions of one-many non-malleable codes in the $2$-split state model, with tampering degree $t=n^{\Omega(1)}$, relative rate $n^{\Omega(1)}/n$ and error $2^{-n^{\Omega(1)}}$.

A major part of this section is on modifying the components used in the construction of $\nmExt$ (Algorithm $3$)  so that the overall extractor is much simpler to analyze as a function, and this enables us to develop efficient sampling algorithms from the pre-image. We present the modified extractor construction in Section $\ref{modified}$. However, we first need to solve a simpler problem.

\subsection{A New Linear Seeded Extractor}
A crucial sub-problem that we have to solve is almost uniformly sampling from the pre-image of a linear seeded extractor in polynomial time. Towards this, we recall a well known property of linear seeded  extractors.

\begin{lemma}[\cite{Rao09}\label{aff_error}] Let $\Ext:\{ 0,1\}^{n} \times \{ 0,1\}^{d} \rightarrow \{ 0,1\}^{m}$ be a linear seeded extractor for min-entropy $k$ with error $\epsilon<\frac{1}{2}$. Let $X$ be an affine $(n,k)$-source. Then $$\Pr_{u \sim U_{d}}[|\Ext(X,u) - U_{m}|>0] \le \epsilon. $$
\end{lemma}\qed

\begin{define}For any seeded extractor $\Ext:\{0,1\}^{n} \times \{ 0,1\}^d \rightarrow \{ 0,1\}^m$, any $s \in \{ 0,1\}^d$ and $r \in \{ 0,1\}^m$, we define:
\begin{itemize}
\item $\Ext(\cdot,s):\{ 0,1\}^n \rightarrow \{ 0,1\}^m$ to be the map $\Ext(\cdot,s)(x) = \Ext(x,s)$.
\item  $\Ext^{-1}(r)$ to be the set $\{ (x,y) \in \{ 0,1\}^n \times \{ 0,1\}^d: \Ext(x,y) = r\}$.
\item $\Ext^{-1}(\cdot,s)$ to be the set $\{ x : \Ext(x,s) = r\}$.
\end{itemize}
\end{define}
We now present a natural way of sampling from pre-images of linear seeded extractors.
\begin{claim}Let $\Ext:\{ 0,1\}^{n} \times \{ 0,1\}^{d} \rightarrow \{ 0,1\}^{m}$ be a linear seeded extractor for min-entropy $k$ with error $\epsilon<2^{-1.5m}$. For any $r \in \{ 0,1\}^{m}$, consider the following efficient sampling procedure $\S$ which on input $r$ does the following: (a) Sample $s \sim U_d$, (b) sample $x$ uniformly from the subspace $\Ext(\cdot,s)^{-1}(r)$. (c) Output $(x,s)$. Let $\D_{r}$ be the distribution uniform on $\Ext^{-1}(r)$, and let $\S(r)$ denote the distribution produced by $\S$ on input $r$. 

Then, $$ |  \S(r) - \D_{r}  | \le 2^{-\Omega(m)}$$
\end{claim}
\begin{proof} Define the sets: $$ Good = \{ s \in \{ 0,1\}^d: \rank(\Ext(\cdot,s))=m \}, \hspace{0.3cm} Bad = \{ 0,1\}^{d}\setminus Good.$$ It follows by Lemma $\ref{aff_error}$ that $|Good| \ge (1-\epsilon)2^{d}$.  Thus, for any $s \in Good$, $|\Ext(\cdot,s)^{-1}(r)|=2^{n-m}$. Thus, we have $$\sum_{s \in Good} |\Ext^{-1}(\cdot,s)(r)| \ge 2^{d+n-m-1}.$$
Further, for any $s^{\prime} \in Bad$, $|\Ext^{-1}(\cdot,s^{\prime})(r)| \le 2^{n}$, and hence $$\sum_{s^{\prime} \in Bad}|\Ext^{-1}(\cdot,s^{\prime})(r)| \le \epsilon2^{d+n} < 2^{d+n-1.5m}. $$
Thus $|\cup_{s^{\prime} \in bad}\Ext^{-1}(\cdot,s^{\prime})(r)|<2^{-0.5m}|\Ext^{-1}(r)|$. It now follows  that $$ |  \S(r) - \D_{r}  | \le 2^{-0.4m}$$
\end{proof}

We note that $\epsilon$ must be $o(2^{-m})$ for the above sampling procedure to work with low enough error. However, this would require a seed length of  $d = O(m^2)$ (by Theorem $\ref{thm:lsext}$). For each step of the alternating extraction protocol the seed length then goes down by a quadratic factor,  which is insufficient for our application.

To get past this difficulty, we construct a new strong linear seeded extractor for high min-entropy sources with the seed length close to the output length with the property that the size of the pre-image of any output is the same for any fixing of the seed. Algorithm $5$ provides this  construction.

 \textbf{Parameters and Subroutines:}
 \begin{enumerate}
 \item Let $\delta>0$ be any constant. Let $d = n^{\delta}$. Let $d=d_1 + d_2$, where $d_1 = n^{\delta_1}$, $\delta>10 \delta_1$. Let $m=d/2$. 
\item  Let $\samp:\{ 0,1\}^{d_1} \rightarrow [n]^t$, $t=d_2$, be an $(\mu,\theta,\gamma)$ averaging sampler with distinct samples, such that $\mu = \frac{(\delta - 2\tau)}{\log(1/\tau)}$, $\theta = \frac{\tau}{\log(1/\tau)}$ and $\gamma = 2^{-\Omega(d_1)}$, $\tau=0.05$.
  \item Let $\IP: \{0,1 \}^{d_2} \times \{ 0,1\}^{d_2} \rightarrow \{ 0,1\}^{\frac{d}{2}}$ be the strong $2$-source  extractor from Theorem $\ref{strong_ip}$.
 \end{enumerate}
\RestyleAlgo{boxruled}
\LinesNumbered
\begin{algorithm}[ht]
  \caption{iExt(x,s)\label{alg}  \vspace{0.1cm}\newline \textbf{Input:} Bit strings $x,s$ of length $n,d$ respectively.  \newline \textbf{Output:} A bit string of length $m$.
 }
Let $s_1 = \slice(s,d_1)$. Let $s_2$ be the remaining $d_2$ bits of $s$.

Let $T = \samp(s_1) \subset [n]$. Let $x_1 = x_{\{ T\}}$.

Output $\IP(x_1,s_2)$. 
\end{algorithm}

Informally the construction of $\iExt$ is as follows. Given a uniform seed $S$, we use a slice $S_1$ of $S$ to sample co-ordinates from the weak source $X$, and then apply a strong $2$-source extractor (based on the inner product function) to the source $X_1$ (which is the projection of $X$ to the sampled co-ordinates) and the remaining bits $S_2$ of $S$  to extract $\frac{d}{2}$ uniform bits. 

The correctness of this procedure relies on the fact that  by pseudorandomly sampling co-ordinates of $X$ and projecting $X$ to these co-ordinates, the min-entropy rate is roughly the preserved for most choices of the uniform seed \cite{Z97} \cite{Vad04} \cite{Li12a}. Thus, we can fix $S_1$, and the strong two-source extractor $\IP$ now receives two independent inputs $S_2$ and $X_2$ with almost full min-entropy. Thus, the output is close to uniform. Further we show that the number of linear constraints on the source $X$ is the same for any fixing of the seed. This allows us to show that size of the pre-image of any particular output is the same for any fixing of the seed. We now formally prove these ideas.

We need the following theorem proved by Vadhan \cite{Vad04}.
\begin{thm}[\cite{Vad04}]\label{samp:vad} Let $1\ge \delta\ge 3 \tau >0$. Let $\samp : \{0,1\}^{r} \rightarrow [n]^t$ be an $(\mu,\theta,\gamma)$ averaging sampler with distinct samples, such that $\mu = \frac{(\delta - 2\tau)}{\log(1/\tau)}$ and $\theta = \frac{\tau}{\log(1/\tau)}$. If $X$  is a $(n,\delta n)$ source, then the  random variable $(U_r,X_{\{Samp(U_r)\}})$ is $(\gamma+2^{-\Omega(\tau n)})$-close to $(U_r , W )$ where for every $a \in \{0, 1\}^r$ , the random variable $W |U_r =a$ is a $(t, (\delta -3\tau )t)$-source.
\end{thm}

\begin{lemma}\label{new_goodext} Let $\iExt$ be the function computed by Algorithm $5$. If $X$ is a $(n,0.9n)$ source and $S$ is an independent uniform seed on $\{ 0,1\}^{d}$, then the following holds: $$ |\iExt(X,S),S - U_{m},S| < 2^{-n^{\Omega(1)}}.$$ 
Further for any $r \in \{ 0,1\}^{m}$ and any $s \in  \{ 0,1\}^{d}$, $| \iExt(\cdot,s)^{-1}(r)|= 2^{n-m}$.
\end{lemma}
\begin{proof} Using Theorem $\ref{samp:vad}$, it follows that $X_1$ is $2^{-n^{\Omega(1)}}$-close to a source with min-entropy at least $0.8n$ for any fixing of $S_1$. Further, we note that after fixing $S_1$, $S_2$ and $X_1$ are independent sources. We now think of $X_1,S_2$ as sources in $\{ 0,1\}^{d_2+1}$ by appending a $1$ to both the sources, so that $S_2 \neq \vec{0}$, and then apply the inner product map. This results in an entropy loss of only $1$. It now follows by Theorem $\ref{strong_ip}$ that  $$ |\iExt(X,S),S - U_{m},S| < 2^{-n^{\Omega(1)}}.$$ 

It is easy to see that for any fixing of the seed $S=s$, $\iExt(\cdot,s)$ is a linear map. Let $X$ be uniform on $n$ bits.  We note that for  any fixing of  $S_2=s_2$,  $X_1$ lies in a subspace of dimension $d_2-m$ over $F_2$. Further, the bits outside $T$ have no restrictions placed on them. Thus the size of $ \iExt(\cdot,s)^{-1}(r)$ is exactly $2^{d_2-m+n-d_2}=2^{n-m}$. This completes the proof of the lemma.
\end{proof}

Based on the above lemma, we construct an efficient procedure for sampling  uniformly from the pre-image of  the function $\iExt$. 
 \begin{claim}\label{samp_ext} Let $\iExt: \{ 0,1\}^n \times \{0,1 \}^d \rightarrow \{ 0,1\}^m$ be the function computed by Algorithm $5$. Then there exists a polynomial time algorithm $\samp_1$  that takes as input $r \in \{ 0,1\}^{m}$, and samples from a distribution that is uniform on $\iExt^{-1}(r)$.
 \end{claim}
 \begin{proof} It follows  by Lemma $\ref{new_goodext}$  that for any fixing of the seed $s$, the size of the set $\iExt(\cdot,s)^{-1}(r)$ is exactly $2^{n-m}$. Thus we can  use the following  strategy: (a) Sample $s \sim U_d$  (b) Sample $x$ uniformly random from the subspace $\iExt(\cdot,s)^{-1}(r)$ (c) Output $(x,s)$. It follows that each element in $\iExt^{-1}(r)$ is picked with probability exactly $\frac{1}{2^{d}}\cdot \frac{1}{2^{n-m}}$. Thus the output of our sampling procedure is indeed uniform on $\iExt^{-1}(r)$.
\end{proof}

\subsection{A Modified Construction of the Seedless $(2,t)$-Non-Malleable Extractor}\label{modified}
We first describe the high level ideas involved in modifying the construction of $\nmExt$ (Algorithm $3$), before presenting the formal construction. 
\begin{itemize}
\item We use the linear seeded extractor $\iExt$  (Algorithm $5$) for any seeded extractor used in the construction of $\nmExt$.
\item Next we divide the sources $X$ and $Y$ into blocks of size $n^{1-\delta}$ respectively for a small constant $\delta$. Since each of $X$ and $Y$ have almost full min-entropy, we now have two block sources, where each block has almost full min-entropy conditioned on the previous blocks. The idea is to use  new blocks of $X$ and $Y$ for each round of alternating extraction in $\nmExt$. 

To implement this however, we need some care. Recall that the alternating extraction protocol is run for two rounds between  either $X$ and $Q_{h}$, or $X$ and $\overline{Q}_{h}$ in the function $2\laExt$. The idea now is to run these two of alternating extraction by dividing $Q_{h}$ into two blocks, and using two new partitions of $X$  (each round being run by using a block from either $X$ or $Q_{h}$). Now to generate these $Q_{h}$'s, we use a $O(t)$  blocks of $Y$, and for each block apply the strong seeded extractor $\iExt$, using as seed the output of the alternating extraction from the previous step, and finally concatenate the outputs. This works because these $O(t)$  blocks form a block source, and using the same seed to extract from all the blocks  is a well known technique of extracting from block sources.

\item By appropriate setting of the lengths of the seeds in the alternating extraction, we  ensure that each block of $X$ and $Y$ still has min-entropy rate $1-o(1)$ even after fixing all the intermediate seeds, the random variables $Q_{h},\overline{Q}_{h}$ and their tampered versions. This can be ensured since each of these variables are of length at most $n^{\delta_1}$ for some small constant $\delta_1$, and the number of adversaries is also $n^{\Omega(1)}$).

\item The above modification is almost sufficient for us to successfully sample from the pre-image of any output. One final modification is to use a specific error correcting code (the Reed-Solomon code over a field of size $n+1$ with characteristic $2$) in the initial step of the construction, when we encode the sources and sample bits from it. We give some intuition as to why this is necessary. Since we are using linear seeded extractors in the alternating extraction, by fixing the seeds we impose linear restrictions on the blocks of $X$ and $Y$. Now, if we fix the output of the  initial sampling step (the random variable $Z$ in Algorithm $3$), we are imposing more linear constraints on the blocks (assuming we are using a linear code). Now, it is not clear if the constraints imposed by the linear seeded extractor is independent from the constraints imposed by $Z$, and thus for different fixings of the $Z$ and the seeds the size of the pre-image of any output of the non-malleable extractor may be different. 

To get past this difficulty, our idea is to first partition $X$ and $Y$ into slightly smaller blocks (which does not affect the correctness of the extractor) such that at least half of the blocks are unused by the alternating extraction steps. Now, we show that by using the Reed-Solomon code over $\F= \F_{2^{\log(n+1)}}$ to encode the sources, fixing $Z$ imposes linear constraints involving the variables from these unused blocks, and we show that this is sufficient to argue that it is linearly independent of the restrictions imposed by the alternating extraction part. We provide complete details of the sampling algorithms in Section \ref{invert_efficient}.
\end{itemize} 

We now proceed to  present the extractor construction.
Recall that if $Z_a,Z_{a+1},\ldots,Z_{b}$ are random variables, we use $Z_{[a,b]}$ to denote the random variable $Z_{a},\ldots,Z_b$. 

\textbf{Subroutines and Parameters (used by Algorithm $6$, Algorithm $7$, Algorithm $8$})
\begin{enumerate}
\item Let $\gamma$ be  a small enough constant and $C$ a large one. Let $t=n^{\gamma/C}$.
\item Let $n_1 = n^{\beta_1}$, $\beta_1 = 10 \gamma$. Let $n_2=n-n_1$. Let $\IP_1:\{ 0,1\}^{n_1} \times \{ 0,1\}^{n_1} \rightarrow \{ 0,1\}^{n_3}$, $n_3 = \frac{n_1}{10}$ be the strong two-source extractor from Theorem $\ref{strong_ip}$.

\item Let $\F$ be the finite field $\F_{2^{\log (n+1)}}$. Let $n_4 = \frac{n_2}{\log (n+1)}$. Let $\RS: \F^{n_4} \rightarrow \F^{n}$ be the Reed-Solomon code encoding $n_4$ symbols of $\F$ to $n$ symbols in  $\F$ (we  overload the use of $\RS$, using it to denote both the code and the encoder). Thus $\RS$ is a $[n,n_4,n-n_4+1]_{n}$ error correcting code.

\item Let $\samp:\{ 0,1\}^{n_3} \rightarrow [n]^{n_5}$  be a $(\mu,\frac{1}{10},2^{-n^{\Omega(1)}})$ averaging sampler with distinct samples. By using the strong seeded extractor from Theorem $\ref{guv}$, we can set $n_5=n^{\beta_2}$, $\beta_2 < \beta_1/2$.

\item Let $\ell = 2(n_1 + n_5\log n)<4n^{\beta_1}$. Thus $\ell \le n^{11\gamma}$.

\item Let $n_6 = 50Ct \ell$. Let $\IP_2 : \{ 0,1\}^{n_6} \times \{ 0,1\}^{n_6} \rightarrow \{ 0,1\}^{2n_q}$, $n_q = 10Ct \ell$,  be the strong two-source extractor from Theorem $\ref{strong_ip}$.

\item Let $n_7 = n- n_1 - n_6$. Let $n_x = \frac{n_7}{8\ell}$. Let $n_{y}= \frac{n_7}{16Ct \ell}$. Thus $n_x, n_y \ge n^{1-15\gamma}$.

\item Let $d_{1} = 80 \ell$.

\item Let $ \iExt_{1}: \{ 0,1\}^{n_x} \times \{ 0,1\}^{d_{1}} \rightarrow \{ 0,1\}^{d_2}$, $d_2 = 40 \ell$, be the extractor computed by Algorithm $5$. 

\item Let $ \iExt_{2}: \{ 0,1\}^{n_q} \times \{ 0,1\}^{d_{2}} \rightarrow \{ 0,1\}^{d_3}$, $d_3 = 20 \ell$, be the extractor computed by Algorithm $5$. 

\item Let $ \iExt_{3}: \{ 0,1\}^{n_{x}} \times \{ 0,1\}^{d_{3}} \rightarrow \{ 0,1\}^{d_4}$, $d_4=10 \ell$  be the extractor computed by Algorithm $5$. 

\item Let $ \iExt_{4}: \{ 0,1\}^{n_{y}} \times \{ 0,1\}^{d_{4}} \rightarrow \{ 0,1\}^{d_5}$, $d_5 = 5 \ell$,  be the extractor computed by Algorithm $5$. 

\item Let $\Ext: \{ 0,1\}^{4Ctn_y} \times \{ 0,1\}^{d_4} \rightarrow \{0,1 \}^{2n_q}$ be defined in the following way. Let $v_1,\ldots,v_{4t}$ be strings, each of length $n_y$. Define $\Ext(v_1 \circ \ldots \circ v_{4Ct},s) = \iExt_4(v_1,s) \circ \ldots \circ \iExt_4(v_{4Ct},s)$.
\end{enumerate}
\RestyleAlgo{boxruled}
\LinesNumbered
\begin{algorithm}[ht]
  \caption{inmExt(x,y)\label{alg}  \vspace{0.1cm}\newline \textbf{Input:} Bit strings $x,y$, each of length $n$. \newline \textbf{Output:} A bit string of length $m$.}
 Let $x_1 = \slice(x,n_1)$, $y_1=\slice(y,n_1)$. Compute $\nu = \IP_1(x,y)$.
 
 Let $x_2, y_2$ be $n_2$ length strings formed by cutting $x_1,y_1$ from $x,y$ respectively.
 
 Let $T = \samp(\nu) \subset [n]$. 
 
 Interpret $x_2,y_2 $ as elements in $\F^{n_4}$. 
 
 Let $\overline{x}_2= \RS(x_2), \overline{y}_2 = \RS(y_2)$. 
 
 Let $\overline{x}_1= (\overline{x}_2)_{\{T\}}, \overline{y}_1= (\overline{y}_2)_{\{T\}}$, interpreting $\overline{x}_2,\overline{y}_2 \in \F^n$. 
  
 Let $z= x_1\circ \overline{x}_1 \circ y_1 \circ \overline{y}_1$, where $z$ is interpreted as a binary string.
 
 Interpret $x_2,y_2 $ as binary strings. 
 
 Output $\inmExt_1(x_2,y_2,z)$.

\end{algorithm}

\RestyleAlgo{boxruled}
\LinesNumbered
\begin{algorithm}[ht]
  \caption{$\inmExt_1(x_2,y_2,z)$\label{alg}  \vspace{0.1cm}}
 Let $x_3 = \slice(x_2, n_6), y_3 = \slice(y_2,n_6)$. Let $w,v$ be the remaining parts of $x_2, y_2$ respectively.
 
 Let $ \IP_2(x_3,y_3) = (q_{1,1},q_{1,2})$, where each of $q_{1,1},q_{1,2}$ is of length $n_q$.  

Let $w_{1},\ldots,w_{8\ell}$ be an equal sized partition of the string $w$ into $8\ell$ stings. 
 
 Let $v_{1},\ldots,v_{16t\ell}$ be an equal sized partition of the string $v$ into $16Ct\ell$ stings. 

\For{$h=1$ to $\ell$} { 
		$(q_{h+1,1},q_{h+1,2}) = 2\ilaExt(v_{[8C(h-1)t+1, 8Cht]},w_{[4h-3,4h]},q_{h,1},q_{h,2},h,z_{\{ h\}})$
}
Ouput $(q_{\ell+1,1} ,q_{\ell+1,2})$.   

\end{algorithm}

\RestyleAlgo{boxruled}
\LinesNumbered
\begin{algorithm}[ht]
  \caption{2$\ilaExt(v_{[8C(h-1)t+1, 8Cht]},w_{[4h-3,4h]},q_{h,1},q_{h,2},h,b)$\label{alg}  \vspace{0.1cm}}
    
Let $s_{h,1} = \slice(q_{h,1},d_1)$, $r_{h,1} = \Ext_1(w_{4h-3},s_{h,1})$, $s_{h,2} = \Ext_2(q_{h,2},r_{h,1})$, $r_{h,2} = \Ext_3(w_{4h-2},s_{h,2})$.

\uIf{$b=0$} {

Let $r_{h} = \slice(r_{h,1},d_4)$.

}
\Else{ Let $r_h = r_{h,2}$
}
Let $\Ext(v_{[8C(h-1)t+1, 8(h-1)t+4Ct]},r_h)= (\overline{q}_{h,1},\overline{q}_{h,2})$, where both $\overline{q}_{h,1},\overline{q}_{h,2}$ are of length $n_q$.

Let $\overline{s}_{h,1} = \slice(\overline{q}_{h,1},d_1)$, $\overline{r}_{h,1} = \Ext_1(w_{4h-1},\overline{s}_{h,1})$, $\overline{s}_{h,2} = \Ext_2(\overline{q}_{h,2},\overline{r}_{h,1})$, $\overline{r}_{h,2} = \Ext_3(w_{4h},\overline{s}_{h,2})$.

\uIf{$b=0$} {

Let $\overline{r}_{h} = \overline{r}_{h,2}$.

}
\Else{ Let $\overline{r}_h = \slice(\overline{r}_{h,1},d_4)$.
}

Let $\Ext(v_{[8C(h-1)t+4Ct+1,8Cht]} ,r_h)= (q_{h+1,1},q_{h+1,2})$, where  both $q_{h+1,1},q_{h+1,2}$ are of length $n_q$.

Ouput $(q_{h+1,1} ,q_{h+1,2})$.
\end{algorithm}
\begin{thm}Let $\inmExt$ be the function computed by Algorithm $7$. Then $\inmExt$  is a seedless $(2,t)$-non-malleable extractor with error $2^{-n^{\Omega(1)}}$.
\end{thm}
 The proof of the above theorem is essentially the same as the proof provided in Section $\ref{section:proof}$,  and we do not repeat it. The correctness of $\inmExt$ follows directly from  the proof of Theorem $\ref{seedless_main}$, and the correctness of the extractor $\iExt$ (Lemma $\ref{new_goodext}$), the fact that by our choice of parameters each block of $X$ and $Y$ still has min-entropy rate  at least $0.9$ after appropriate conditioning of the intermediate random variables and their tampered versions, and the fact that using the $\RS$ in place of a binary error correcting code does not affect correctness of the procedure.

\subsection{Efficiently Sampling from the Pre-Image of inmExt}\label{invert_efficient}
Since the construction of the non-malleable extractor $\inmExt$ (Algorithm $6$, Algorithm $7$, Algorithm $8$) is composed of various sub-parts and sub-functions, we first argue about the invertibility of these parts and then show a way to compose these sampling procedure to sample almost uniformly from the pre-image of $\inmExt$. We refer to all the variables, sub-routines and notations introduced in these algorithms while developing the sampling procedures. Unless we state otherwise, by a subspace we mean a subspace over $\F_2$. 

We first show how to sample uniformly from the pre-image of  $2\ilaExt$ (Algorithm $8$), since it is a crucial sub-part of $\inmExt$. We have the following claim.
\begin{claim}\label{same_size} For any fixing of the variables $\{ s_{1,i},r_{1,i},\overline{s}_{1,i},\overline{r}_{1,i}: i \in \{1,2 \}\}$, and any $b \in \{0,1\}$ define the set: 
\begin{align*}
2\ilaExt^{-1}(q_{2,1},q_{2,2})= \{ (x_3,y_3,v_{[1, 4Ct]},w_{[1,4]}) \in \{ 0,1\}^{2n_6+4Ctn_y+4n_x} : \\ 2\ilaExt(v_{[1,4Ct]}, w_{[1,4]}, q_{1,1},q_{1,2},b) = (q_{2,1},q_{2,2})\}
\end{align*}
 There exists an efficient algorithm $\samp_{2}$ that takes as input $q_{2,1},q_{2,2},b,\{ s_{1,i},r_{1,i},\overline{s}_{1,i},\overline{r}_{1,i}: i \in \{1,2 \}\}$, and samples uniformly from $2\ilaExt^{-1}(q_{2,1},q_{2,2})$. 
 
 Further, the set $2\ilaExt^{-1}(q_{2,1},q_{2,2})$ is a subspace over $\F_2$ of dimension $d_1$, and its size   does not depend on the inputs to  $\samp_2$.
\end{claim}
\begin{proof} The general idea is that by fixing the seeds in the alternating extraction, each block  of $w$ takes values independent of the fixing of the other blocks of $w$ and the $q_{i,j}$'s, and similarly the $q_{i,j}$'s takes values independent of each other and the blocks of $w$. We now formally prove this intuition. 

Since, $s_{1,1}$ is a slice of $q_{1,1}$ it follows that $q_{1,1}$ is restricted to the subspace of size $2^{n_q-d_1}$. Since $r_{1,1}=\iExt_1(w_1,s_{1,1})$, it follows that $w_1$ is restricted to the set $\iExt_1(\cdot,s_{1,1})^{-1}(r_{1,1})$. Further, it follows by Lemma \ref{new_goodext} that this is a subspace of size $2^{n_x-d_2}$. Similar arguments show that $q_{1,2}$ is restricted to the subspace of dimension $2^{n_q-d_3}$, and $w_{2}$ is restricted to a subspace of dimension $2^{n_x-d_4}$. Further, we note that each of these variables have no correlation.

By repeating this argument for the next two rounds of alternating extraction, it follows that $\overline{q}_{1,1}$ is restricted to  a subspace of size $2^{n_q-d_1}$, $w_3$ is restricted to a subspace of size $2^{n_x-d_2}$, $\overline{q}_{1,2}$ is restricted to a subspace of size $2^{n_q-d_3}$,  and $w_4$ is restricted to  a subspace of size $2^{n_x-d_4}$. 

Further since $(q_{2,1},q_{2,2}) = \Ext(v_{[4Ct+1,8t]},r_{1})=\iExt_4(v_{4Ct+1},r_{1})\circ \ldots \circ \iExt_4(v_{8Ct},r_{1})$, it follows by an application of Lemma $\ref{new_goodext}$ that for any fixed  $q_{2,1}$, $v_{[4Ct+1,6t]}$  is restricted to   a subspace of size $2^{2Ct(n_y-d_5)}$. A similar argument shows that for any fixed   $q_{2,2}$,  $v_{[6Ct+1,8Ct]}$ is restricted to  a subspace of size $2^{2Ct(n_y-d_5)}$.

Finally, since $\IP_1(x_3,y_3)= (q_{1,1},q_{1,2})$, it follows that for any fixed  $x_3,q_{1,1},q_{1,2}$, the variable $y_3$ lies in a subspace of size $2^{n_6-\log(2n_q)}$ since by fixing the variables $x_3,q_{1,1},q_{1,2}$, we are restricting $y_3$ to a subspace of dimension $\l(\frac{n_6}{\log(2n_q)}-1\r)$ over the field $\F_{2^{\log(2n_q)}}$.

It is clear from the arguments that we did not use any specific values of the inputs given to the algorithm $\samp_1$ (including  the value of the bit $b$) to argue about the size of $2\ilaExt^{-1}(q_{2,1},q_{2,2})$. Also note that each of $x_3,y_3,v_{[1, 4Ct]},w_{[1,4]}$ is restricted to some subspace. Since $2\ilaExt^{-1}(q_{2,1},q_{2,2})$ is the cartesian product of these subspaces, it follows that it is a subspace over $\F_2$. Thus the lemma now follows since we can efficiently   sample from a given subspace.
\end{proof}

Using arguments very similar to the above claim, we obtain the following result.
\begin{claim}\label{same_size_h} For any $h \in \{2,\ldots,\ell\}$,  any fixing of  the variables $\{ s_{h,i},r_{h,i},\overline{s}_{h,i},\overline{r}_{h,i}: i \in \{1,2 \}\}$, and any $b \in \{0,1\}$ define the set: 
\begin{align*}
2\ilaExt^{-1}(q_{h+1,1},q_{h+1,2})= \{ (v_{[8C(h-1)t-4Ct+1, 8C(h-1)t+4Ct]},w_{[4h-3,4h]}) \in \{ 0,1\}^{8Ctn_y+4n_x} : \\ 2\ilaExt(v_{[8C(h-1)t+1,8Cht]}, w_{[4h-3,4h]}, q_{1,1},q_{1,2},b) = (q_{h+1,1},q_{h+1,2})\}.
\end{align*}
 There exists an efficient algorithm $\samp_{h+1}$ that takes as input $q_{h+1,1},q_{h+1,2},b,\{ s_{h,i},r_{h,i},\overline{s}_{h,i},\overline{r}_{h,i}: i \in \{1,2 \}\}$, and samples uniformly from $2\ilaExt^{-1}(q_{h+1,1},q_{h+1,2})$. 
 
 Further,  $2\ilaExt^{-1}(q_{h+1,1},q_{h+1,2})$ is a subspace over $\F_2$, and its size   does not depend on the inputs to $\samp_{h+1}$.
\end{claim} \qed

We now show a way of efficiently sampling from the pre-image of the function $\inmExt_1$ (Algorithm $7$).
\begin{claim}\label{almost_invert}For any string $\alpha \in \{ 0,1\}^{\ell}$, and  any fixing of the variables $\{  s_{h,i},r_{h,i},\overline{s}_{h,i},\overline{r}_{h,i}: h \in [\ell], i \in \{1,2 \}\}$ define the set: 
\begin{align*}
\inmExt_1^{-1}(q_{\ell+1,1},q_{\ell+1,2})= \{ (x_2,y_2) \in \{ 0,1\}^{2n_2} :  \inmExt_1(x_2,y_2,\alpha) = (q_{\ell+1,1},q_{\ell+1,2})\}.
\end{align*}
 There exists an efficient algorithm $\samp_{nm_1}$ that takes as input $\{ s_{h,i},r_{h,i},\overline{s}_{h,i},\overline{r}_{h,i}: h \in [\ell], i \in \{1,2 \}\}, \alpha, q_{\ell+1,1},q_{\ell+1,2}$, and samples uniformly from $\inmExt_1^{-1}(q_{\ell+1,1},q_{\ell+1,2})$. 
 
 Further,  $\inmExt_1^{-1}(q_{\ell+1,1},q_{\ell+1,2})$ is a subspace over $\F_2$, and its size   does not depend on the inputs to  $\samp_{nm_1}$.
  \end{claim}

 \begin{proof}     We observe that once we fix all the seeds  $\{  s_{h,i},r_{h,i},\overline{s}_{h,i},\overline{r}_{h,i}: h \in [\ell], i \in \{1,2 \}\}$, for different $h \in  [\ell]$, the blocks  $(v_{[8C(h-1)t-4Ct+1, 8C(h-1)t+4Ct]}$, $w_{[4h-3,4h]})$ can be sampled independently.  Thus, by using  the algorithms $\{ \samp_{h+1}: h \in {\ell} \}$ from  Claim $\ref{same_size}$ and Claim $\ref{same_size_h}$, we sample the variable  $x_3,y_3,w_{[1,4]},v_{[1,4Ct]}, \{  v_{[8C(h-1)t-4Ct+1, 8C(h-1)t+4Ct]},w_{[4h-3,4h]} : h \in [\ell]\}$. 
 
 Finally, since $\Ext(v_{[8C(\ell-1)t+4Ct+1,8C\ell t]},\overline{r}_{\ell})=(q_{\ell+1,1},q_{\ell+1,2})$, it follows by the arguments in Lemma $\ref{same_size}$, that the block $v_{[8C(\ell-1)t+4Ct+1,8C\ell t)]}$ is restricted to a subspace of size $2^{4Ct(n_y-d_5)}$. Thus, we can efficiently sample this block as well.
 
 Further the variable $w_{[4\ell +1,8\ell ]}$ is unused by the algorithm $\inmExt_1$,  and hence takes all values in $\{ 0,1\}^{4\ell n_x}$. Similarly the variable $v_{[8C\ell t+1,16 C \ell t]}$ is unused by the algorithm $\inmExt_1$  and hence takes all values in $\{ 0,1\}^{8Ct\ell}$. Thus, we sample these variables as uniform strings of the appropriate length.
 
 Since $x_2,y_2$ are concatenations of the various blocks sampled above, we can indeed sample efficiently from a distribution uniform on  $\{(x_2,y_2) \in \{ 0,1\}^{2n_2}: \inmExt(x,y,\alpha)=(q_{\ell+1,1},q_{\ell+1,2})\} $. Further since by Claim $\ref{same_size}$ and Claim $\ref{same_size_h}$, the size of the pre-images of each of the blocks generated do not depend on the inputs (and is also a subspace), it follows that $2\inmExt_1^{-1}(q_{\ell+1,1},q_{\ell+1,2})$ is a subspace, and its size does not depend on the inputs to $\samp_{nm_1}$.
 \end{proof}
 
We now proceed to construct an algorithm to  uniformly sample from the pre-image of any output of the function $\inmExt$ (Algorithm $6$), which will yield the required efficient encoder for the resulting one-many non-malleable codes.

\begin{claim}\label{done_sample} For any fixing of the variable $z=x_1 \circ \overline{x}_1 \circ y_1 \circ \overline{y}_1$ and the variables $\{  s_{h,i},r_{h,i},\overline{s}_{h,i},\overline{r}_{h,i}: h \in [\ell], i \in \{1,2 \}\}$, define the set: 
\begin{align*}
\inmExt^{-1}(q_{\ell+1,1},q_{\ell+1,2})= \{ (x,y) \in \{ 0,1\}^{2n} :  \inmExt(x,y) = (q_{\ell+1,1},q_{\ell+1,2})\}.
\end{align*}
 There exists an efficient algorithm $\samp_{nm}$ that takes as input $\{ s_{h,i},r_{h,i},\overline{s}_{h,i},\overline{r}_{h,i}: h \in [\ell], i \in \{1,2 \}\}, z, q_{\ell+1,1},q_{\ell+1,2}$, and samples uniformly from $\inmExt^{-1}(q_{\ell+1,1},q_{\ell+1,2})$. 
 
 Further,  $\inmExt^{-1}(q_{\ell+1,1},q_{\ell+1,2})$ is a subspace over $\F_2$, and its size   does not depend on the inputs to  $\samp_{nm}$.
  \end{claim}
  \begin{proof} We fix the variables $x_1$ and $y_1$. Let $T = \samp(\nu) = \{ t_1,\ldots,t_{n_5}\}$. We now think of $x_2$ as an element in $\F^{n_4}$, $\F=\F_{2^{\log (n+1)}}$.  Let $x_2 = (x_{2,1},\ldots,x_{2,n_4})$, where each $x_{2,i}$ is in $\F$. Recall that the $n_4 \times n$ generator matrix $G$ of the code $\RS$ is the following: 
  $$
 G = 
 \begin{pmatrix}
  1 & 1 & \cdots & 1 \\
  \alpha_1 & \alpha_{2} & \cdots & \alpha_{n} \\
  \vdots  & \vdots  & \ddots & \vdots  \\
  \alpha_1^{n_4-1} & \alpha_2^{n_4-1} & \cdots & \alpha_n^{n_4-1}
 \end{pmatrix}
$$ where $\alpha_1,\ldots,\alpha_{n}$ are distinct non-zero field elements of $\F$.

Let $$
 G_T = 
\begin{pmatrix}
  1 & 1 & \cdots & 1 \\
  \alpha_{t_1} & \alpha_{t_2} & \cdots & \alpha_{t_{n_5}} \\
  \vdots  & \vdots  & \ddots & \vdots  \\
  \alpha_{t_1}^{n_4-1} & \alpha_{t_2}^{n_4-1} & \cdots & \alpha_{t_{n_5}}^{n_4-1}
 \end{pmatrix} 
$$ 
Since $\overline{x}_1 = \RS(x_2)_{\{ T\}}$,  we have the following identity: 
\begin{align} \label{matrix}
\begin{pmatrix}x_{2,1} & \cdots & x_{2,n_4}\end{pmatrix}
 G_T = \overline{x}_1
\end{align}
Thus, for any fixing of $\overline{x}_1$, the variable $x_2$ is restricted to a subspace of dimension $(n_4-n_5)$ over the field $\F$. 

Now, let $j \in [n_4]$ be such that $(x_{2,1},\ldots,x_{2,j})$ is the string $(x_3,w_{[1,4\ell]})$, and $(x_{2,j+1},\ldots,x_{2,n_4})$ is the string $w_{[4\ell+1,8\ell]}$. Clearly, $(n_4-j)\log n = 4\ell n_x $, and thus by our choice of parameters it follows that $j = n_4 - \frac{4\ell n_x}{\log n}= \frac{n_4}{2}+\frac{n_6}{log(n+1)}<\frac{2n_4}{3}<n_4-n_5$.

We further note since any $n_5 \times n_5$ sub-matrix of $G_T$ has full rank (since it is the Vandermonde's matrix), it follows by the rank-nullity thorem that any $j \times n_5$ sub-matrix of $G_T$ has null space of dimension exactly $j-n_5$. Thus for any $\la \in \F^{n_5}$, the equation: 
\begin{align} \label{matrix}
  \begin{matrix}\begin{pmatrix}x_{2,j+1} & \cdots & x_{2,n_4}\end{pmatrix}\\\mbox{}\end{matrix}
  \begin{pmatrix}
  \alpha_{t_1}^{j} & \alpha_{t_2}^{j} & \cdots & \alpha_{t_{n_5}}^{j} \\
  \vdots  & \vdots  & \ddots & \vdots  \\
  \alpha_{t_1}^{n_4-1} & \alpha_{t_2}^{n_4-1} & \cdots & \alpha_{t_{n_5}}^{n_4-1}
 \end{pmatrix} = \overline{x}_1 + \la
\end{align}
has exactly $|\F|^{(j-n_5)}$ solution.

Thus, for any fixing of the variables, $x_{2,1},\ldots,x_{2,j}$, equation $(1)$ has exactly $|\F|^{j-n_5}$ solutions. In other words, for any fixing of  $x_3,w_{[1,4\ell]},\overline{x}_1$, the variable $w_{[4\ell+1,8\ell]}$ is restricted to a subspace, and the size of the subspace does not depend on the fixing of  $x_3,w_{[1,4\ell]},\overline{x}_1$. Using, a similar argument, we can show that for any fixing of $y_3,v_{[1,8Ct\ell]},\overline{y}_1$, the variable $v_{[8Ct\ell+1,16Ct\ell]}$ is restricted to a subspace, and the size of the subspace  does not depend on the fixing of $y_3,v_{[1,8Ct\ell]},\overline{y}_1$.

Now consider any fixing of the variables  $\{  s_{h,i},r_{h,i},\overline{s}_{h,i},\overline{r}_{h,i}: h \in [\ell], i \in \{1,2 \}\}, z$.  As proved in the Claim $\ref{almost_invert}$, we can efficiently sample the variables $x_3,w_{[1,4\ell]}, y_{3}, v_{[1,8Ct\ell]}$.   By the above argument, the variables $v_{[4\ell+1,8\ell]}$ and $w_{[8Ct\ell+1,16Ct\ell]}$ now lie in a subspace, and hence we can efficiently sample these variables as well. Thus we have an efficient procedure $\samp_{nm}$ for uniformly sampling $(x,y)$ from the set $\inmExt^{-1}(q_{\ell+1,1},q_{\ell+1,2})$ .

 It also follows by Claim $\ref{almost_invert}$, that the total size of the pre-image of the variables  $x_3,w_{[1,4\ell]}, y_{3}, v_{[1,8Ct\ell]}$ does not depend on $z$ or the variables $\{  s_{h,i},r_{h,i},\overline{s}_{h,i},\overline{r}_{h,i}: h \in [\ell], i \in \{1,2 \}\}$. Further, for any fixing of  $x_3,w_{[1,4\ell]}, y_{3}, v_{[1,8Ct\ell]},z$, as argued above, the variables $v_{[4\ell+1,8\ell]}$ and $w_{[8Ct\ell+1,16Ct\ell]}$ now lie in a subspace, whose size does not depend on the fixed variables. Thus, overall the size of the total pre-image of $x,y$  does not depend on the inputs to $\samp_{nm}$.
\end{proof}

We now state the main result of this section.
\begin{thm}\label{final_sample} There exists an efficient procedure that given an input $(q_{\ell+1,1},q_{\ell+1,2}) \in \{ 0,1\}^{n_q} \times \{ 0,1\}^{n_q}$, samples uniformly from the set $\{ (x,y): \inmExt(x,y)=(q_{\ell+1,1},q_{\ell+1,2})\}$.
\end{thm}
\begin{proof} We use the following simple strategy.
\begin{enumerate}
\item Uniformly sample the variables $z, \{  s_{h,i},r_{h,i},\overline{s}_{h,i},\overline{r}_{h,i}: h \in [\ell], i \in \{1,2 \}\}$,
\item Use the variables sampled in Step $(1)$ as input to the algorithm $\samp_{nm}$ to sample $(x,y)$. 
\end{enumerate}
The correctness of this procedure follows directly from Claim $\ref{done_sample}$, since it was proved that for any fixing of the variables of Step $1$, the size of pre-image of $\inmExt$ is the same.
\end{proof}

\section*{Acknowledgments}
The first author would like to thank his advisor, David Zuckerman, for his constant guidance and encouragement. 

\bibliographystyle{alpha}
\bibliography{nmc}

\appendix
\end{document}